\documentclass[a4paper,UKenglish]{lipics}
\makeatletter
\newif\if@restonecol
\makeatother

\usepackage[ruled,vlined,linesnumbered,lined]{algorithm2e}

\usepackage{textcomp} 
\usepackage{color}
\usepackage{amsmath,amsfonts,amssymb,wasysym}

\usepackage{tikz}
\usetikzlibrary{arrows,shapes,snakes,automata,backgrounds,positioning, decorations.pathmorphing, patterns,shapes.geometric}
\tikzstyle{locr}=[draw,circle,minimum size=5mm] 
\tikzstyle{min}=[draw,circle,minimum size=5mm] 
\tikzstyle{max}=[draw,rectangle,minimum size=5mm] 
\tikzstyle{triangle}= [fill=purple!50,regular polygon, regular polygon sides=3,minimum size=5pt, inner sep=+0pt, minimum height=0pt]
\tikzstyle{widget}=[draw=gray,rectangle, rounded rectangle=10pt,dashed,minimum size=6mm]
\tikzset{every loop/.style={looseness=7},trans/.style={font=\footnotesize,above=2mm},}

\tikzstyle{locLabel}=[node distance=5mm,color=black]
\tikzstyle{del}=[node distance=5mm,color=blue]
 \tikzstyle{cdel}=[node distance=5mm,color=white]

\newcommand{\set}[1]{\left\{ #1 \right\}}

\newcommand{\seq}[1]{\langle #1 \rangle}

\newcommand{\Nat}{\mathbb N}
\makeatletter
\newcommand{\mypm}{\mathbin{\mathpalette\@mypm\relax}}
\newcommand{\@mypm}[2]{\ooalign{%
  \raisebox{.1\height}{$#1+$}\cr
  \smash{\raisebox{-.6\height}{$#1-$}}\cr}}
\makeatother

\tikzstyle{dashedl}=[dotted]
\tikzstyle{diminish}=[dashed,thick]
\tikzstyle{fun1}=[thick,blue]
\tikzstyle{fun2}=[thick,magenta!50!blue]
\tikzstyle{fun3}=[thick,red]
\tikzstyle{cf}=[thick,purple!50]
\tikzstyle{rep}=[black]
\newcommand{\thesis}[1]{}

\usepackage{styFile}
\usepackage{enumitem}

\title{Revisiting Robustness in Priced Timed Games}

\author[1]{S. Guha} 
\author[2]{S. N. Krishna}
\author[2]{L. Manasa} 
\author[2]{A. Trivedi}
\affil[1]
{
Department of Computer Science \& Engineering, IIT Delhi, India.\\
\texttt{shibashis@cse.iitd.ernet.in}}
\affil[2]{
  Department of Computer Science \& Engineering, IIT Bombay, India.\\
    \texttt{krishnas,manasa,trivedi@cse.iitb.ac.in}} 
\authorrunning{Guha, Krishna, Manasa and Trivedi}
\Copyright{Guha, Krishna, Manasa and Trivedi}

\EventShortName{}

\begin{document}
\maketitle

\begin{abstract}
  Priced timed games are optimal-cost reachability games played between
  two players---the controller and the environment---by moving a token along the
  edges of infinite graphs of configurations of priced timed automata. 
  The goal of the controller is to reach a given set of target locations as
  cheaply as possible, while the goal of the environment is the opposite. 
  Priced timed games are known to be undecidable for timed automata with $3$ or
  more clocks, while they are  known to be decidable for automata  with $1$
  clock. 
  In an attempt to recover decidability for priced timed games  Bouyer,
  Markey, and Sankur studied robust priced timed games  where the environment has the power to
  slightly perturb delays proposed by the controller. 
  Unfortunately, however, they showed that the natural problem of deciding the
  existence of optimal limit-strategy---optimal strategy of the
  controller where the perturbations tend to vanish in the limit---is undecidable
  with $10$ or more clocks.  
  In this paper we revisit this problem and improve our understanding of the
  decidability of these games. 
  We show that the limit-strategy problem is already undecidable for a subclass of robust
  priced timed games with $5$ or more clocks.  
  On a positive side, we show the decidability of the existence of almost optimal
  strategies for the same subclass of one-clock robust priced timed games  by adapting a classical construction by
  Bouyer at al. for one-clock priced timed games. 
\end{abstract}

\newpage
\section{Introduction}
\label{sec:robust-intro}
Two-player zero-sum games on priced timed automata provide a mathematically
elegant modeling framework for the control-program synthesis problem in
real-time systems.  
In these games, two players---the \emph{controller} and
the \emph{environment}---move a token  along the edges of the infinite graph of
configurations of a timed automaton to construct an infinite execution of the
automaton in order to optimize a given performance criterion. 
The optimal strategy  of the controller in such game then corresponds to
control-program with the optimal performance.   
By priced timed games (PTGs) we refer to such games on priced timed automata with
optimal reachability-cost objective. 
The problem of deciding the existence of the optimal controller strategy in PTGs
is undecidable~\cite{BBR05}  with $3$ or more clocks, while  it is known to be
decidable~\cite{BouLar06} for automata with $1$ clock.  
Also, the $\varepsilon$-optimal strategies can be computed for priced timed
games under the non-Zeno assumption~\cite{ABM04,BCFL04}. 
Unfortunately, however, the optimal controller strategies obtained as a result
of solving games on timed automata may not be physically realizable due to
unrealistic assumptions made in the modeling using timed automata, regarding the
capability of the controller in enforcing precise delays.  
This severely limits the application of priced timed games in control-program
synthesis for real-time systems. 

In order to overcome this limitation, Bouyer, Markey, and Sankur~\cite{ocan}
argued the need for considering the existence of robust optimal strategies and
introduced  two different robustness semantics---\emph{excess} and
\emph{conservative}---in priced timed games.
The key assumption in their modeling is that the controller may not be able to
apply an action at the exact time delays suggested by the optimal strategy.
This phenomenon is modeled as a \emph{perturbation} game where the time delay
suggested by the controller can be perturbed by a bounded quantity. 
Notice that such a perturbation may result in the guard of the corresponding
action being disabled. 
In the conservative semantics, it is the controller's responsibility to make
sure that the guards are satisfied after the perturbation. 
On the other hand, in the excess semantics, the controller is supposed to make
sure that the guard is satisfied before the perturbation: an action can be
executed even when its guard is disabled (``excess'') post perturbation and the
valuations post perturbation will be reflected in the next state. 
The game based characterization for robustness in timed
automata under ``excess'' semantics was first proposed by Bouyer, Markey, and
Sankur~\cite{BMS12} where they study the parameterized robust (qualitative)
reachability problem and show it to be EXPTIME-complete.
The ``conservative'' semantics were studied for reachability and B\"uchi objectives
in~\cite{BMS13a} and shown to be PSPACE-complete.
For a detailed survey on robustness in timed setting we refer to an excellent
survey by Markey~\cite{Mar11}. 

Bouyer, Markey, and Sankur~\cite{ocan} showed that the problem for
deciding the existence of the optimal strategy is undecidable for priced timed
games with $10$ or more clocks under the excess semantics. 
In this paper we further improve the understanding of the decidability of these
games. 
However, to keep the presentation simple, we restrict our attention to
turn-based games under excess semantics.  
To further generalize the setting, we permit both positive and negative price
rates with the restriction that the accumulated cost in any cycle is
non-negative (akin to the standard no-negative-cycle restriction in shortest path game
problems on finite graphs).  
We improve the undecidability result of~\cite{ocan} by proving that optimal
reachability remains undecidable for robust priced timed automata with 5 clocks.
Our second key result is that, for a fixed $\delta$,  the cost optimal
reachability problem for one clock priced timed games with no-negative-cycle
restriction is decidable for robust priced timed games with given bound on
perturbations. To the best of our knowledge, this is the first decidability result 
known for robust timed games under the excess semantics. A 
closely related result is~\cite{prabhu},    
where decidability is shown for robust timed games under the conservative semantics 
 for a fixed $\delta$.


\section{Preliminaries}
\label{sec:prelims}
We write $\Real$ for the set of reals and $\Int$ for the set of integers.
Let $\pclocks$ be a finite set of real-valued variables called \emph{clocks}. 
A \emph{valuation} on $\pclocks$ is a function $\nu : \pclocks \to \Real$.
We assume an arbitrary but fixed ordering on the clocks and write $x_i$
for the clock with order $i$. 
This allows us to treat a valuation $\nu$ as a point $(\nu(x_1), \nu(x_2),
\ldots, \nu(x_n)) \in \Real^{|\pclocks|}$. 
Abusing notations slightly, we use a valuation on $\pclocks$ and a point in
$\Real^{|\pclocks|}$ interchangeably. 
For a subset of clocks $X \subseteq \pclocks$ and valuation $\nu \in \V$, we
write $\nu[X{:=}0]$ for the valuation where $\nu[X{:=}0](x) = 0$ if $x \in X$, and
$\nu[X{:=}0](x) = \nu(x)$ otherwise.
The valuation $\zero \in \V$ is a special valuation such that
$\zero(x) = 0$ for all $x \in \pclocks$.
A clock constraint over $\pclocks$ is a subset of $\Real^{|\pclocks|}$.
We say that a constraint is \emph{rectangular} if it is a conjunction of a finite set of constraints of the form 
$x  \bowtie k,$  where $k \in \Int$, $x \in \pclocks$, and 
$\bowtie \in \{<,\leq, =, >, \geq\}$.    
 For a constraint $g \in \rect(\pclocks)$, we write $\sem{g}$ for the set of valuations in
 $\Real^{|\pclocks|}$ satisfying $g$.  
 We write $\rect(\pclocks)$ for the set of rectangular constraints over
 $\pclocks$. We use the terms constraints and guards interchangeably. 
 
 Following~\cite{BouLar06} we introduce priced timed games with external cost
 function on target locations (see Appendix \ref{app:cost-eg}). 
 For this purpose, we define a \textit{cost function}\cite{BouLar06} as a
 piecewise affine continuous function $f: \Rplus^n \to \R \cup \set{+\infty,-\infty}$. 
 We write $\costfun$ for the set of all cost functions. 

\begin{definition}[Priced Timed Games]
  A turn-based two player \emph{priced timed game} is a tuple
  $\ptg = (\locations_1,\locations_2, L_{init}, \pclocks, \trans, \prices,\rgoals,f_{goal})$ where 
   $\locations_i$ is a finite set of \emph{locations} of Player~$i$, 
     $L_{init} \subseteq L_1 \cup L_2$(let $L_1\cup L_2=L$) is a set of initial locations,
  $\pclocks$ is an (ordered) set of \emph{clocks}, 
  $\trans \subseteq \locations \times {\rect}(\pclocks)  \times 2^{\pclocks}
  \times (\locations \cup \rgoals)$ 
  is the \emph{transition relation}, 
  $\prices \colon \locations \rightarrow \Z$ is the price
  function,
  $\rgoals$ is the set of target locations, $\rgoals \cap \locations =\emptyset$; 
      and
 $f_{goal} : \rgoals \to \costfun$ assigns external cost functions to target locations. 
  \end{definition}
We refer to Player~1 as the controller and Player~2 as the environment. 
A priced timed game begins with a token placed on some initial location $\ell$ with  
valuation $\zero$ and cost accumulated being so far being $0$.  
At each round, the player who controls the current location $\ell$ chooses a delay
$t$ (to be elapsed in  $l$) and an outgoing transition $e =(\ell, g, r, \ell') \in
\trans$ to be taken after $t$ delay at $\ell$. 
The clock valuation is then updated according to the delay $t$, the
reset $r$,  the cost is incremented by $\prices(\ell) \cdot t$ and the token is moved 
to the location $\ell'$.   
The two players continue moving the token in this fashion, and give rise to a
sequence of locations and transitions called a \emph{play} of the game. 
A configuration or state of a PTG is a tuple $(\ell, \val, \cost)$ where $\ell \in
\locations$ is a location,  $\val \in \V$ is a valuation, and  $\cost$ is the
cost accumulated from the start of the play.
We assume, w.l.o.g~\cite{BFHLPRV01}, that the clock valuations are bounded. 
\begin{definition}[PTG semantics]
The semantics of a PTG $\ptg$ is a labelled state-transition game arena 
$\sem{\ptg}$ = $(\nSR = \states_1 \uplus \states_2, S_{init}, A, \nTR, \nP, \nC)$
where 
\begin{itemize}[topsep=0pt,itemsep=-1ex,partopsep=1ex,parsep=1ex]
\item 
  $\states_j = \locations_j \times \V $ are the Player~$j$ states with 
  $\nSR = \states_1 \uplus \states_2$, 
\item 
  $S_{init} \subseteq \nSR$ are initial states s.t. $(\ell, \nu)
  \in S_{init}$ if $\ell \in L_{init}$, $\nu = \zero$,
\item 
  $A = \Rplus \times \trans$ is
  the set of \emph{timed moves}, 
\item
  $\nTR: (\nSR \times A) \to \nSR$ is the transition function s.t. 
  for $s = (\ell, \nu), s' = (\ell', \nu') {\in} \nSR$ and
  $ \tau = (t, e) \in A$ the function $\nTR(s, \tau)$ is defined if 
  $ e = (\ell, g, r, \ell')$ is a transition of the PTG and $\nu \in \sem{g}$;
  moreover $\nTR(s, \tau) = s'$ if $\nu' = (\nu+t)[r{:=}0]$ (we write $s
  \xrightarrow{\tau} s'$ when $\nTR(s, \tau) = s'$);  
\item
  $\nP: \nSR \times A \to \Real$ is the price function such that $\nP((\ell,
  \nu), (t, e)) = \prices(\ell)\cdot t$; and 
\item
  $\nC: \nSR \to \Real$ is an external cost function such that $\nC(\ell, \nu)$  is
  defined when $\ell \in T$  such that $\nC(\ell, \nu) = 
  f_{goal}(\ell)(\nu)$. 
\end{itemize}\end{definition}

A \emph{play} $\rho = \seq{s_0, \tau_1, s_1, \tau_2, \ldots, s_n}$ is a
finite sequence of states and actions s.t.  $s_0 \in S_{init}$ and $s_i
\xrightarrow{\tau_{i+1}} s_{i+1}$ for all $0 \leq i < n$.
The infinite plays are defined in an analogous manner. 
For a finite play $\rho$ we write its last state as $\last(\rho) = s_n$. 
For a (infinite or finite) play $\rho$ we write $\istop(\rho)$ for the index of
first target state and if it doesn't visit a target state then $\istop(\rho) = \infty$.
We denote the set of plays as $\runs_{\ptg}$. 
For a play $\rho = \seq{s_0, (t_1, a_1), s_1, (t_2, a_2), \ldots}$ if
$\istop(\rho) = n < \infty$ 
then $\costp_{\ptg}(\rho) = \nC(s_n) + \sum_{j{=}1}^{n} \nP(s_{i-1}, (t_i,
a_i))$ else $\costp_{\ptg}(\rho) = +\infty$. 

A \emph{strategy} of player $j$ in $\ptg$ is a function $\strat : \runs_{\ptg}
\to A$ such that for a play $\rho$ the function $\strat(\rho)$  is defined if
$\last(\rho) \in \states_j$.  
We say that a strategy $\strat$ is memoryless if $\strat(\rho) = \strat(\rho')$ when 
$\last(\rho) = \last(\rho')$, otherwise we call it memoryful. 
We write $\mminstrats$ and $\mmaxstrats$ for the set of strategies of player $1$ and
$2$, respectively. 

A play $\rho$ is said to be \emph{compatible to a strategy} $\strat$ of player
$j \in \set{1, 2}$ if for every state $s_i$ in $\rho$ that belongs to Player~$j$,   
$s_{i+1} = \strat(s_i)$. 
Given a pair of strategies $(\sigma_1, \sigma_2) \in \mminstrats \times
\mmaxstrats$, and a state $s$,  the outcome of $(\sigma_1, \sigma_2)$ from
$s$ denoted  $\outcome(s, \strat_1, \strat_2)$ is the unique play that starts at $s$ and is
compatible with both strategies.  
Given a player $1$ strategy $\sigma_1 \in \mminstrats$ we define its cost
$\costp_{\ptg}(s,\strat_1)$  as $\sup_{\strat_2 \in
  \mmaxstrats}(\costp(\outcome(s,\strat_1,\strat_2)))$.
We now define the  \emph{optimal reachability-cost} for Player~1 from a state
$s$ as 
\[
\optcost{\ptg}(s) = \inf_{\strat_1 \in \mminstrats} \sup_{\strat_2 \in
  \mmaxstrats} (\costp(\outcome(s,\strat_1,\strat_2))).
\]   
A strategy $\strat_1 \in \mminstrats$ is said to be optimal from $s$ if 
$\costp_{\ptg}(s,\strat_1) = \optcost{\ptg}(s)$. 
Since the optimal strategies may not always exist~\cite{BouLar06} we define
$\epsilon$ optimal strategies.  
For $\epsilon  > 0$ a strategy $\sigma_\epsilon \in \mminstrats$
is called $\epsilon$-optimal if  
$\optcost{\ptg}(s)\leq \costp_{\ptg}(s,\strat_\epsilon) < \optcost{\ptg}(s)+\epsilon$.  
Given a PTG $\ptg$ and a bound $K \in \Z$, the \emph{cost-optimal} reachability
problem for PTGs is to  decide whether there exists a strategy for player~1 such
that $\optcost{\ptg}(s) \leq K$ from some starting state $s$. 
\begin{theorem}[\cite{BBM06}]
  Cost-optimal reachability problem is undecidable for PTGs with $3$ clocks. 
\end{theorem}
\begin{theorem}[\cite{BouLar06,HIM13,Rut}]
\label{thm:1ptg_pat}
The $\epsilon$-optimal strategy is computable for $1$ clock PTGs.
\end{theorem}

\section{Robust Semantics}
\label{sec:robust-prelims}
Under the robust semantics of priced timed games the environment player---also
called as the perturbator---is more privileged as it has the power to perturb any
delay chosen by the controller by an amount  in $[-\delta, \delta]$, where
$\delta>0$ is a pre-defined bounded  quantity. 
However, in order to ensure time-divergence there is a restriction that the time
delay at all locations of the RPTG must be $\geq \delta$.          
There are the following two perturbation semantics as defined in~\cite{ocan}.
\begin{itemize}
\item   {\it{Excess semantics}}. At any controller location, the time delay     
  $t$ chosen by the controller is altered to some $t' \in [t-\delta,t+\delta]$ by the perturbator. 
  However, the constraints on the outgoing transitions of the controller
  locations are evaluated with respect to the time elapse $t$ chosen by the
  controller.  If the constraint is satisfied with respect to $t$, then the
  values of all variables which are not reset on the transition are updated with
  respect to $t'$; the variables which are reset  obtain value 0. 
\item {\it{Conservative semantics}}. In this, the constraints on the outgoing
  transitions are evaluated with respect to $t'$. 
\end{itemize}
In both semantics, the delays chosen by perturbator at his locations are not
altered, and the constraints on outgoing transitions are evaluated  in the
usual way, as in PTG.  
   
A Robust-Priced Timed Automata (RPTA)
is an RPTG which has only controller locations. 
At all these locations, for any time delay $t$ 
chosen by controller, 
perturbator can implicitely perturb $t$ by a quantity in $[-\delta,\delta]$.  
The excess  as well as the conservative perturbation semantics for RPTA are defined 
in the same way as in the RPTG. 
Note that our RPTA coincides with that of \cite{ocan} when 
the cost functions at all target locations are of the form 
$cf: \Rplus^n \rightarrow \{0\}$. 
  Our RPTG are turn-based, and have cost funtions at the targets, while  RPTGs studied in \cite{ocan} are concurrent.


\begin{definition}[Excess Perturbation Semantics]
Let $\rtg = (\locations_1,\locations_2, L_{init}, \pclocks, \trans, \prices,\rgoals,f_{goal})$
be a RPTG. 
Given a $\delta>0$, the excess perturbation semantics of  RPTG $\rtg$ is a LTS 
$\sem{\rtg}$ = $(\nSR,\nAR,\nTR)$
where $\nSR = \states_1 \cup \states_2 \cup (\rgoals \times \Rplus)$, 
$\nAR = \actions_1 \cup \actions_2$ and $\nTR =\nTR_1 \cup \nTR_2$. 
We define the set of states, actions and transitions for each player below. 
\begin{itemize}[topsep=0pt,itemsep=-1ex,partopsep=1ex,parsep=1ex]
\item $\states_1 = \locations_1 \times \V $ are the controller states,
\item $\states_2 = (\locations_2 \times \V ) \cup (\states_1 \times \Rplus \times \trans)$ are the perturbator states. 
The first kind of states are encountered at perturbator locations. The second kind
of states are encountered when controller chooses a delay $t \in \Rplus$ and a transition $e \in X$ at a controller location. 
 \item $\actions_1 = \Rplus \times \trans$ are controller actions
\item $\actions_2 = (\Rplus \times \trans) \cup [-\delta,\delta]$ are perturbator actions. The first 
kind of actions $(\Rplus \times \trans)$ are chosen at states of the form  
$\locations_2 \times \V \in \states_2$, while the second kind of actions are chosen 
at states of the form $\states_1 \times \Rplus \times \trans \in \states_2$,
\item $\nTR_1 = (\states_1 \times \actions_1 \times \states_2)$ is the set of 
controller transitions such that for a controller state $(l,\val)$ and a controller action $(t,e)$, 
$\nTR_1((l,\val),(t,e))$ is defined iff there is a transition $e=(l,g,a,r,l')$ in $\rtg$ such that 
$\val + t \in \sem{g}$. 
\item $\nTR_2 = \states_2 \times \actions_2 \times (\states_1 \cup \states_2 \cup (\rgoals \times \Rplus))$ is the set of perturbator transitions
such that 
\begin{itemize}[topsep=0pt,itemsep=-1ex,partopsep=1ex,parsep=1ex]
\item 
For a perturbator state of the type $(l,\val)$ and a perturbator action $(t,e)$, we have 
$(l',\val')  = \nTR_2((l,\val),(t,e))$ iff there is a transition $e=(l,g,a,r,l')$ in $\rtg$ such that 
$\val + t \in \sem{g}$, $\val' = (\val + t)[r:=0]$,
\item For a perturbator state of type $((l,\val),t,e)$ and a perturbator action $\varepsilon \in [-\delta,\delta]$, we have $(l',\val') = \nTR_2(((l,\val),t,e),\varepsilon)$ iff 
 $e = (l,g,a,r,l')$,
 and $\val' = (\val + t + \varepsilon)[r:=0]$. 
\end{itemize}
\end{itemize}
\end{definition}
We now define the cost of the transitions, denoted as $\costp(t,e)$ as follows : 
\begin{itemize}[topsep=0pt,itemsep=-1ex,partopsep=1ex,parsep=1ex]
\item For controller transitions : $(l,\val) \xrightarrow{(t,e)} ((l,\val),t,e)$ : the cost accumulated is $\costp(t,e)=0$. 
 \item For perturbator transitions : 
 \begin{itemize}[topsep=0pt,itemsep=-1ex,partopsep=1ex,parsep=1ex]
  \item From perturbator states of type $(l,\val)$ :  $(l,\val) \xrightarrow{t,e} (l',\val')$, the cost accumulated is  $\costp(t,e)=t * \prices(l)$. 
  \item From perturbator states of type $((l,\val),t,e)$ : $((l,\val),t,e) \xrightarrow{\varepsilon} (l',\val')$, 
  the cost accumulated is $(t + \varepsilon) * \prices(l)$. Note that although this transition has no edge choice involved and 
  the perturbation delay chosen is $\varepsilon\in[-\delta,\delta]$, the controller action $(t,e)$ chosen 
  in the state $(l,\val)$ comes into effect in this transition. Hence for the sake of uniformity, we denote the 
  cost accumulated in this transition to be $\costp(t+\varepsilon,e) =  (t + \varepsilon) * \prices(l)$.  
 \end{itemize}
\end{itemize}
Note that  we check  satisfiability 
of the constraint $g$ before the perturbation; however, the  reset occurs after the perturbation. 
The notions of a path and a winning play are  the same as in PTG. We shall now 
adapt the definitions of cost of a play, and a strategy for the excess perturbation semantics. 
Let $\rho$ $= \seq{s_1, (t_1, e_1),  s_2, (t_2, e_2), \cdots (t_{n-1}, e_{n-1}),
s_n}$ be a path in the LTS $\sem{\rtg}$. 
Given a $\delta>0$, for a finite play $\rho$ ending in target location, we define $\costp_{\rtg}^{\delta}(\rho) = \sum_{i=1}^{n} \costp(t_i,e_i) + f_{goal}(l_n)(\val_n)$ as the sum of the costs of all transitions as defined above along with the value from the cost function of the target location $l_n$.  Also, we re-define the cost of a strategy $\strat_1$ from 
a state $s$ for a given $\delta>0$ as $\costp_{\rtg}^{\delta}(s,\strat_1) 
 = \sup_{\strat_2 \in \maxstrats{\rtg}} \costp_{\rtg}^{\delta}(\outcome(s,\strat_1,\strat_2))$. Similarly, $\optcostd{\rtg}$ is the optimal cost 
 under excess perturbation semantics for a given $\delta>0$ defined as 
 \[ 
 \optcostd{\rtg}(s) = 
 \inf_{\strat_1 \in \minstrats{\rtg}} \sup_{\strat_2 \in \maxstrats{\rtg}}
 (\costp^{\delta}_{\rtg}(\outcome(s,\strat_1,\strat_2))).
\]
Since  optimal strategies may not always exist,  we define 
$\epsilon-$optimal strategies such that 
for every $\epsilon>0$, 
$\optcostd{\rtg}(s) \leq \costp_{\rtg}^{\delta}(s,\strat_1) < \optcostd{\rtg}(s)+\epsilon$.
Given a $\delta$ and a RPTG $\rtg$ with a single clock $x$, a strategy $\sigma_1$ is called
\emph{$(\epsilon,N)-$acceptable} \cite{BouLar06} for $\epsilon>0, N \in \Nat$ when 
(1)it is memoryless, (2)it is $\epsilon-$optimal and (3)there exist $N$ consecutive intervals 
$(I_i)_{1{\leq}i{\leq}N}$ partitioning $[0,1]$ such that for every location $l$, 
for every $1{\leq}i{\leq}N$ and every integer $\alpha < M$ (where $M$ is the
maximum bound on the clock value), the function that maps the clock values
$\nu(x)$ to the cost of the strategy $\sigma_1$ at every state $(l,\nu(x))$,
($\nu(x)\mapsto \costp_{\rtg}^{\delta}((l,\nu(x)),\strat_1)$) is affine for
every interval $\alpha + I_i$. Also, the strategy $\sigma_1$ is constant over
the values $\alpha+I_i$ at all locations, that is, when $\nu(x) \in \alpha+I_i$,
the strategy $\strat_1(l,\nu(x))$ is constant. 
The number $N$ is an important attribute of the strategy as it establishes that
the strategy  does not fluctuate infinitely often and is implementable.   

Now, we shall define limit variations of costs, strategies and values as $\delta \rightarrow 0$. 
The \emph{limit-cost} of a controller strategy $\sigma_1$ from  state $s$ 
is defined over all plays $\rho$ starting from $s$ that are compatible with
$\sigma_1$ as:
\[ 
  {\limcost}_{\rtg}(s,\strat_1) 
  = \lim_{\delta \to 0} \sup_{\strat_2 \in \maxstrats{\rtg}}
  \costp_{\rtg}^{\delta}(\outcome(s,\strat_1,\strat_2)).
  \] 
  The \emph{limit strategy upper-bound problem~}\cite{ocan} for excess perturbation semantics asks, given a RPTG $\rtg$, state $s=(l,\zero)$ with  cost 0 and a rational number $K$, whether there exists 
 a strategy $\strat_1$ such that $\limcost_{\rtg}(s,\strat_1) \leq K$. 
The following are the main results of \cite{ocan}. 
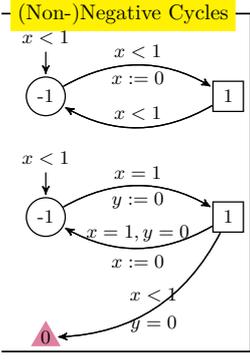
 \begin{wrapfigure}[10]{o}{3.2cm}
   \scalebox{.8}{
    \begin{tikzpicture}[overlay]
     \draw[thick]  (0.1,3)--(4.2,3) --(4.2,-2.6)--(0.1,-2.6) ;
     
     \node[fill=yellow] (label) at (2.1,3) {(Non-)Negative Cycles};
    \end{tikzpicture}
    \begin{tabular}{c}
 
 \begin{tikzpicture}[->,>=stealth',shorten >=1pt,auto,node distance=1cm, semithick]
\tikzstyle{every state}=[minimum size=1em]
   \node[state,initial, initial where=above, initial text=$x<1$] at (0,12) (A) {-1} ;
  \node[max] at (3,12) (B) {1} ;
  \node[state,initial,initial where=above, initial text=$x<1$] at (0,10) (C) {-1} ;
  \node[max] at (3,10) (D) {1} ;
  \node[triangle] at (0,8) (E) {0} ;
        \path(A) edge[bend left]  node[midway,above] {$x<1$} (B);
        \path(A) edge[bend left]  node[midway,below] {$x:=0$} (B);
\path(B) edge[bend left]  node [midway,above] {$x<1$}(A);
\path(B) edge[bend left]  node [midway,above] {}(A);

\path(C) edge[bend left]  node[midway,above] {$x=1$} (D);
       \path(C) edge[bend left]  node[midway,below] {$y:=0$} (D);
\path(D) edge[bend left]  node [midway,above] {$x<1$}(E);
\path(D) edge[bend left]  node [midway,below] {$y=0$}(E);
\path(D) edge[bend left]  node [midway,above] {$x=1,y=0$}(C);
\path(D) edge[bend left]  node [midway,below] {$x:=0$}(C);

\end{tikzpicture}
 
   \end{tabular}

}
    
  \end{wrapfigure}
\begin{theorem}[Known results \cite{ocan}]
 \label{thm:robust_known}
\begin{enumerate}
 \item The limit-strategy upper-bound problem is undecidable for 
 RPTA and RPTG under excess perturbation semantics, for $\geq 10$  clocks.
    \item  For a fixed $\delta \in[0, \frac{1}{3}]$, 
 and a given RPTA $\rta$, a target location $l$ and a rational $K$, 
 it is undecidable whether  $\inf_{\strat_1} \sup_{\strat_2}cost_{\strat_1,\strat_2}(\rho) < K$ such that $\rho$ ends in $l$.
  $cost_{\strat_1,\strat_2}(\rho)$ is the cost of the unique run $\rho$ obtained 
  from the pair of strategies $(\strat_1, \strat_2)$.
\end{enumerate}
 \end{theorem}
 We consider a semantic subclass of RPTGs in which the accumulated cost
  of any cycle is non-negative: that is,  
  any iteration of a cycle will always have a non-negative cost. 
  Consider the two cycles depicted.  
The one on top  has a non-negative cost, while 
the one below always has a negative cost. 
In the cycle below, the perturbator will not perturb, since that will lead to a
target state. 
In the rest of the paper, we consider this semantic class of RPTGs (RPTAs), and prove decidability and undecidability 
results; however, we will refer to them as RPTGs(RPTAs). 
Our key contributions are the following theorems. 
   \begin{theorem}
  \label{thm:robust_undec}
  The limit-strategy upper-bound problem is undecidable for RPTA with 5 clocks, location prices in $\set{0,1}$, and 
  cost functions $cf: \Rplus^n \rightarrow \{0\}$ at all target locations. 
 \end{theorem}
 
 \begin{theorem}
 \label{thm:robust_dec}
  Given a 1-clock RPTG $\rtg$
  and a $\delta>0$, we can compute $\optcostd{\rtg}(s)$ for every state $s=(l,\val)$. 
   For every $\epsilon >0$, there exists an $N \in \Nat$ such that the controller 
 has an $(\epsilon,N)$-acceptable strategy.  
 \end{theorem}
The rest of the paper is devoted to the proof sketches of these two theorems,
 while we give detailed proofs in the appendix. 

\section{Undecidability with 5 clocks}
 \label{sec:robust-undec}
 In this section, we improve the result of \cite{ocan} by showing that the
limit strategy upper bound problem is undecidable for robust priced timed automata with 5 or more clocks.
The undecidability result is obtained using a  reduction to the halting problem of two-counter machines.

A two-counter machine has counters $C_1$ and $C_2$,
and a list of instructions $I_1, I_2, \dots, I_n$, where 
$I_n$ is the \emph{halt instruction}. 
For each $1 \le i \le n-1$,
$I_i$ is one of the following instructions:
\texttt{increment $c_b$: $c_b := c_b + 1; \;goto\; I_j$},
for $b=1 \text{ or } 2$,
\texttt{decrement $c_b$ with zero test: $if\; (c_b = 0)\; goto\; I_j \;
else\; c_b := c_b - 1;\; goto\; I_j$}, where $c_1, c_2$ represent the counter values. 
 The initial values of both counters are 0. 
Given the initial configuration $(I_1, 0, 0)$ 
the halting problem for two counter machines is to find if 
the configuration $(I_n, c_1, c_2)$ is reachable, with $c_1, c_2 \geq 0$.
This problem is known to be undecidable. 

We simulate the two counter machine using a RPTA with 
5 clocks $x_1,z,x_2,y_1$ and $y_2$
under the excess perturbation semantics.
The counters are encoded in clocks 
$x_1$ and $z$ as 
$x_1 = \frac{1}{2^i} + \varepsilon_1$
and $z = \frac{1}{2^j} + \varepsilon_2$
where $i,j$  are respectively  the values of counters $C_1,C_2$, and 
$\varepsilon_1$ and $\varepsilon_2$ denote 
accumulated values due to possible perturbations.  
 Clocks $x_2$, $y_1$ and $y_2$ help with the rough work.
The simulation is achieved as follows:
for each instruction, we have a module simulating it.
Upon entering the module, the clocks are in their normal form 
i.e. $x_1 = \frac{1}{2^i}+ \varepsilon_1,
z = \frac{1}{2^j} + \varepsilon_2$ and $x_2=0$ and $y_1=y_2=0$.

\subsection{Increment module} \label{sec:incre}
The module in Figure \ref{fig_increment}
simulates the increment of counter $C_1$.
The value of counter $C_2$ remains unchanged
since the value of clock $z$ remains unchanged at the
exit from the module.
Upon entering $A$ the clock values are $x_1= \frac{1}{2^i} + \varepsilon_1,
z= \frac{1}{2^j} + \varepsilon_2, x_2=y_1=y_2=0$.
Here $\varepsilon_1$ and $\varepsilon_2$ respectively denote the
perturbations accumulated so far.
We denote by $\alpha$, the value of clock $x_1$,
i.e. $\frac{1}{2^i} + \varepsilon_1$.
Thus at $A$, the delay is $1 - \alpha$.
Note that the dashed edges are unperturbed (this is a short hand notation. A
small gadget that  implements this is described in Appendix \ref{app:undec}), so 
$x_1=1$ on entering $B$.
  No time elapse happens at $B$, and 
at $C$, controller chooses a delay $t$. This $t$ must be 
$\frac{\alpha}{2}$ to simulate the increment correctly. 
$t$ can be perturbed by an amount $\delta$ by the perurbator, 
where $\delta$ can be both positive or negative, obtaining
$x_2=t+\delta, x_1=0,y_1=1-\alpha+t+\delta$ on entering $D$.
At $D$, the delay is $\alpha-t - \delta$.
Thus the total delay from the entry point $A$ in this module to the mChoice
module is 1 time unit.
At the entry of the $mChoice$ ($mChoice$ and Restore modules are in
Appendix \ref{app:undec}) module, the clock values are
$x_1=\alpha-t-\delta,z=1+\frac{1}{2^j} + \varepsilon_2, 
x_2=\alpha, y_1=1,y_2=0$.
To correctly simulate the increment of $C_1$,
 $t$ should be exactly $\frac{\alpha}{2}$.

At the mChoice module, perturbator can either 
continue the simulation (by going through the Restore module) or
 verify the correctness of controller's delay (check $t=\frac{\alpha}{2}$).
  The mChoice module 
adds 3 units to the values of 
$x_1, x_2$ and $z$, and resets $y_1, y_2$.
Due to the mChoice module, the clock values are
$x_1=3+\alpha-t - \delta,z=4+\frac{1}{2^j} + \varepsilon_2,
x_2=3+\alpha, y_1=1,y_2=0$.
If perturbator chooses to continue the simulation,
then Restore module brings all the clocks back to normal form. 
Hence upon entering $F$,  the clock values are
$x_1=\alpha-t -\delta, z=\frac{1}{2^j} + \varepsilon_2, x_2=y_1=1,y_2=0$.
This value of $x_1$ is $\frac{\alpha}{2}+\varepsilon_1$, since 
$t=\frac{\alpha}{2}$ and $\varepsilon_1=-\delta$, the perturbation effect. 

Let us now see how perturbator verifies $t=\frac{\alpha}{2}$ by entering
the Choice module. 
The Choice  module also adds 3 units to the values of 
$x_1, x_2$ and $z$, and resets $y_1, y_2$. 
The module $Test~Inc^{C_1}_{>}$
is invoked to check if 
$t > \frac{\alpha}{2}$, and 
the module $Test~Inc^{C_1}_{<}$ is invoked to check 
if $t < \frac{\alpha}{2}$.
Note that using the  mChoice module and the Choice module one after the other, 
the clock values upon entering $Test~Inc^{C_1}_{>}$ or
$Test~Inc^{C_1}_{<}$ are $x_1=6+\alpha-t - \delta,
z=7+\frac{1}{2^j} + \varepsilon_2,x_2=6+\alpha, y_1=0,y_2=0$.
\begin{figure}[t]
\begin{center}
\begin{tikzpicture}[->,>=stealth',shorten >=1pt,auto,node distance=1cm,  semithick,scale=0.9]
    

  \node at (-4.5,0) (i) {};
  
  \node[locr] at (-3,0) (l1) {$0$} ;
  \node[locLabel] () [above of =l1] {$A$};

  \node[locr] at (-1.5,0) (l2) {$0$} ;
  \node[locLabel] () [above of =l2] {$B$};

  \node[locr] at (0,0) (l3) {$0$} ;
  \node[locLabel] () [above of =l3] {$C$};

  \node[locr] at (1.5,0) (l4) {$0$} ;
  \node[locLabel] () [above of =l4] {$D$};
  
  \node[widget] at (3.5,0) (c) {$mChoice$};

  \node[locr] at (5.75,0) (l5) {$0$} ;
  \node[locLabel] () [above of =l5] {$E$};

  \node[widget] at (8,1) (r) {$Restore_{Inc}^{C_1C_2}$};
  \node[widget] at (8,-1) (s) {$Restore_{Inc}^{C_2C_1}$};
  
  \node[locr] at (10.3,0) (l6) {0};
  \node[locLabel] () [above of = l6] {$F$};

  \draw[trans,dashed] (i) -- (l1)node[midway,above]{$x_2{=}0$}node[midway,below]{$\set{y_2}$};
  \draw[trans,dashed] (l1) -- (l2)node[midway,above]{$x_1{=}1$}node[midway,below]{$\set{x_2}$};
  \draw[trans,dashed] (l2) -- (l3)node[midway,above]{$x_2{=}0$}node[midway,below]{$\set{x_1}$};  
  \draw[trans] (l3) -- (l4)node[midway,above]{$x_1 {\le} 1$}node[midway,below]{$\set{x_1}$};  
  \draw[trans,dashed] (l4) -- (c)node[midway,above]{$y_1{=}1$} node[midway,below]{$\set{y_2}$};
  \draw[trans,dashed] (c) -- (l5)node[midway,above]{$y_1{=}0$} node[midway,below]{$\set{x_2,y_2}$};
  \draw[trans,dashed] (l5) -- (r)node[midway,sloped,above]{$y_1{=}0$};
  \draw[trans,dashed] (l5) -- (s)node[midway,sloped,above]{$y_1{=}0$};
  \draw[trans,dashed] (r) -- (l6)node[midway,sloped,above]{$y_1{=}0$};
  \draw[trans,dashed] (s) -- (l6)node[midway,sloped,above]{$y_1{=}0$};
  
   \node[widget] at (3.5,-2) (c2) {$Choice$};

   \node[widget,fill=yellow!20] at (6.3,-2) (F) {$Test~Inc^{C_1}_{<}$};


  \draw[trans,dashed] (c) -- (c2)node[midway,left]{$y_1{=}0$};
  \draw[trans,dashed] (c2) -- (F)node[midway,sloped,above]{$y_1{=}0$};
  

\draw[dashed,fill=yellow!20,draw=gray,rounded corners=10pt] (-4.7,-1.2) rectangle (2,-2.5);
   \node at (-3.5,-2.75) (j) {$\mathbf{Test~Inc^{C_1}_{>}}$};

  
  \node[locr] at (1.5,-2) (L) {1} ;
  \node[locLabel] () [above of =L] {$A'$};
   
   \node[locr] at (0,-2) (M) {0};
   \node[locLabel] () [above of = M] {$B'$};

   \node[locr] at (-1.5,-2) (N) {1};
   \node[locLabel] () [above of = N] {$C'$};
  
   \node[locr] at (-3,-2) (O) {1};
   \node[locLabel] () [above of = O] {$D'$};

  \node[triangle] at (-4.5,-2)(T){0};
  
  \draw[trans,dashed] (c2) -- (L)node[midway,above]{$y_1{=}0$};
  \draw[trans,dashed] (L) -- (M)node[midway,above]{$x_1{=}7$};
  \draw[trans,dashed] (L) -- (M)node[midway,below]{$\{x_1\}$};
  \draw[trans,dashed] (M) -- (N)node[midway,above]{$x_2{=}8$};
  \draw[trans,dashed] (N) -- (O)node[midway,above]{$x_1{=}1$};
\draw[trans,dashed] (N) -- (O)node[midway,below]{$\{x_1\}$};
\draw[trans,dashed] (O) -- (T)node[midway,above]{$x_1{=}1$};

  
\end{tikzpicture}

\caption{$\mathbf{Increment~C_1~module}$ : 
The module keeps the fractional part of the
clock $z$ unchanged. The dashed edges represent 
unperturbed edges (detailed in Appendix \ref{app:undec}).
} 
\label{fig_increment}
\end{center}
\end{figure}
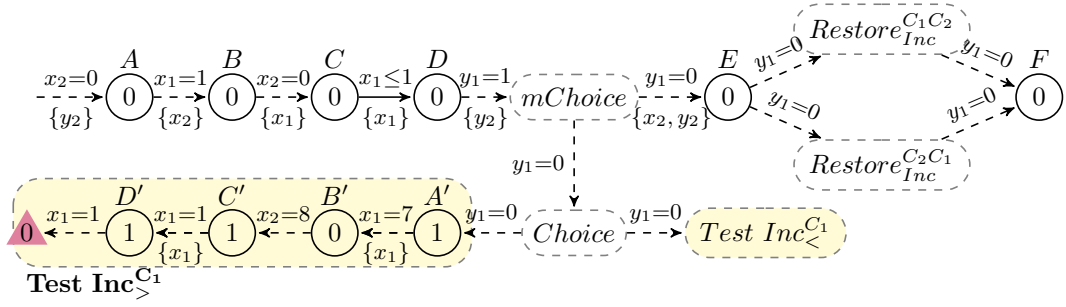

$\mathbf{Test~Inc^{C_1}_{>}}$: 
The delay at $A'$ is $1-\alpha+t+\delta$, 
obtaining $x_2=7+t+\delta$, and the cost accumulated is 
$1-\alpha+t+\delta$. At $B'$, $1-t-\delta$  time is spent, obtaining 
$x_1=1-t-\delta$. Finally, at $C'$, a time $t+\delta$ is spent, and at $D'$, one time unit, making 
the total cost accumulated $2-\alpha+2t+2\delta$ at the target location.
The cost function at the target assigns the cost 0
for all valuations, hence the total cost to reach the target is 
 $2+2t - \alpha + 2 \delta$
which is greater than $2+2\delta$ iff $2t - \alpha > 0$, i.e.
iff $t>\frac{\alpha}{2}$.

\begin{lemma} \label{lem:incre}
Assume that an increment $C_b$ ($b \in \{0, 1\}$) module
is entered with the clock valuations in their normal forms.
Then controller has a strategy to reach either location $l_j$
corresponding to instruction $I_j$ of the two-counter machine
or a target location is reached with cost at most $2+|2 \delta |$,
where $\delta$ is the perturbation added by perturbator.
\end{lemma}

\subsection{Complete Reduction}
The entire reduction consists of constructing a module
corresponding to each instruction $I_i$, $1 \le i \le n$,
of the two-counter machine.
The first location of the module corresponding to instruction $I_1$
is the initial location.
We simulate the halting instruction $I_n$
by a target  location with cost function $cf: \R^5_{\geq 0} \rightarrow \{0\}$. 
We denote the robust timed automaton simulating the two
counter machine by $\mathcal{A}$,
$s$ is the initial state $(l,\zero,\zero)$.
\begin{lemma} \label{lem:reduction}
The two counter machine halts if and only if there
is a strategy $\sigma$ of controller
such that $limcost_{\mathcal{A}}(\sigma, s) \le 2$.
\end{lemma}

The details of the decrement and zero test modules 
are in Appendix \ref{app:undec}. 
They are similar to the increment module; if player 2 
desires to verify the correctness of player 1's simulation, a cost 
$>2+|2\delta|$ is accumulated on reaching a target location 
iff player 1 cheats. In the limit, as $\delta \to 0$, the limcost will be $>2$ iff controller cheats. 
The other possibility to obtain a limcost $>2$ is when the two counter machine does not halt.

\section {Decidability of One-clock RPTG}
\label{sec:robust-dec}
\vspace{0.8em}
\begin{wrapfigure}[5]{o}{3.1cm}
  \scalebox{.8}{
    \begin{tikzpicture}[overlay]
      \draw[thick]   (0.4,-0.7) -- (4,-0.7) -- (4,1.6)--(0.4,1.6);
      \node[fill=yellow] (label) at (2.2,1.6) {A Dwell-time PTG};
    \end{tikzpicture}
    \begin{tabular}{c}
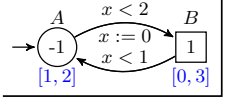
      \begin{tikzpicture}[->,>=stealth',shorten >=1pt,auto,node distance=1cm, semithick]
        \tikzstyle{every state}=[minimum size=1em]
        \node[state,initial, initial text={},initial where=left] at (-.2,0) (A) {-1} ;
        \node[del] () [below of=A] {$[1,2]$};
        \node[locLabel] () [above of =A] {$A$};
        \node[max] at (2,0) (B) {1} ;
        \node[del] () [below of=B] {$[0,3]$};
        \node[locLabel] () [above of =B] {$B$};
        \path(A) edge[bend left]  node[midway,above] {$x<2$} (B);
        \path(A) edge[bend left]  node[midway,below] {$x:=0$} (B);
        \path(B) edge[bend left]  node [midway,above] {$x<1$}(A);
        \path(B) edge[bend left]  node [midway,above] {}(A);
      \end{tikzpicture}
    \end{tabular}
}\end{wrapfigure}
  
In order to show the decidability of the optimal reachability game for 1 clock
RPTG $\mathcal{R}$ and a fixed $\delta>0$, we perform a series of reachability
and optimal cost preserving transformations. 
The idea is to reduce the RPTG into a simpler priced timed game, while preserving 
the optimal costs.  The advantages of this conversion is that the semantics of PTGs are easier to
understand, and one could adapt known algorithms to solve PTGs. 
On the other hand, the PTGs that we obtain are $1$-clock PTGs with dwell-time
requirement (having restrictions on minimum as well as maximum amount of time spent
at certain locations), see for example, 
a dwell-time PTG  with two locations $A, B$. 
A minimum of 1 and a maximum of two units of time should be spent at $A$, while 
a maximum of 3 time units can be spent at $B$.
If we wish to model this using standard PTGs, we need one extra clock and we can not use the
decidability results of $1$ clock PTG  to show the decidability of our model. 
We show in Section~\ref{solve} how to solve 
$1$-clock PTGs with dwell-time requirements.

Our transformations  are as follows: 
(i) for a given $\delta$, our first transformation reduces the RPTG $\rtg$ into
a  dwell-time PTG $\ptg$ (Section \ref{sec:rta_dwell_ptg});  
(ii) our second transformation restricts to dwell-time PTGs where  
   the clock is bounded by $1+\delta$. To achieve this, we use a notion of {\it
   fractional resets}, and denote these  PTGs as $\frptg$
   (Section \ref{app:tran2}); 
(iii) our third and last transformation restricts $\frptg$ without resets
(Section \ref{sec:solve_frptg}). The reset-free dwell-time PTG is denoted $\frptg'$.  
For each transformation, we prove that the optimal cost in each state of the
original game is the same as the optimal cost at some corresponding state of the
new game. We also show that an $(\epsilon,N)$-strategy of the original game can
be computed from some $(\epsilon',N')$-strategy in the new game. 
The details of each transformation and correctness is established in subequent
sections.  
We then solve $\frptg'$ employing a
technique inspired by \cite{BouLar06} while ensuring that the robust semantics
are satisfied.  
  
\subsection{Transformation 1: RPTG $\rtg$ to dwell-time PTG $\ptg$}
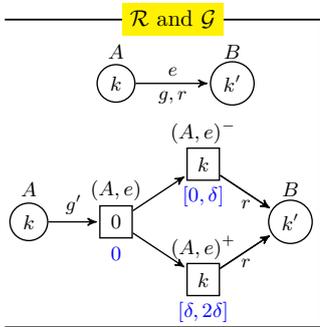
\label{sec:rta_dwell_ptg}
 \begin{wrapfigure}[11]{L}{4.65cm}
   \scalebox{.85}{
    \begin{tikzpicture}[overlay]
     \draw[thick] (0.4,-2.1)-- (5.3,-2.1) -- (5.3,2.7) -- (0.4,2.7) ;
     
     \node[fill=yellow] (label) at (3,2.7) {$\rtg$ and $\ptg$ };
    \end{tikzpicture}
    \begin{tabular}{c}
    
 \begin{tikzpicture}[->,>=stealth',shorten >=1pt,auto,node distance=1cm,  semithick,scale=0.9]
    \node[min] at (-29,17.4) (A) {$k$};
  \node[locLabel] () [above of =A] {$A$};
  \node[cdel] () [below of=A] {$t$};

  \node[min] at (-27,17.4) (B) {$k'$};
  \node[locLabel] () [above of =B] {$B$};

  \draw[trans] (A) -- (B) node[midway,above] {$e$} node[midway,below] {$g,r$};

\node[min] at (-30.5,15) (Ap) {$k$};
  \node[locLabel] () [above of =Ap] {$A$};
  \node[cdel] () [below of=Ap] {$t-\delta$};

  \node[max] at (-29,15) (mu) {$0$};
  \node[locLabel] () [above of =mu] {$\urg{A,e}$};
  \node[del] () [below of=mu] {$0$};

  \node[max] at (-27.5,14) (mp) {$k$};
  \node[locLabel] () [above of =mp] {$\posP{(A,e)}$};
  \node[del] () [below of=mp] {$[\delta,2\delta]$};

  \node[max] at (-27.5,16) (mn) {$k$};
  \node[locLabel] () [above of =mn] {$\negP{(A,e)}$};
  \node[del] () [below of=mn] {$[0,\delta]$};

  \node[min] at (-26,15) (Bp) {$k'$};
  \node[locLabel] () [above of =Bp] {$B$};

  \draw[trans] (Ap)--(mu) node[midway,above] {$g'$};
    \draw[trans] (mu)--(mp) node[midway,above] {};
  \draw[trans] (mu)--(mn) node[midway,above] {};
  \draw[trans] (mp)--(Bp) node[midway,above] {};
  \draw[trans] (mp)--(Bp) node[midway,below] {$r$};
  \draw[trans] (mn)--(Bp) node[midway,below] {$r$};
  \draw[trans] (mn)--(Bp) node[midway,above] {};
\end{tikzpicture}
   \end{tabular}

}
    
  \end{wrapfigure}   
  Given a one clock RPTG $\rtg = (\locations_1, \locations_2,\set{x}, \trans, \prices,\rgoals,f_{goal})$ and a $\delta>0$,  
 we construct a dwell-time PTG $\ptg=(\locations_1,\locations_2\cup L', \set{x},\trans',\prices',\rgoals,f_{goal})$. 
  All the controller, perturbator locations of  $\rtg$ ($\locations_1$ and $\locations_2$) are carried over respectively as player 1, player 2 locations  in $\ptg$.  In addition, we have some new player~2 locations $L'$ in $\ptg$.
   The dwell-time PTG $\ptg$ constructed has  dwell-time restrictions for the new player 2 locations $L'$. 
The locations of $L'$ are either urgent, 
or have a  a dwell-time of $[\delta, 2 \delta]$ or  $[0, \delta]$.
 All the  perturbator transitions of $\rtg$ are retained as it is in $\ptg$. Every transition in $\rtg$
 from a controller location $A$ to some location $B$  
 is replaced in $\ptg$ by a game graph 
as shown. 
\noindent Let $e=(A, g, r,B)$ be the transition from a controller location $A$ to a location $B$ with guard $g$, and reset $r$.  Depending on the guard $g$, in the transformed 
game graph, we have 
the new guard $g'$.  If $g$ is $x=H$, then $g'$ is $x=H-\delta$, while 
if $g$ is $H < x < H+1$, then $g'$ is $H-\delta < x < H+1-\delta$, for $H>0$. 
When $g$ is $0 < x < K$, then $g'$ is $0 \leq x < K-\delta$ and 
$x=0$ stays unchanged.
It can be seen that doing this transformation to all the controller edges 
of a RPTG $\rtg$ gives rise to a dwell-time PTG $\ptg$.
 
Lets consider the transition from $A$ to $B$ in $\rtg$. 
 Assume that the transition from $A$ to $B$ (called edge $e$) had a constraint
$x=1$, and assume that $x=\nu$ on entering $A$. Then, in $\rtg$, 
controller elapses a time $1
-\nu$, and reaches $B$; however on reaching $B$, 
the value of $x$ is in the range $[1-\delta, 1+\delta]$ depending on the perturbation.
Also, the cost accumulated at $A$ is $k*(1-\nu+\gamma)$, where $\gamma \in [-\delta, \delta]$.
To take into consideration these semantic restrictions of $\rtg$, we 
transform the RPTG $\rtg$ into a dwell-time PTG $\ptg$.  First of all, we change the constraint $x=1$ into 
$x=1-\delta$ 
from $A$ (a player 1 location) and enter a new player 2 location $(A,e)$. This player 2 location is an urgent 
location. The correct strategy for player 1 is to spend a time $1-\nu-\delta$ at $A$ (corresponding to 
the time $1-\nu$ he spent at $A$ in $\rtg$). 
 At $(A,e)$, player 2 can either proceed to one of the player 2 locations 
$\negP{(A,e)}$ or $\posP{(A,e)}$.  
The player 2 location $(A,e)$ models perturbator's choices 
of positive or negative perturbation in $\rtg$. 
If player 2 goes to $\negP{(A,e)}$, then on reaching $B$, the value of $x$ is in the interval $[1-\delta,1]$ (this corresponds to perturbator's choice of $[-\delta,0]$ in $\rtg$)
and 
if he goes to $\posP{(A,e)}$, then the value of $x$ at $B$ is in the interval $[1, 1+\delta]$ (this corresponds to perturbator's choice of $[0,\delta]$ in $\rtg$).
 The reset happening in the transition from $A$ to $B$ in $\rtg$ is 
 now done on the transition from $\negP{(A,e)}$ to $B$ and 
 from $\posP{(A,e)}$ to $B$. Thus, note that the possible ranges of $x$ 
 as well as the accumulated cost in $\rtg$ while reaching $B$ are preserved in the 
 transformed dwell-time PTG.  

\begin{lemma}
 \label{lem:rtg_to_ptg}
 Let $\rtg$ be a RPTG and $\ptg$ be the corresponding 
 dwell-time PTG obtained using the transformation 
 above. Then  
 for every state $s$ in $\rtg$, $\optcost{\rtg}(s) = \optcost{\ptg}(s)$. 
 An  $(\epsilon,N)-$strategy in $\rtg$ can be computed from a $(\epsilon,N)-$strategy in $\ptg$ and vice versa.
\end{lemma}
 Proof In Appendix \ref{app:tran1}.

\subsection{Transformation 2: Dwell-time PTG $\ptg$ to  Dwell-time FRPTG $\frptg$}
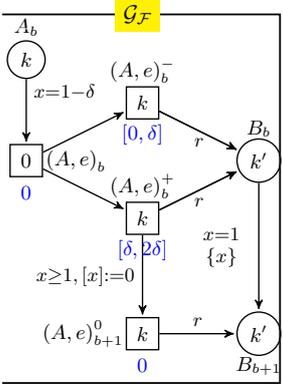
\label{app:tran2}
  \begin{wrapfigure}[14]{L}{4.15cm}
   \scalebox{.85}{
    \begin{tikzpicture}[overlay]
     \draw[thick] (0.4,-2.75)-- (4.7,-2.75) -- (4.7,3) -- (0.4,3) ;
     
     \node[fill=yellow] (label) at (2.5,3) {$\frptg$ };
    \end{tikzpicture}
    \begin{tabular}{c}
    
\begin{tikzpicture}[->,>=stealth',shorten >=1pt,auto,node distance=1cm,  semithick,scale=0.9]

    \node[min] at (6,1.75) (Ap) {$k$};
  \node[locLabel] () [above of =Ap] {$A_b$};
  \node[cdel] () [above of=Ap] {$t-\delta$};

  \node[max] at (6,0) (mu) {$0$};
  \node[locLabel] () [right of =mu] {$~~~~~\urg{A,e}_b$};
  \node[del] () [below of=mu] {$0$};


  \node[max] at (8,-1) (mp) {$k$};
  \node[locLabel] () [above of =mp] {$\posP{(A,e)}_b$};
  \node[del] () [below of=mp] {$[\delta,2\delta]$};

  \node[max] at (8,-3) (mp1) {$k$};
  \node[locLabel] () [left of =mp1,node distance = 9.2mm] {$\zeroP{(A,e)}_{b+1}$};
  \node[del] () [below of=mp1] {$0$};

  \node[max] at (8,1) (mn) {$k$};
  \node[locLabel] () [above of =mn] {$\negP{(A,e)}_b$};
  \node[del] () [below of=mn] {$[0,\delta]$};

  \node[min] at (10,0) (Bp) {$k'$};
  \node[locLabel] () [above of =Bp] {$B_b$};

  \node[min] at (10,-3) (Bp1) {$k'$};
  \node[locLabel] () [node distance =5mm,below of =Bp1] {$B_{b+1}$};

  \draw[trans] (Ap)--(mu) node[midway,right] {} node[pos=0.2,right] {$x{=}1{-}\delta$};
  \draw[trans] (mu)--(mp) node[midway,above] {$ $};
  \draw[trans] (mu)--(mn) node[midway,above] {$ $};
  \draw[trans] (mn)--(Bp) node[midway,below] {$r$};
  
  \draw[trans] (mp)--(Bp) node[midway,below] {$r$};
  \draw[trans] (mp)--(Bp) node[midway,above] {$ $};
  \draw[trans] (mp1)--(Bp1) node[midway,above] {$r$};
  \draw[trans] (mn)--(Bp) node[midway,above] {$ $};
  \draw[trans] (mp)--(mp1) node[midway,left] {$x{\geq} 1, \fr{x}$};
  \draw[trans] (Bp) -- (Bp1) node[midway,left] {$\begin{array}{c}x{=}1\\\set{x}\end{array}$};


\end{tikzpicture}
   \end{tabular}

}
    
  \end{wrapfigure}
Recall that the locations of the dwell-time PTG $\ptg$ is $L_1 \cup L_2 \cup L'$ where $L_1 \cup L_2$ 
are the set of locations of $\rtg$, and $L'$ are new player 2 locations introduced in $\ptg$.
In this section, we  transform the dwell-time PTG $\ptg$  
into a dwell-time PTG $\frptg$
having the restriction that the value of $x$ 
is in [0,1] at all locations corresponding to 
$L_1 \cup L_2$, and is in $[0, 1+\delta]$ at all locations 
corresponding to $L'$. 
  While this transformation is the same as that used in \cite{BouLar06}, the main difference 
is that we introduce special resets called \emph{fractional resets} which reset only the integral part of clock $x$
while its fractional part is retained.
  For instance, if the value of $x$ was 1.3, then 
the operation $[x]:=0$ makes the value of $x$ to be 0.3.   
    
\noindent    Given a one clock, dwell-time PTG  
 $\ptg=(\locations_1,\locations_2 \cup L',\set{x},\trans,\prices,\rgoals,f_{goals})$ with $M$ being the maximum value 
 that can be assumed by clock $x$,   
 we define a dwell-time PTG  with fractional resets (FRPTG) $\frptg$. 
In $\frptg$, we have $M+1$ copies of the locations in $L_1 \cup L_2$ as well as 
 the locations in $L'$ with dwell time $[0, \delta]$, $[0,0]$. These $M+1$ copies of $L'$ have the same dwell-time restrictions in $\frptg$. The copies are indexed 
 by $i, 0 \leq i \leq M$, capturing the integral part 
 of clock $x$ in $\ptg$. 
    Finally, we have in $\ptg$, the locations of $L'$ with dwell-time restriction $[\delta, 2\delta]$.
  For each such location $\posP{(A,e)}$, we have in $\frptg$, the locations 
  $\posP{(A,e)}_i$ and $\zeroP{(A,e)}_{i+1}$ for $0 \leq i \leq  M$.  
  The dwell-time restriction for  
  $\posP{(A,e)}_i$ is same as $\posP{(A,e)}$, while 
  locations $\zeroP{(A,e)}_{i+1}$ are urgent. 
  The prices of locations are carried over as they are in  the various copies. 
  
 The transitions in $\frptg$ consists of 
      the following:
  (1) $l_i \xrightarrow{(g-i)\cap 0 \leq x <1}m_i$\footnote{$g-i$ 
   represents the constraint obtained by shifting the constraint by $-i$}
          if $l \xrightarrow{g} m \in \trans$; 
   (2) $l_i \xrightarrow{(g-i)\cap 0 \leq x < 1; \set{x}}m_0$ if $l \xrightarrow{g;\set{x}} m \in \trans$;
   (3)  $l_i \xrightarrow{x=1,\set{x}} l_{i+1}$, for $l \in L_1 \cup L_2$, and 
  $\posP{(A,e)}_i \xrightarrow{x \geq 1,\fr{x}} \zeroP{(A,e)}_{i+1}$ for $i < M$.
          Consider for example, the constraint $g'$ between $A$ and $(A,e)$ as $x=(b+1)-\delta$ in $\ptg$.
Then the value of $x$ is $b+(1-\delta)$ for $b < M$ when $\posP{(A,e)}$ is entered in $\ptg$.   The location $\posP{(A,e)}$ with $\nu(x) =b+(1-\delta)$ is represented 
in  $\frptg$
 as $\posP{(A,e)}_b$ with $\nu(x)=1-\delta$.
  If player 2 spends 
$[\delta,2\delta]$ time at $\posP{(A,e)}$ in $\ptg$, then $\nu(x) \in [b+1, b+1+\delta]$. 
If there are no resets to goto $B$, 
 then $\nu(x) \in [b+1, (b+1)+\delta]$ at $B$. 
 Correspondingly in $\frptg$,  
 $\nu(x) \in [1,1+\delta]$ 
 at  $\posP{(A,e)}_b$. By construction, $B_b$ is not reachable, 
 since we check $0 \leq x < 1$ on the transition to $B_b$. 
The fractional reset is employed to obtain $x=\delta$ 
while moving to $\zeroP{(A,e)}_{b+1}$. This ensures that
$x=\delta$ on reaching $B_{b+1}$, thereby preserving the perturbation, 
and keeping $x < 1$. A normal reset would have destroyed the 
value obtained by perturbation. 
   The mapping $f$ between states of $\ptg$ and $\frptg$ is as follows:
   $f(l, x)=(l_b, x-b)$, $b < M$, and $x \in [b, b+1]$, $l \in L_1 \cup L_2$, 
     $f((A,e), x)=({(A,e)}_b, x-b)$, $b < M$, and $x \in [b, b+1]$, 
    $f(\negP{(A,e)}, x)=(\negP{(A,e)}_b, x-b)$, $b < M$, and $x \in [b, b+1]$. Finally,  
    $f(\posP{(A,e)}, x)=(\posP{(A,e)}_b, x-b)$, $b < M$, and 
    $x \in [b, b+1] \cup [b+1, b+2]$.
      Note that  in the last case, the value of $x-b$ can exceed 1 but is less than or equal to $1+\delta$.  
\begin{lemma}
 \label{lem:ptg_to_frptg}
 For every state $(l,\val)$ in $\ptg$, $\optcost{\ptg}(l,\val)$ in $\ptg$ is the same as $\optcost{\frptg}(f(l,\val))$  in $\frptg$. For every $\epsilon>0$, $N \in \Nat$, an $(\epsilon,N)$-acceptable strategy in $\ptg$ can be computed from an $(\epsilon,N)$-acceptable strategy in $\frptg$ and vice versa.
\end{lemma}

\subsection{Transformation 3: Dwell-time FRPTG $\frptg$ to resetfree FRPTG $\frptg'$}
\label{sec:solve_frptg}
We now apply the final transformation to the FRPTG $\frptg$ and construct a reset-free
version of the $\FRPTG$ denoted $\frptg'$. Assume that there are a total of $n$ resets (including fractional resets) 
in the FRPTG.   $\frptg'$ consists of  $n+1$ copies of the $\FRPTG$ : $\frptg_0, \frptg_1, \dots, \frptg_{n}$.
Given the locations $L$ of the FRPTG, the locations of $\frptg_i$ are $L^i$, $0 \leq i \leq n$.
$\frptg_0$ starts with $l^0$, where $l$ is the initial location of the $\FRPTG$ and continues until a resetting transition happens. At the first resetting transition,  
$\frptg_0$ makes a transition to $\frptg_1$. The $n$th copy is directed 
to a sink target location $S$ with cost function $cf: \mathbb{R}_{\geq 0} \rightarrow \{+\infty\}$ on the $(n+1)$th reset. Note that each $\frptg_i$ is reset-free.
One crucial property of each $\frptg_i$ is that 
on entering with some value of $x$ in $[0, \delta]$, 
the value of $x$ only increases as the transitions go along in $\frptg_i$; moreover,
$x \leq 1+\delta$ in each $\frptg_i$ by construction. The formal details and proof of Lemma \ref{lem:resetfree_frptg} can be found in Appendix \ref{app:resetfree}.
 \begin{wrapfigure}{o}{3cm}
   \scalebox{.8}{
    \begin{tikzpicture}[overlay]
    \draw[thick]   (0.4,-5.5) -- (3.5,-5.5)--(3.5,5.5) -- (0.4,5.5);
     \node[fill=yellow] (label) at (2,6) {Example};
    \end{tikzpicture}
    \begin{tabular}{c}
 \begin{tikzpicture}[scale=0.5,->]
\tikzstyle{every node}=[font=\scriptsize]

    \draw [<->,thick] (0,5) node (yaxis) [above] {$y$}
        |- (5,0) node (xaxis) [right] {$x$};
    
    \node at (2.5,-0.5) {Superimposition};
    
    \node at (1.6,3) (f2) {$f_2$};
    \coordinate (f2p) at (1.6,1.8);
    \draw[trans] (f2p) -- (f2);
    
    \node at (0.6,4.5) (f1) {$f_1$};
    \coordinate (f1p) at (0.4,3.3);
    \draw[trans] (f1p) -- (f1);

    \draw[-,fun2] (0,3) coordinate (b_1) -- (4,0) coordinate (b_2);
    \node[node distance = 2mm,left of=b_1] {};
    \node[node distance = 2mm,below of=b_2] {};

    \draw[-,fun1] (0,4) coordinate (a_1) -- (2,0.5) coordinate (a_2) -- (3,3) coordinate (a_3) -- (4,0) coordinate (a_4);
    
    \node[node distance = 2mm,left of=a_1] {};
    
    \coordinate (c1) at (intersection of a_1--a_2 and b_1--b_2);
        \draw[-,dashedl] (yaxis |- c1) node[left] {};
         -| (xaxis -| c1) node[below] {};

    \coordinate (c2) at (intersection of a_2--a_3 and b_1--b_2);
        \draw[-,dashedl] (yaxis |- c2) node[left] {};
         -| (xaxis -| c2) node[below] {};
    
 \end{tikzpicture}
 \\
   \begin{tikzpicture}[scale=0.5]
\tikzstyle{every node}=[font=\scriptsize]
    \draw [<->,thick] (0,4) node (yaxis) [above] {$y$}
        |- (5,0) node (xaxis) [right] {$x$};
    
    \node at (2.5,-0.5) {Interior};
    
    \draw[white] (0,3) coordinate (b_1) -- (4,0) coordinate (b_2);
    \node[node distance = 2mm,left of=b_1] {};
    \node[node distance = 2mm,below of=b_2] {};

    \draw[white] (0,4) coordinate (a_1) -- (2,0.5) coordinate (a_2) -- (3,3) coordinate (a_3) -- (4,0) coordinate (a_4);
    
    \node[node distance = 2mm,left of=a_1] {};
    
    \coordinate (c1) at (intersection of a_1--a_2 and b_1--b_2);
        \draw[dashedl] (yaxis |- c1) node[left] {};
         -| (xaxis -| c1) node[below] {};

    \coordinate (c2) at (intersection of a_2--a_3 and b_1--b_2);
        \draw[dashedl] (yaxis |- c2) node[left] {};
         -| (xaxis -| c2) node[below] {};

    \draw[fun2] (b_1) -- (c1);
    \draw[fun1]  (c1) -- (a_2) -- (c2);
    \draw[fun2] (c2) -- (b_2);
    
\end{tikzpicture}
 \\
  \begin{tikzpicture}[scale=0.5]
\tikzstyle{every node}=[font=\scriptsize]
    \draw [<->,thick] (0,5) node (yaxis) [above] {$y$}
        |- (5,0) node (xaxis) [right] {$x$};
    
    \node at (2.5,-0.5) {Exterior};
    
    \draw[white] (0,3) coordinate (b_1) -- (4,0) coordinate (b_2);
    \node[node distance = 2mm,left of=b_1] {};
    \node[node distance = 2mm,below of=b_2] {};

    \draw[white] (0,4) coordinate (a_1) -- (2,0.5) coordinate (a_2) -- (3,3) coordinate (a_3) -- (4,0) coordinate (a_4);
    
    \node[node distance = 2mm,left of=a_1] {};
    
    \coordinate (c1) at (intersection of a_1--a_2 and b_1--b_2);
        \draw[dashedl] (yaxis |- c1) node[left] {};
         -| (xaxis -| c1) node[below] {};

    \coordinate (c2) at (intersection of a_2--a_3 and b_1--b_2);
        \draw[dashedl] (yaxis |- c2) node[left] {};
         -| (xaxis -| c2) node[below] {};

    \draw[fun1] (a_1) -- (c1);
    \draw[fun2] (c1) -- (c2);
    \draw[fun1] (c2) -- (a_3)--(a_4);
\end{tikzpicture}
   \end{tabular}
}
    
  \end{wrapfigure}
Using the cost function of $S$ and those of the targets, we compute 
the optimal cost functions for all the locations of the deepest component $\frptg_n$.  The cost functions of the locations 
of $\frptg_i$ are used to compute 
that of $\frptg_{i-1}$, and so on  until the cost function of $l^0$, the starting location of $\frptg_0$ is computed. An example
can be seen in Appendix \ref{app:eg}. 
\begin{lemma}
 \label{lem:resetfree_frptg}
 For every state $(l,\val)$ in $\frptg$, $\optcost{\frptg}(l,\val) = \optcost{\frptg'}(l^0,\val)$, where  
 $\frptg'$ is the resetfree \FRPTG. For every $\epsilon >0, N \in \Nat$, given an $(\epsilon,N)$-acceptable strategy $\sigma'$ in $\frptg'$, we can compute a $(2\epsilon,N)$-acceptable strategy $\sigma$ in $\frptg$ and vice versa.
\end{lemma}

\subsection{Solving the Resetfree $\FRPTG$}
\label{solve}
Before we sketch the details, let us introduce some key notations.  
Observe that after our simplifying transformations, the cost functions $cf$ are
piecewise-affine continuous functions that assign a value to every valuation
$x \in [0,1+\delta]$ (construction of FRPTG ensures $x {\leq}1{+}\delta$
always).  
The {\it interior} of two cost functions $f_1$ and $f_2$ is a cost function
$f_3: [0, 1+\delta] \rightarrow \mathbb{R}$ defined by $f_3(x)= \min(f_1(x),
f_2(x))$. 
Similarly, the {\it exterior} of $f_1$ and $f_2$ is a cost function $f_4: [0,
1+\delta] \rightarrow \mathbb{R}$ defined as $f_4(x) = \max(f_1(x),
f_2(x))$.   
Clearly, $f_3$ and $f_4$ are also piecewise-affine continuous. 
The interior and exterior can be easily computed by {\it superimposing} $f_1$
and $f_2$  as shown graphically in the example by computing lower envelope and
upper envelope respectively. 

We now work on the reset-free components $\frptg_i$,  and give an algorithm to  
compute $OptCost_{\frptg_i}(l, \nu)$ 
  for every state 
$(l, \nu)$ of $\frptg_i$, $\nu(x) \in [0, 1+\delta]$.   We also show the existence of an  $N$ 
such that for any $\epsilon>0$,  
and every $l \in L^i, \nu(x) \in [0,1+\delta]$,  an $(\epsilon,N)$-acceptable strategy can be computed.
Consider the location of $\frptg_i$ that has the smallest price and call it
$l_{min}$. 
If this is a player 1 location, then intuitively, player 1 would want to spend
as much time as possible here, and if this is a player 2 location, then player 2
would want to spend as less time as possible here.  
By our assumption, all the cycles in $\frptg_i$ are non-negative, and hence if
$l_{min}$ is part of a cycle,  
revisiting it will only increase the total cost if at all. 
Player 1 thus would like to spend all the time he wants to during the first visit
itself. We now prove that this is indeed the case.  
We consider two cases separately. 
\subsubsection{$\lmin$ is a Player~1 location}
We split $\frptg_i$ such that $\lmin$ is visited only once. We transform $\frptg_i$ 
into $\frptg''$ which has two copies 
of all locations except $\lmin$ such that corresponding to every  location $l
\neq l_{min}$, we  have the copies $(l,0)$ and $(l,1)$. A special target
location $S$ is added with cost function assigning $+\infty$ to all clock
valuations.    

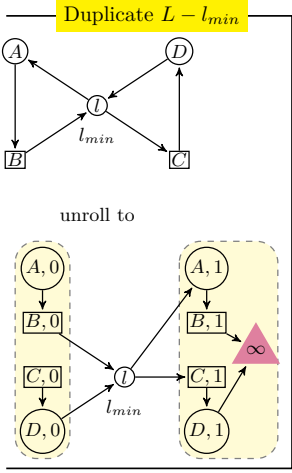
\begin{wrapfigure}[14]{l}{4cm}
   \scalebox{.8}{
    \begin{tikzpicture}[overlay]
    \draw[thick]   (0.4,4) -- (5.1,4) --(5.1,-3.5)--(0.4,-3.5);
     \node[fill=yellow] (label) at (2.8,4) {Duplicate $L-\lmin$};
    \end{tikzpicture}
    \begin{tabular}{c}
   \begin{tikzpicture}[->,>=stealth',shorten >=1pt,auto,node distance=1cm,  semithick,scale=0.9]
\tikzstyle{min}=[draw,circle,minimum size=2mm,inner sep=0.1em] 
\tikzstyle{max}=[draw,rectangle,minimum size=2mm,inner sep=0.1em] 

  \node[min] at (0,1) (A) {$A$};
  \node[max] at (0,-1) (B) {$B$};
  \node[min] at (1.5,0) (m) {$l$};
  \node () [node distance=5mm,below of=m] {$l_{min}$};
  \node[max] at (3,-1) (C) {$C$};
  \node[min] at (3,1) (D) {$D$};
  
  \draw[trans] (m) -- (A);
  \draw[trans] (B) -- (m);
  \draw[trans] (m) -- (C);
  \draw[trans] (D) -- (m);
  \draw[trans] (A) -- (B);  
  \draw[trans] (C) -- (D);
  
  \node at (1.5, -2) (g0) {${\mbox{unroll to}}$};

  \draw[dashed,fill=yellow!20,draw=gray,rounded corners=10pt] (0,-2.5) rectangle (1,-6.5);
  \draw[dashed,fill=yellow!20,draw=gray,rounded corners=10pt] (3,-2.5) rectangle (4.8,-6.5);
  
  \node[min] at (0.5,-3) (A0) {$A,0$};
  \node[max] at (0.5,-4) (B0) {$B,0$};
  \node[max] at (0.5,-5) (C0) {$C,0$};
  \node[min] at (0.5,-6) (D0) {$D,0$};
  
  \node[min] at (2,-5) (mn) {$l$};
    \node () [node distance=5mm,below of=mn] {$l_{min}$};
  
  \node[min] at (3.5,-3) (A1) {$A,1$};
  \node[max] at (3.5,-4) (B1) {$B,1$};
  \node[max] at (3.5,-5) (C1) {$C,1$};
  \node[min] at (3.5,-6) (D1) {$D,1$};
  
  \node[triangle] at (4.4,-4.5) (F) {$\infty$};

  \draw[trans] (mn) -- (A1);
  \draw[trans] (B0) -- (mn);
  \draw[trans] (mn) -- (C1);
  \draw[trans] (D0) -- (mn);
  \draw[trans] (A0) -- (B0);  
  \draw[trans] (A1) -- (B1);  
  \draw[trans] (C0) -- (D0);
  \draw[trans] (C1) -- (D1);
  \draw[trans] (B1) -- (F);
 \draw[trans] (D1) -- (F);
\end{tikzpicture}
    \end{tabular}

}
    
  \end{wrapfigure}
Given the transitions $\trans$ of $\frptg_i$, the FRPTG $\frptg''$ has the following transitions.
\begin{itemize}   
\item if $l \xrightarrow{g}l' \in \trans$ and $l,l' \not = \lmin$ then $(l,0) \xrightarrow{g}(l',0)$ 
  and $(l,1) \xrightarrow{g}(l',1)$
\item if $l \xrightarrow{g}l' \in \trans$ and $l' = \lmin$ then $(l,0) \xrightarrow{g}\lmin$ 
  and  $(l,1) \xrightarrow{g}S$, 
\item if $\lmin \xrightarrow{g} l$, then $\lmin \xrightarrow{g} (l,1)$
\end{itemize}
\begin{lemma}
 \label{lem:minLoc_split}
 For every state $(l,\val)$ if $\val {\in} [0,1+\delta]$ and $l {\neq} \lmin$, 
 we have that 
 \begin{eqnarray*} 
 \optcost{\frptg_i}(l,\val) = \optcost{\frptg''}((l,0),\val) & \text{ and } &
 \optcost{\frptg_i}(\lmin,\val) = \optcost{\frptg''}(\lmin,\val).
\end{eqnarray*}
\end{lemma}
We give an intuition for Lemma \ref{lem:minLoc_split}. Locations $(l,0)$ have all the transitions available to location
$l$ in $\frptg_i$.  Also, any  play in $\frptg''$ which is compatible with a winning strategy
of player 1  in $\frptg_i$ contains only  one of the locations $(l,0), (l,1)$ by construction 
of $\frptg''$.  The outcomes from $(l,0)$ are more favourable than $(l,1)$ for $l$ as a player 1 location. 
Based on these intuitions, we conclude that 
$\optcost{\frptg_i}(l,\val)$ is same as that for $((l,0),\val)$. This observation 
also leads to the $\epsilon-$optimal strategy being the same as that for $(l,0)$.
Given a strategy $\strat'$ in $\frptg''$, we construct $\strat$ in $\frptg_i$ as 
$\strat(l,\val) = \strat'((l,0),\val)$.
Further, any strategy that revisits $\lmin$ in $\frptg_i$ 
cannot be winning for player 1, since all cycles are non-negative; 
we end up at $S$ with cost $\infty$ in $\frptg''$. However, 
all strategies that do not revisit $\lmin$ in $\frptg_i$ are preserved 
in $\frptg''$, and hence 
$\optcost{\frptg_i}(\lmin,\val) = \optcost{\frptg''}(\lmin,\val)$. 
%
We iteratively solve the part of $\frptg''$ with locations indexed 1 (i.e; $(l,1)$) 
in the same fashion (picking minimal price locations) 
each time obtaining a smaller PTG.  
Computing the cost function of the minimal price location of the last such PTG, 
and propagating this backward, we compute the cost function of $\lmin$.
  We then use the cost function of $\lmin$ to solve the 
part of $\frptg''$ with locations indexed 0 (i.e; $(l,0)$). 

\noindent{\it {\bf Computing the Optcost function of $\lmin$}}:
\label{sec:optcost_min}
\LinesNumberedHidden
\begin{algorithm}[t] 
    Let $l_1, \dots, l_n$ be the successors of $\lmin$ with optcost functions $f_1,f_2 \cdots f_n$.\;
    \underline{\textbf{STEP~1 : Superimpose :}} Superimpose all the optcost functions $f_1,f_2 \cdots f_n$.\;
    \underline{\textbf{STEP~2 : Interior :}} Take the {interior} of the superimposition; call it  $f$.\;
    Let $f$ be composed of line segments $g_1,g_2 \cdots g_m$ such that 
    $g_i \in \{f_1, \dots, f_n\}$, for all $i$. 
        $\forall~k,$ let the domain of $g_k$ be  $[u_k,v_k]$.
    Set $i=m$.\;
    \underline{\textbf{STEP~3 : Selective Replacement : }} 
    \While{ $i \geq 1$}{
	\eIf{slope of $g_i \leq - \prices(\lmin)$}{
	   replace $g_i$ with line $h_i$ with slope $-\prices(\lmin)$ and passing through $(v_i,g_i(v_i))$\;
	   Let $h_i$ intersect $g_j$ (largest $j<i$) at some point $x=v_j''$, $v_j'' \in [u_j,v_j]$\;
	   Update domain of $g_j$  from $[u_j,v_j]$ to $[u_j,v''_j]$\;
	   \If{$j<i-1$}{Remove functions $g_{j+1}$ to $g_{i-1}$ from $f$}
	   Set $i=j$;
	}{
	    $i=i-1$\;
	}
    }
    \underline{\textbf{STEP~4 : Refresh Interior :}} Take the interior after STEP~3 and call it  $f'$.\;
    \If{$l'' \xrightarrow{}\lmin$}{update the optcost function of $l''$}
    \caption{Optimal Cost Algorithm when $\lmin$ is a Player~1 location}
  \label{algo:optcost_compute}
\end{algorithm}
Algorithm \ref{algo:optcost_compute} computes the optcost function 
for a player 1 location $\lmin$, assuming all the constraints on outgoing transitions 
from $\lmin$ are the same, namely $x \in [0, 1]$. We discuss adapting the algorithm to work for 
transitions with different constraints in Appendix \ref{app:pl1}.
A few words on the notation used: if a location $l$ has price $\eta(l)$, then 
 slope associated with $l$ is $-\eta(l)$ (see STEP 3 in Algorithm \ref{algo:optcost_compute}). 

Let $l_1, \dots, l_n$ be the successors of $\lmin$, with cost functions $f_1, \dots, f_n$.
Each of these cost functions are piecewise affine continuous over the domain $[0, 1]$. 
The first thing to do is to superimpose $f_1, \dots, f_n$, and obtain the 
cost function $f$ corresponding to the interior of $f_1, \dots, f_n$ ($\lmin$ is a player 1 location and would 
like to obtain the minimal cost, hence the interior). 
    The line segments comprising $f$ come from the various $f_i$. 
Let $dom(f)=[0,1]$ be composed of $0=u_{i_1} \leq v_{i_1}=u_{i_2} \leq \dots u_{i_m} \leq v_{i_m}=1$ : that is, 
$f(x)=f_{i_j}(x)$, $dom(f_{i_j})=[u_{i_j},v_{i_j}]$, for $i_j \in \{1,2,\dots,n\}$ and  $1 \leq j \leq m$.
Let us denote $f_{i_j}$ by $g_j$, for $1 \leq j \leq m$.  Then, $f$ is composed of $g_1, g_2, \dots, g_m$, 
and  $dom(f)$ is composed of $dom(g_1), \dots, dom(g_m)$ from left to right.  Let $dom(g_i)=[u_i,v_i]$. 
  Step 2 of the algorithm achieves this.
  
      For a given valuation $\nu(x)$, if $\lmin$ is an urgent location, 
    then player 1 would go to a location $l_k$ 
   if the interior $f$ is such that $f(\nu(x))=g_k(\nu(x))$(the least cost is given by $g_k$, obtained 
  from the outside cost function of $l_k$).  
    If $\lmin$ is not an urgent location, then player 1 would prefer  delaying  $t$ units at $\lmin$ 
 so that $\nu(x)+t \in [u_i, v_i]$ rather than goto some location $l_i$  if  
 $g_i(\nu(x)) > \prices(\lmin)(v_i-\nu(x))$.
 Again, $g_i$ is a part of the ouside cost function of $l_i$, and 
 player 1 prefers delaying time at $\lmin$ rather than goto $l_i$ since that minimizes the cost. 
   In this case, the 
 cost function $f$ is refined by replacing the line segment $g_i$ 
 over $[u_i,v_i]$
     by another line segment $h_i$ passing through $(v_i, g_i(v_i))$, and having a slope $-\prices(\lmin)$. Step 3 of the algorithm does this.

    Recall that by our transformation 2,  the value of clock $x$ in any player 1   
      location is $\leq 1-\delta$. The value of $x$ is in  $[1-\delta,1+\delta]$
    only  at a player 2 location ($\posP{(A,e)_b}$ in the FRPTG, section \ref{app:tran2}). Hence, the domain of cost functions  for player 1 locations is actually $[0,1-\delta]$, and not $[0,1+\delta]$. Let the 
      domain of $g_m$ be $[u_m, 1]$. Then we can split $g_m$ into 
      two functions $g^1_m, g^2_m$ with  domains $[u_m,1-\delta]$ and $[1-\delta,1]$.
    Now, we ensure that no time is spent in the player 1 location $\lmin$
       over $dom(g^2_m)$, by not applying step 3 
       of the algorithm for $g^2_m$. This way, selective replacement 
       of the cost functions $g_i$ occur only in the domain $[0,1-\delta]$, and 
       we remain faithful to transformation 2, and the semantics of RPTGs.
       
 \noindent{\it {\bf Computing Almost Optimal Strategies}}: The strategy corresponding to this 
computed optcost is derived as follows.  
$f'$ is the optcost of location $\lmin$ computed 
in Step 4 of the algorithm.  $f'$ is composed of two kinds of functions (a) the functions $g_i$ computed in step 2 as a result of the interior of superimposition and (b) 
functions $h_i$ which replaced some functions $g_j$ from $f$, corresponding to 
delay at $\lmin$. For functions $h_j$ of $f'$ with domain $[u_j,v_j]$, we prescribe the strategy to 
delay at $\lmin$ till $x =v_j$ when entered with clock $x \in [u_j,v_j]$.
 For functions $g_i$, that come  
 from $f$ at Step 2, where $g_i$ is part of some optcost function $f_k$, ($f_k$ is the optcost function  of one of the successors $l_k$ of $\lmin$), 
  the strategy dictates moving 
immediately to $l_k$ when entered with 
clock $x \in [u_i,v_i]$.

 \noindent{\it {\bf Termination}}: 
Finally, we prove the existence of a number $N$, the number of 
affine segments that appear in the cost functions of all locations.
Start with the resetfree FRPTG with $m$ locations
having $p$ segments in the outside cost functions.  
Let $\alpha(m,p)$ denote the total number of 
affine segments appearing in cost functions 
across all locations. 
The transformation
of resetfree components $\frptg$ into $\frptg''$ gives rise to two smaller resetfree FRPTGs
with $m-1$ locations each, after separating out $\lmin$. 
The resetfree FRPTG $(\frptg,1)$ with $m-1$ locations indexed with 1 of the form $(l,1)$ are solved first, 
these cost functions are added as outside cost functions to solve 
 $\lmin$,  and finally, the cost function of $\lmin$ 
 is added as an outside cost function to solve 
 the resetfree FRPTG $(\frptg,0)$ with $m-1$ locations indexed with 0 of the form $(l,0)$.
Taking into account the new sink target location added, we have $\leq p+1$ segments in outside cost functions in $(\frptg,1)$. 
This gives atmost $\beta=\alpha(m-1,p+1)$ segments in solving 
$(\frptg,1)$, and $\alpha(1, p+\beta)=\gamma$ segments to solve $\lmin$, and
finally $\alpha(m-1, p+\gamma)$ segments to solve 
$(\frptg,0)$. Solving this, one can easily check that $\alpha(m,p)$ is atmost triply exponential 
in the number of locations $m$ of the resetfree component $\frptg$.
Obtaining a bound of the number of affine segments, it is easy to see that Algorithm 1 terminates; the time taken to compute almost optimal strategies and optcost functions is triply exponential.

We illustrate the computation of Optcost of a Player~1 location in Figure \ref{fig:optcost_compute}. The proof of Lemma \ref{lem:replaceFun} is given in Appendix \ref{app:pl1}, while 
Lemma \ref{lem:optcost_function} follows from Lemma \ref{lem:replaceFun} 
and Step 4 of Algorithm \ref{algo:optcost_compute}. 

\begin{lemma}
\label{lem:replaceFun}
In Algorithm \ref{algo:optcost_compute},
 if a function $g_i$ (in $f$ of Step 2) has domain $[u_i,v_i]$ and slope $\leq -\prices(l)$
 then $ \optcost{}(l,\val) = (v_i - \val) * \prices(l) + g(v_i)$. 
\end{lemma}
\begin{lemma}
 \label{lem:optcost_function}
 The function $f'$ in  Algorithm \ref{algo:optcost_compute} computes  the optcost at any location $l$. That is, 
$\forall x \in [0,1],$
 $\optcost{\ptg}(l,x) = f'(x)$.
\end{lemma}
 Note that the strategy under construction is a player 1 strategy,  and player 1 has no control over the interval $[1,1+\delta]$. $x \in [1,1+\delta]$ after a positive perturbation, and is under player 2's 
 control. Thus, at a player 1 location, proving for $x \in [0,1]$ suffices.
\subsubsection{$\lmin$ is a Player~2 location}
If $l_{min}$ is a player 2 location in the reset-free component $\frptg_i$, 
then intuitively, player 2 would want to spend as little time as possible there.
Keeping this in mind, 
we first run steps 1, 2 of Algorithm \ref{algo:optcost_compute}
by taking the exterior of $f_1, \dots, f_n$ instead of the interior(player 2 would maximise the cost). 
There is no time elapse at $\lmin$ on running steps 1,2 of the algorithm.  
  Let $f$ be the computed exterior using steps 1,2. 
   If  $f$ comprises of functions $g_i$ having a greater slope 
than $-\prices(l)$, then 
 it is better to delay at $\lmin$ to increase the cost.  
 In this case, player 2 would want to 
improve his optcost using Step 3, by  spending time at $l_{min}$. 
Finally, while doing Step 4, we take the exterior 
of the replaced functions $h_i$ and old functions $g_i$.
Recall that our transformations resulted in 3 kinds of player 2 locations : urgent, those with 
dwell-time restriction $[0, \delta]$ and finally those with $[\delta,2\delta]$.
The 3 cases are discussed in detail in Appendix~\ref{app:pl2}. 
\begin{figure}
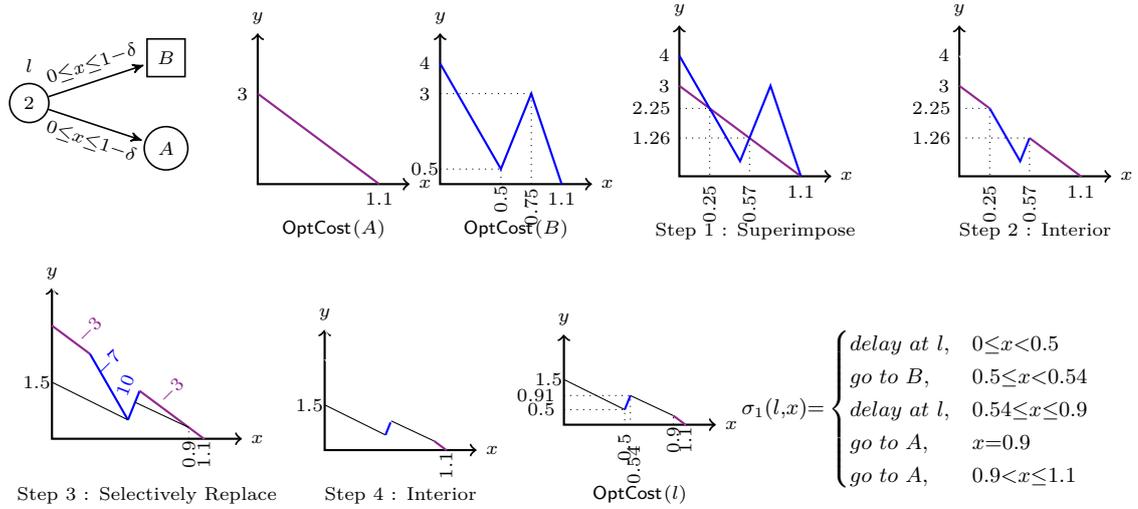

\begin{tabular}{c c c c}
\input{./figs/fun_ex.tex}
& 
\input{./figs/fun_s0.tex}
&
\input{./figs/fun_s1.tex}
 & 
 \input{./figs/fun_s2.tex}
 
 \end{tabular}

\begin{tabular}{c c c c}
 \input{./figs/fun_s3_new.tex} 
 &
 \input{./figs/fun_s4.tex} 
& 
\input{./figs/fun_fin.tex}
&
\scalebox{0.9}{
\hspace{-2.5cm}
\vbox{

$ $
\begin{equation*}
    \strat_1(l{,}x){=}
    \begin{cases}
    delay~at~l, & \text{$0{\leq} x {<}0.5$}\\
    go~to~B, & \text{$0.5 {\leq} x {<} 0.54$}\\
    delay~at~l, & \text{$0.54{\leq} x {\leq} 0.9$}\\
    go~to~A, & \text{$ x {=} 0.9$}\\
    go~to~A, & \text{$0.9 {<} x {\leq} 1.1$}
  \end{cases}
\end{equation*} }}
\end{tabular}

\caption{Optcost Computation for a Player~1 location ($\delta=0.1$): we can keep the guards as $0\leq x \leq 1$
and not apply Step 3 for $x \in [1-\delta,1]$.  
} 
  \label{fig:optcost_compute}
\end{figure}


\section{Conclusion and Future Work}
\label{sec:robust-concl}
In this paper we studied excess robust semantics and provided the first
decidability result for excess semantics and  improved the known undecidability
result with 10 clocks to 5 clocks. 
To the best of our knowledge, the other known decidability result 
for robust timed games is under the conservative semantics 
 for a fixed $\delta$, \cite{prabhu}. 
As a  consequence of our decidability result, the reachability problem for 1
clock PTG with arbitrary prices is shown to be decidable too under the
assumption that the PTG does not have any negative cost cycle. 
The decidability we show is for a fixed perturbation bound $\delta>0$.  
We use  $\delta$ in the constraints of the 
dwell-time PTG after the first transformation for ease of understanding the 
robust semantics. 
Implementing this in step 3 of Algorithm 1 and ensuring 
no time elapse in the interval $[1-\delta,1]$ takes no extra effort while 
$\lmin$ is a player 1 location. 
In that sense, we could have avoided explicit use of $\delta$ in the constraints 
in our simplifying transformations, and taken the appropriate steps 
in the algorithm itself.
The existence of limit-strategy with $\delta \to 0$ seems rather hard. 
Our construction would not directly extend to limit-strategy problem as it is
heavily dependant on the fixed $\delta$.  


\newpage
\bibliographystyle{plain}
\bibliography{papers}

\begin{thebibliography}{10}

\bibitem{ABM04}
R.~Alur, M.~Bernadsky, and P.~Madhusudan.
\newblock Optimal reachability for weighted timed games.
\newblock In J.~D\'iaz, J.~Karhum\"aki, A.~Lepist\"o, and D~Sannella, editors,
  {\em Proc. ICALP'04}, volume 3142 of {\em LNCS}, pages 122--133. Springer,
  2004.

\bibitem{BFHLPRV01}
G.~Behrmann, A.~Fehnker, T.~Hune, K.~G. Larsen, P.~Pettersson, J.~Romijn, and
  F.~W. Vaandrager.
\newblock Minimum-cost reachability for priced timed automata.
\newblock In M.~D. Di~Benedetto and A.~L. Sangiovanni-Vincentelli, editors,
  {\em Proc. HSCC'01}, volume 2034 of {\em LNCS}, pages 147--161, Heidelberg,
  2001. Springer.

\bibitem{BBM06}
P.~Bouyer, T.~Brihaye, and N.~Markey.
\newblock Improved undecidability results on weighted timed automata.
\newblock {\em Information Processing Letters}, 98:188--194, 2006.

\bibitem{BCFL04}
P.~Bouyer, F.~Cassez, E.~Fleury, and K.~G. Larsen.
\newblock Optimal strategies in priced timed game automata.
\newblock In K.~Lodaya and M.~Mahajan, editors, {\em {FSTTCS}'04}, volume 3328
  of {\em LNCS}, pages 148--160. Springer, 2004.

\bibitem{BouLar06}
Patricia Bouyer, Kim~Guldstrand Larsen, Nicolas Markey, and Jacob~Illum
  Rasmussen.
\newblock Almost optimal strategies in one clock priced timed games.
\newblock In {\em {FSTTCS} 2006: Foundations of Software Technology and
  Theoretical Computer Science, 26th International Conference, Kolkata, India,
  December 13-15, 2006, Proceedings}, pages 345--356, 2006.

\bibitem{BMS12}
Patricia Bouyer, Nicolas Markey, and Ocan Sankur.
\newblock Robust reachability in timed automata: A game-based approach.
\newblock In {\em Automata, Languages, and Programming}, volume 7392 of {\em
  Lecture Notes in Computer Science}, pages 128--140. Springer, 2012.

\bibitem{ocan}
Patricia Bouyer, Nicolas Markey, and Ocan Sankur.
\newblock Robust weighted timed automata and games.
\newblock In V{\'\i}ctor Braberman and Laurent Fribourg, editors, {\em
  {P}roceedings of the 11th {I}nternational {C}onference on {F}ormal
  {M}odelling and {A}nalysis of {T}imed {S}ystems ({FORMATS}'13)}, volume 8053
  of {\em Lecture Notes in Computer Science}, pages 31--46, Buenos Aires,
  Argentina, August 2013. Springer.

\bibitem{BBR05}
T.~Brihaye, V.~Bruy{\`e}re, and J.~Raskin.
\newblock On optimal timed strategies.
\newblock In P.~Pettersson and W.~Yi, editors, {\em Proc. FORMATS'05}, volume
  3829 of {\em LNCS}, pages 49--64. Springer, 2005.

\bibitem{prabhu}
Krishnendu Chatterjee, Thomas~A. Henzinger, and Vinayak~S. Prabhu.
\newblock Timed parity games: Complexity and robustness.
\newblock {\em Logical Methods in Computer Science}, 7(4), 2011.

\bibitem{HIM13}
T.~D. Hansen, R.~Ibsen-Jensen, and P.~B. Miltersen.
\newblock A faster algorithm for solving one-clock priced timed games.
\newblock In PedroR. D’Argenio and Hernán Melgratti, editors, {\em CONCUR
  2013 – Concurrency Theory}, volume 8052 of {\em Lecture Notes in Computer
  Science}, pages 531--545. Springer, 2013.

\bibitem{Mar11}
N.~Markey.
\newblock Robustness in real-time systems.
\newblock In {\em Industrial Embedded Systems (SIES), 2011 6th IEEE
  International Symposium on}, pages 28--34, June 2011.

\bibitem{Rut}
Michal Rutkowski.
\newblock Two-player reachability-price games on single-clock timed automata.
\newblock In {\em QAPL}, pages 31--46, 2011.

\bibitem{BMS13a}
Ocan Sankur, Patricia Bouyer, Nicolas Markey, and Pierre-Alain Reynier.
\newblock Robust controller synthesis in timed automata.
\newblock In {\em CONCUR 2013 – Concurrency Theory}, volume 8052 of {\em
  Lecture Notes in Computer Science}, pages 546--560. Springer, 2013.

\end{thebibliography}
\newpage
\appendix
\centerline{\Large \bf{Appendix}}
\section{Cost Functions}
\label{app:cost-eg}
\input{./figs/cost-ptg.tex}
We illustrate the cost functions with an example. In the PTG given here, the cost function $f$ corresponding to the target gives the cost incurred when the target is entered various values of clock $x$. For example, if target is reached with clock value $x=0$ then cost incurred is $cost= 3 = f(0)$  while $cost=0$ if entered with $x\in[0.5,1]$. 
Suppose $B$ is entered with $x=0$ and Player 1 decides to go to the target immediately with no delay at $B$ (i.e; delay $d=0$) then the cost is $cost = 1 * d + f(v) = 1 * 0 + f(0) = 3$, and $x=v$ upon entering the target. However, if player 1 waited at $B$ till $x=0.5$ and then went to target, the cost is $1 * 0.5 + f(0.5) = 0.5$. Similarly, if $d=0.75$ then the cost is 0.75. From this, we can infer that the best strategy for Player 1 to achieve the optimal cost is to wait till $x=0.5$ and then go to target. The second function labelled \textit{OptCost of $B$} gives the optimal cost achievable for every value of $x$ that $B$ is entered with. 
Similar analysis for location $A$, reveals that the cost incurred is $-1$ if Player 2 went to target directly. Else, he could wait at $A$ and then go to $B$. Due to the negative price at $A$, it is obvious that the best strategy for Player 2 is to go to $B$ immediately. Thus, the optimal cost function for $A$ is the same as that of $B$.

\section{UndecidabilityProof}
\label{app:undec}

We present below a set of figures which depict in full detail the simulation of all the instructions of two counter machine - increment, zero test and decrement.

First we describe a few \emph{support modules}
that will be used in the main modules
for simulating increment, zero test and decrement instructions.

\subsection{Prevent perturbation module} \label{sec:prev_perturb}
For correct simulation of the instructions,
it will often be needed that the delay made by
controller should not be perturbed by perturbator.
The module in Figure \ref{fig_unpertub}
shows the construction that prevents perturbator
from making any perturbation along the edge from $A$ to $B$.
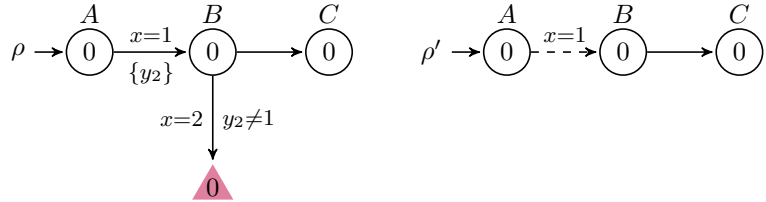
\begin{figure}[h]
\begin{center}

\begin{tikzpicture}[->,>=stealth',shorten >=1pt,auto,node distance=1cm,
  semithick,scale=0.9]
  
  \node at (-1.5, 1) (n) {$\mathbf{Module : Prevent \;perturbation}$};
  
  \node[initial,initial text={$\rho$}, locr] at (-1.6,0) (A) {$0$} ;
  \node[locLabel] () [above of = A]{$A$};
  \node[locr] at (0.2,0) (B){$0$ } ;
  \node[locLabel]() [above of = B]{$B$};
  \node[locr] at (1.9,0) (C){$0$};
  \node[locLabel] () [above of = C]{$C$};
  \node[triangle] at (0.2,-2) (T) {0};

  \draw[trans] (A) -- (B) node[midway,above] {$x{=}1$} node[midway,below]{$\set{y_2}$};
  \draw[trans] (B) -- (T) node[midway,left]{$x{=}2$} node[midway,right]{$y_2{\neq}1$};
  \draw[trans] (B) -- (C);

  \node[initial,initial text={$\rho'$}, locr] at (4.5,0) (A1) {$0$} ;
  \node[locLabel] () [above of = A1]{$A$};
  \node[locr] at (6.2,0) (B1){$0$ } ;
  \node[locLabel]() [above of = B1]{$B$};
  \node[locr] at (7.9,0) (C1){$0$};
  \node[locLabel] () [above of = C1]{$C$};
     
  \draw[trans,dashed] (A1) -- (B1) node[midway,above]{$x{=}1$} ;  
  \draw[trans] (B1) -- (C1);

\end{tikzpicture}

\caption{Prevent perturbation module: $x$ is some clock and $x = k$ could be the constraint for any $k\in \Nat$.  The triangle with $0$ represents target location $l$ with cost function $cf(l): \Rplus \to 0$. }
\label{fig_unpertub}
\end{center}
\end{figure}
In run $\rho$, the edge from $B$ to the  target ensures that
if the delay chosen at $A$ was perturbed then 
controller can achieve a cost 0.
For better readability,
we represent these unperturbed edges as 
dashed arrows as shown in path $\rho'$.
We note that the clock which is used in the equality constraint,
($x$ in Figure \ref{fig_unpertub})
cannot be reset along the same edge.
If we do not specify a clock 
that is being reset along the dashed edge,
we consider it to be $y_2$.
For any other clock, we show it as being reset along the dashed edge.
Note that in the `prevent perturbation module',
we need at least one equality constraint
($x=1$ in Figure \ref{fig_unpertub}),
thus ensuring a deterministic delay.

\subsection{Choice module} \label{sec:choice}
Since we consider a priced timed automaton and not a PTG,
perturbator does not own a location
from where it can suggest the successor location
of its choice.
We show in Figure \ref{fig_choice_module}, the construction of a module
that allows perturbator to choose the successor location.
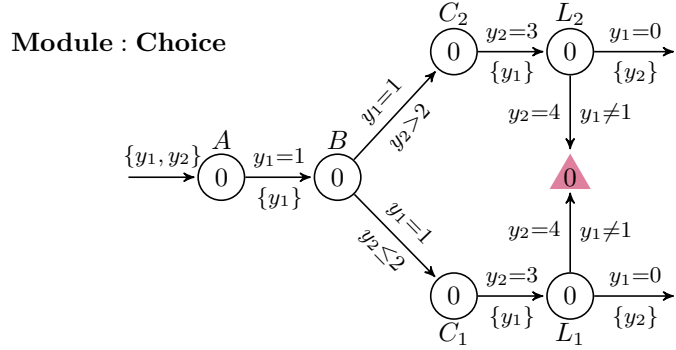
\begin{figure}[h]
\begin{center}

\begin{tikzpicture}[->,>=stealth',shorten >=1pt,auto,node distance=1cm,  semithick,scale=0.9]

  \node at (-1.5, 2) (n) {$\mathbf{Module : Choice}$};
  \node at (-1.5,0) (i) {};
  
  \node[locr] (A) {$0$} ;
  \node[locLabel] () [above of =A] {$A$};
   
   \node[locr] at (1.7,0) (B) {0};
   \node[locLabel] () [above of = B] {$B$};
  
  \node[locr] at (3.4,-1.75) (C1){$0$}; 
  \node[locLabel] () [below of = C1] {$C_1$};
  
  \node[locr] at (3.4,1.85)(C2){$0$};
  \node[locLabel] () [above of = C2] {$C_2$};
  
  \node[locr] at (5.1,-1.75)(L1){$0$};
  \node[locLabel] () [below of = L1] {$L_1$};
  
  \node[locr] at (5.1,1.85)(L2){$0$};
  \node[locLabel] () [above of = L2] {$L_2$};
  
  \node[triangle] at (5.1,0)(T){0};

  \node[] at (6.8,-1.75)(F1){};
  \node[] at (6.8,1.85)(F2){};
    
  \draw[trans] (i) -- (A)node[midway,above]{$\set{y_1,y_2}$};
  \draw[trans] (A) -- (B)node[midway,above]{$y_1{=}1$} node[midway,below]{$\set{y_1}$};
  \draw[trans] (B) -- (C1)node[midway,sloped,above]{$y_1{=}1$} node[midway,sloped,below]{$y_2{\leq}2$};
  \draw[trans] (B) -- (C2)node[midway,sloped,above]{$y_1{=}1$} node[midway,sloped,below]{$y_2{>}2$};
  \draw[trans] (C1) -- (L1)node[midway,above]{$y_2{=}3$} node[midway,below]{$\set{y_1}$};
  \draw[trans] (C2) -- (L2)node[midway,above]{$y_2{=}3$} node[midway,below]{$\set{y_1}$};
  \draw[trans] (L1) -- (T)node[midway,left]{$y_2{=}4$} node[midway,right]{$y_1{\neq}1$};
  \draw[trans] (L2) -- (T)node[midway,left]{$y_2{=}4$} node[midway,right]{$y_1{\neq}1$};

  \draw[trans] (L1) -- (F1)node[midway,above]{$y_1{=}0$} node[midway,below]{$\set{y_2}$};
  \draw[trans] (L2) -- (F2)node[midway,above]{$y_1{=}0$} node[midway,below]{$\set{y_2}$};

\end{tikzpicture}

\caption{Choice module : Perturbator can choose to go to $C_2$ if he peturbs the delay at $B$ by a positive value. If he does not perturb or perturbs by a negative value then goes to $C_1$.}
\label{fig_choice_module}
\end{center}
\end{figure}
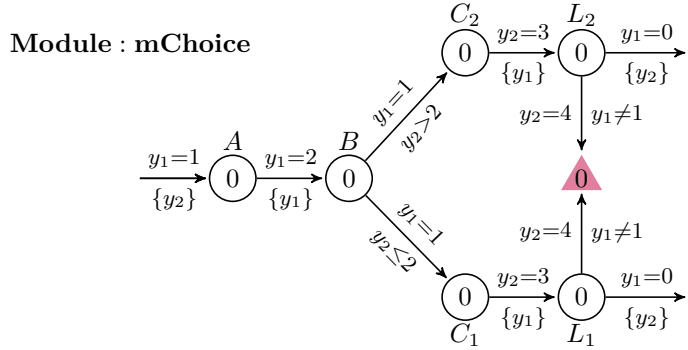
\begin{figure}[h]
\begin{center}

\begin{tikzpicture}[->,>=stealth',shorten >=1pt,auto,node distance=1cm,  semithick,scale=0.9]

  \node at (-1.5, 2) (n) {$\mathbf{Module : mChoice}$};
  \node at (-1.5,0) (i) {};
  
  \node[locr] (A) {$0$} ;
  \node[locLabel] () [above of =A] {$A$};
   
   \node[locr] at (1.7,0) (B) {0};
   \node[locLabel] () [above of = B] {$B$};
  
  \node[locr] at (3.4,-1.75) (C1){$0$}; 
  \node[locLabel] () [below of = C1] {$C_1$};
  
  \node[locr] at (3.4,1.85)(C2){$0$};
  \node[locLabel] () [above of = C2] {$C_2$};
  
  \node[locr] at (5.1,-1.75)(L1){$0$};
  \node[locLabel] () [below of = L1] {$L_1$};
  
  \node[locr] at (5.1,1.85)(L2){$0$};
  \node[locLabel] () [above of = L2] {$L_2$};
  
  \node[triangle] at (5.1,0)(T){0};
  
  \node[] at (6.8,-1.75)(F1){};
  \node[] at (6.8,1.85)(F2){};
    
  \draw[trans] (i) -- (A)node[midway,above]{$y_1{=}1$};
  \draw[trans] (i) -- (A)node[midway,below]{$\set{y_2}$};
  \draw[trans] (A) -- (B)node[midway,above]{$y_1{=}2$} node[midway,below]{$\set{y_1}$};
  \draw[trans] (B) -- (C1)node[midway,sloped,above]{$y_1{=}1$} node[midway,sloped,below]{$y_2{\leq}2$};
  \draw[trans] (B) -- (C2)node[midway,sloped,above]{$y_1{=}1$} node[midway,sloped,below]{$y_2{>}2$};
  \draw[trans] (C1) -- (L1)node[midway,above]{$y_2{=}3$} node[midway,below]{$\set{y_1}$};
  \draw[trans] (C2) -- (L2)node[midway,above]{$y_2{=}3$} node[midway,below]{$\set{y_1}$};
  \draw[trans] (L1) -- (T)node[midway,left]{$y_2{=}4$} node[midway,right]{$y_1{\neq}1$};
  \draw[trans] (L2) -- (T)node[midway,left]{$y_2{=}4$} node[midway,right]{$y_1{\neq}1$};

  \draw[trans] (L1) -- (F1)node[midway,above]{$y_1{=}0$} node[midway,below]{$\set{y_2}$};
  \draw[trans] (L2) -- (F2)node[midway,above]{$y_1{=}0$} node[midway,below]{$\set{y_2}$};

\end{tikzpicture}

\caption{mChoice module: The mChoice (modified choice) module is the
same as the Choice module except for the fact that
here the value of clock $y_1$ is 1 upon entry.}
\label{fig_mchoice_module}
\end{center}
\end{figure}
The delay from  location $A$ to location $B$
can be perturbed by perturbator.
Controller chooses $C_2$ as the successor if the perturbation is
positive, and chooses $C_1$ as its successor if the perturbation is negative.
We note that if the module was entered with $x_1= \alpha_1,
z=\beta,x_2=\alpha_2, y_1=y_2=0$ then
upon leaving either $L_1$ or $L_2$
the clock values are
$x_1= 3+\alpha_1,z=3+\beta, x_2=3+\alpha_2, y_1=y_2=0$.
The mChoice (modified choice) module 
shown in Figure \ref{fig_mchoice_module} is the
same as the Choice module except for the fact that
here the value of clock $y_1$ is 1 upon entry.
Thus the constraint on the edge between locations $A$ and $B$ is
$y_1 = 2$ instead of $y_1 = 1$ as in choice module.
Here also the value of clocks $x_1$, $z$ and $x_2$ are increased
by 3 as in the choice module
while clocks $y_1$ and $y_2$ have value 0 on exit.
 
\subsection{Restore module} \label{sec:restore}
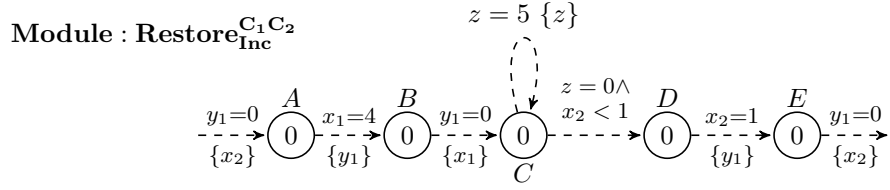
\begin{figure}[h]
\begin{center}

\begin{tikzpicture}[->,>=stealth',shorten >=1pt,auto,node distance=1cm,  semithick,scale=0.9]

\node at (-2,1.5) (n) {$\mathbf{Module : Restore_{Inc}^{C_1C_2}}$};

  \node at (-1.5,0) (i) {};
  
  \node[locr] (A) {$0$} ;
  \node[locLabel] () [above of =A] {$A$};
   
   \node[locr] at (1.7,0) (B) {0};
   \node[locLabel] () [above of = B] {$B$};

   \node[locr] at (3.4,0) (C) {0};
   \node[locLabel] () [below of = C] {$C$};
  
   \node[locr] at (5.5,0) (D) {0};
   \node[locLabel] () [above of = D] {$D$};
  
   \node[locr] at (7.4,0) (E) {0};
   \node[locLabel] () [above of = E] {$E$};
  
  \node at (8.9,0)(F){};

  \draw[trans,dashed] (i) -- (A)node[midway,above]{$y_1{=}0$}node[midway,below]{$\set{x_2}$};
  \draw[trans,dashed] (A) -- (B)node[midway,above]{$x_1{=}4$} node[midway,below]{$\set{y_1}$};
  \draw[trans,dashed] (B) -- (C)node[midway,sloped,above]{$y_1{=}0$}
  node[midway,below]{$\set{x_1}$};
  \path (C) edge [dashed,loop above, min distance=15mm,,looseness=10] node {$z=5$~$\set{z}$} (C);
  
  \draw[trans,dashed] (C) -- (D)node[midway,above]{$\begin{array}{c} z =0 \wedge \\ x_2 < 1\end{array}$};
  \draw[trans,dashed] (D) -- (E)node[midway,above]{$x_2{=}1$} node[midway,below]{$\set{y_1}$};
  \draw[trans,dashed] (E) -- (F)node[midway,above]{$y_1{=}0$} node[midway,below]{$\set{x_2}$};
  
\end{tikzpicture}

\caption{Restore module :
This is actually a group of four modules.
$Restore_{Inc}^{C_1C_2}$ is shown in the figure.}
\label{fig_restore}
\end{center}
\end{figure}
Both choice and mChoice modules
add a shift of 3 to the clock values $x_1, x_2$ and $z$.
Since the main modules simulating
increment and decrement of the counters
expect the values to be in their normal forms,
we need to remove the shift of 3; this is achieved by the Restore module
shown in Figure \ref{fig_restore}.
The restore modules used in the main modules simulating the
operations on counter $C_1$
are a group of four different modules
as mentioned below.
$Restore_{Inc}^{C_1C_2}$ denotes the module used as part of the
Increment module for counter $C_1$.
We also similarly have $Restore_{Dec}^{C_1C_2}$ module
which is used as part of the Decrement module.
The $Restore_{Dec}^{C_1C_2}$ module is similar to the
$Restore_{Inc}^{C_1C_2}$ module with the only difference being
that the clock constraint on the loop on $C$ is $z = 6$ instead of
$z = 5$ as in the $Restore_{Inc}^{C_1C_2}$ module.
$C_1C_2$ here denotes that the fractional part of clock $x_1$ is more
than the fractional part of clock $z$.
We also use $Restore_{Dec}^{C_2C_1}$ and $Restore_{Inc}^{C_2C_1}$
to denote that the fractional part of clock 
$z$ is more than the fractional part of clock $x_1$.
The $Restore_{Inc}^{C_2C_1}$ module can be obtained from $Restore_{Inc}^{C_1C_2}$ module by
replacing all the occurrences of clock $x_1$ with clock $z$ and
replacing all the occurrences of clock $z$ with clock $x_1$.
$Restore_{Dec}^{C_2C_1}$ can also be obtained from
$Restore_{Dec}^{C_1C_2}$ in the same way.
The edge from location $C$ to location $D$ forces controller
to take the loop at location $C$ only once.
The $Restore_{Inc}^{C_1C_2}$ and $Restore_{Inc}^{C_2C_1}$
modules are entered  with
clock values $x_1 = 3 + \frac{1}{2^i}+\varepsilon_1, z = 4+\frac{1}{2^j}+\varepsilon_2, x_2=y_1=y_2=0$, at the starting location $A$ of  the module. At location $E$,  the clock
values are $x_1 = \frac{1}{2^i}+\varepsilon_1, z = \frac{1}{2^j}+\varepsilon_2, x_2=y_1=y_2=0$, i.e.
restored to their normal form.

The restore modules used in the modules simulating operations
on counter $C_2$ are analogous.
Corresponding to $Restore_{Inc}^{C_1C_2}$,
the delays at locations $A$, $C$ and $D$ are
$1-\frac{1}{2^i}-\varepsilon_1$,
$\frac{1}{2^i}+\varepsilon_1 - \frac{1}{2^j}-\varepsilon_2$
and $\frac{1}{2^j}+\varepsilon_2$ respectively,
while the delays at locations $B$ and $E$ are 0.
The value of clock $z$ at the entry of the $Restore_{Dec}^{C_1C_2}$
and $Restore_{Dec}^{C_2C_1}$ is $5+\frac{1}{2^j}+\varepsilon_2$
and the the clock values at the exit are as
$Restore_{Inc}^{C_1C_2}$ and $Restore_{Inc}^{C_2C_1}$ modules.

We show below the main modules which are
used for simulating  zero test and decrement.
We show here the modules corresponding to
the operations on counter $C_1$.
The modules corresponding to the operations on counter $C_2$
are analogous.

\subsection{Decrement module} \label{sec:decre}
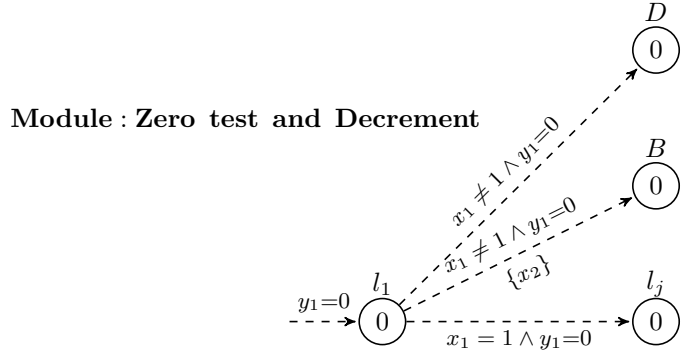
\begin{figure}[h]
\begin{center}

\begin{tikzpicture}[->,>=stealth',shorten >=1pt,auto,node distance=1cm,  semithick,scale=0.9]

  \node at (-3.5,3) (n) {$\mathbf{Module :Zero \;\:test\;\: and\;\: Decrement}$};

  \node at (-3,0) (i) {};
  
  \node[locr] at (-1.5,0) (l1) {$0$} ;
  \node[locLabel] () [above of =l1] {$l_1$};
  
  \node[locr] at (2.5,0) (c1) {$0$};
  \node[locLabel] () [above of =c1] {$l_j$};

  \node[locr] at (2.5,2) (B) {$0$};
  \node[locLabel] () [above of =B] {$B$};  

  \node[locr] at (2.5,4) (D) {$0$};  
  \node[locLabel] () [above of =D] {$D$};

  


  

  \draw[trans,dashed] (i) -- (l1)node[midway,above]{$y_1{=}0$};
  \draw[trans,dashed] (l1) -- (c1)node[midway,below]{$x_1 = 1 \wedge y_1{=}0$};
  \draw[trans,dashed] (l1) -- (B)node[midway,sloped,above]{$x_1 \neq1 \wedge y_1{=}0$}node[midway,sloped,below]{$\set{x_2}$};
  \draw[trans,dashed] (l1) -- (D)node[midway,sloped,above]{$x_1 \neq1 \wedge y_1{=}0$};
  
\end{tikzpicture}

\caption{Zero test and Decrement Module: This module simulates the instruction \texttt{If $(c_1=0)$ go to $l_j$ else go to $l_j'$}.
The extensions from $B$ and $D$
are shown in Figures \ref{fig_decrement}
and \ref{fig_decrement1to0} respectively.
}
\label{fig_testndec}
\end{center}
\end{figure}
\begin{figure}[h]
\begin{center}

\begin{tikzpicture}[->,>=stealth',shorten >=1pt,auto,node distance=1cm,  semithick,scale=0.9]

\node at (-2,1.5) (n) {$\mathbf{Module : Decrement~C_1}$};

  \node at (-4.5,0) (i) {};
  
  \node[locr] at (-3,0) (l1) {$0$} ;
  \node[locLabel] () [above of =l1] {$B$};

  \node[locr] at (-1.5,0) (l2) {$0$} ;
  \node[locLabel] () [above of =l2] {$C$};

  \node[locr] at (0,0) (l3) {$0$} ;
  \node[locLabel] () [above of =l3] {$D$};

  \node[locr] at (1.5,0) (l4) {$0$} ;
  \node[locLabel] () [above of =l4] {$E$};
  
  \node[widget] at (4,0) (c) {$mChoice$};

  \node[locr] at (6.5,0) (l5) {$0$} ;
  \node[locLabel] () [above of =l5] {$F$};

  \node[widget] at (9,1) (r) {$Restore_{Dec}^{C_1C_2}$};
  \node[widget] at (9,-1) (s) {$Restore_{Dec}^{C_2C_1}$};
  
  \node[locr] at (11.5,0) (l6) {0};
  \node[locLabel] () [above of =l6] {$l_j'$};

  \draw[trans,dashed] (i) -- (l1)node[midway,above]{$\begin{array}{c} x_1 \neq 1 \\ \wedge y_1 = 0\end{array}$}
  node[midway,below]{$\set{y_2}$};
  \draw[trans,dashed] (l1) -- (l2)node[midway,above]{$x_1{=}1$}node[midway,below]{$\set{x_2}$};
  \draw[trans,dashed] (l2) -- (l3)node[midway,above]{$y_2{=}1$}node[midway,below]{$\set{y_1}$};  
  \draw[trans] (l3) -- (l4)node[midway,above]{$y_1 \le 1$}node[midway,below]{$\set{x_1}$};  
  \draw[trans,dashed] (l4) -- (c)node[midway,above]{$y_1=1$} node[midway,below]{$\set{y_2}$};
  \draw[trans,dashed] (c) -- (l5)node[midway,above]{$y_1=0$} node[midway,below]{$\set{x_2,y_2}$};
  \draw[trans,dashed] (l5) -- (r)node[midway,sloped,above]{$y_1{=}0$};
  \draw[trans,dashed] (l5) -- (s)node[midway,sloped,above]{$y_1{=}0$};
  \draw[trans,dashed] (r) -- (l6)node[midway,sloped,above]{$y_1{=}0$};
  \draw[trans,dashed] (s) -- (l6)node[midway,sloped,above]{$y_1{=}0$};
  
   \node[widget] at (4,-2) (c2) {$Choice$};

   \node[widget] at (2.5,-4) (F) {$Test~Dec^{C_1}_{<}$};

    \node[widget] at (5.5,-4) (G) {$Test~Dec^{C_1}_{>}$};

  \draw[trans,dashed] (c) -- (c2)node[midway,left]{$y_1{=}0$};
  \draw[trans,dashed] (c2) -- (F)node[midway,sloped,above]{$y_1{=}0$};
  \draw[trans,dashed] (c2) -- (G)node[midway,sloped,above]{$y_1{=}0$};
  
  
  \node at (-3,-6) (j) {$\mathbf{Test~Dec^{C_1}_{>}}$};
  
%
%
%
%
%
%
  
  \node at (-1.5,-6) (j) {};
  
  \node[locr] at (0,-6) (L) {$2$};
  \node[locLabel] () [above of =L] {$A'$};
   
   \node[locr] at (1.7,-6) (M) {0};
   \node[locLabel] () [above of = M] {$B'$};

   \node[locr] at (3.4,-6) (N) {0};
   \node[locLabel] () [above of = N] {$C'$};
  
   \node[locr] at (5.2,-6) (O) {1};
   \node[locLabel] () [above of = O] {$D'$};
  
   \node[triangle] at (6.9,-6) (P) {0};

  \draw[trans,dashed] (j) -- (L)node[midway,above]{$y_1{=}0$};
  \draw[trans,dashed] (L) -- (M)node[midway,above]{$x_2{=}8$};
 \draw[trans,dashed] (L) -- (M)node[midway,below]{$x_2{=}0$};
    \draw[trans,dashed] (M) -- (N)node[midway,above]{$y_1{=}1$};
    \draw[trans,dashed] (M) -- (N)node[midway,below]{$\{y_1\}$};
    \draw[trans,dashed] (N) -- (O)node[midway,above]{$x_1{=}8$};
    \draw[trans,dashed] (O) -- (P)node[midway,above]{$y_1{=}1$};

\end{tikzpicture}

\caption{\textbf{Decrement $C_1$ module} : 
The section of the module
shown in Figure \ref{fig_testndec}
starting from location $B$.
This section is used if $c_1 > 1$ before being decremented.
It keeps the fractional part of clock $z$ unchanged.
The price 2 at $A$ is a shorthand, and can be replaced with 1
on having a longer sequence of transitions.
}
\label{fig_decrement}
\end{center}
\end{figure}
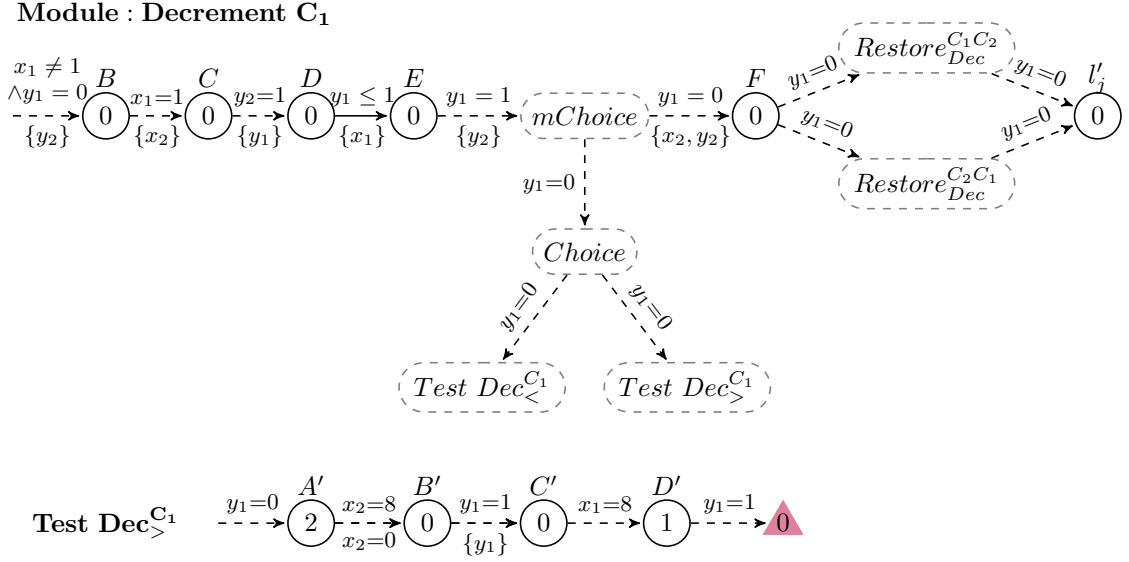
The module simulating decrement of counter $C_1$
is shown in Figure \ref{fig_testndec}.
Recall that by the normal form, the values 
of the clocks are $x_1=\frac{1}{2^i}+\varepsilon_1$,
$z=\frac{1}{2^j}+\varepsilon_2$, $y_1=y_2=x_2=0$
at $l_1$. 
\begin{enumerate}
\item Assume that $c_1>0$ at $l_1$. Controller can choose to goto $B$ or $D$, since the constraints on both the edges are the same. If $c_1>1$, controller chooses to goto $B$, and if $c_1=1$, then controller goes to $D$.
Consider $c_1>1$, and controller visiting $B$. By the encoding, $x_1=\frac{1}{2^i}+\varepsilon_1$, $i>1$,
$z=\frac{1}{2^j}+\varepsilon_2, x_2=y_1=y_2=0$.
Here $\varepsilon_1$ and $\varepsilon_2$ denote 
errors accumulated so far in clocks $x_1$ and $z$
due to perturbation made by perturbator so far.
Figure \ref{fig_decrement} shows the section of the module
shown in Figure \ref{fig_testndec}
starting from location $B$.
This section simulates the decrement of counter $C_1$
when the value of the counter is greater than 1.
The value of clock $z$ simulating counter $C_2$
remains unchanged.

Let us denote the value of $x_1$ at the entry of the module Decrement $C_1$ in Figure \ref{fig_decrement},
i.e. $\frac{1}{2^i}+\varepsilon_1$ by $\alpha$.
Thus the delays at locations $B$ and $C$ are respectively
$1 - \alpha$ and $\alpha$. 
On entry at $D$, we thus have $x_2=\alpha, y_1=0,y_2=1$. 
A non-deterministic time $t$ is spent at $D$ simulating the decrement of $C_1$. Ideally, 
$t$ must be $1-2\alpha$.
Perturbator can perturb it by $\delta$,
where $\delta$ can be both positive and negative
and clock $x_1$ is reset.
On entering $E$ we thus have $x_1=0,y_1=t+\delta, x_2=\alpha+t+\delta$. 
At the entry to  mChoice module, the values of the clocks are
$x_1=1-t - \delta, z=2+\frac{1}{2^j}+\varepsilon_2, x_2=1+ \alpha, y_1=1, y_2=0$.
To correctly decrement $C_1$ (whose value is $i$),
$1-t$ should be exactly $2\alpha$, i.e. $\frac{1}{2^{i-1}}+2\varepsilon_1$.

Perturbator uses the mChoice module to either continue the simulation (by going to the Restore module) or  verifies controller's delay $t$. 
Due to the mChoice module, the clock values are $x_1=4-t - \delta,
z=5+\frac{1}{2^j}+\varepsilon_2, x_2=4+ \alpha, y_1=0, y_2=0$.
If perturbator chooses to continue the simulation,
then the Restore module restores the clocks back to normal form and
hence upon entering $l_j'$ the clock values are $x_1=1-t- \delta, z=\frac{1}{2^j}+\varepsilon_2, x_2=0, y_1=0, y_2=0$.
Thus, we have $x_1=\frac{1}{2^{i-1}}+2\varepsilon_1-\delta$, where
$2\varepsilon_1-\delta$ is the value due to the perturbations so far.

However, if perturbator chooses to verify,
he first goes to yets another Choice module.  If 
$1-t > 2\alpha$, then  the module
$Test~Dec^{C_1}_{>}$  is used and if $1-t < 2 \alpha$, then the module 
$Test~Dec^{C_1}_{<}$ is used. 
Note that due to the two Choice modules one after the other,  
the clock values upon entering $Test~Dec^{C_1}_{<}$
or $Test~Dec^{C_1}_{>}$ are
$x_1= 7-t - \delta, x_2=7+\alpha, y_1=y_2=0$. 


$\mathbf{Test~Dec^{C_1}_{>}}$ :
At $A'$, on entry we have 
$x_1= 7-t -\delta, x_2=7+\alpha, y_1=y_2=0$. 
A time $1-\alpha$ is spent at $A'$ with accumulated cost $2-2\alpha$.\footnote{The price 2 
on $A$ can be replaced with 1, by having a slightly longer sequence of transitions}
On entry to $B$, we have $x_1=8-\alpha-t-\delta, y_1=1-\alpha$. A time 
$\alpha$ is spent at $B'$, and $x_1=8-t-\delta$.
A time $t+\delta$ is spent at $C'$, obtaining 
$y_1=t+\delta$. A time $1-t-\delta$  
is spent at $D'$ obtaining the 
 accumulated cost
$2-2\alpha+1-t-\delta$. The target is reached with this cost. 
If $1-2\alpha > t$, then this is $>2-\delta$. 
The perturbator can choose $\delta <0$, making this cost $>2$.  

\item Controller chooses the outgoing edge to $D$ in Figure \ref{fig_testndec}
if $c_1$ is 1 in which case the decremented value is 0
which is encoded by the exact value $x_1 = 1$.
 The module from $D$ has been shown in Figure \ref{fig_decrement1to0}.

Figure \ref{fig_decrement1to0} shows the section of the module
of Figure \ref{fig_testndec}
starting from location $D$.
This section simulates the decrement of counter $C_1$
when $c_1=1$.
Upon entering $D$, in the Test and Decrement module, the clock values are $x_1= \frac{1}{2}+\varepsilon_1,
z=\frac{1}{2^j}+\varepsilon_2, x_2=y_1=y_2=0$.
Let $\alpha$ denote the value of $x_1$, i.e. $\frac{1}{2}+\varepsilon_1$.
The time elapsed in locations $D, E$ and $F$ 
in Figure \ref{fig_decrement1to0}
are respectively
$1-\alpha$, $\alpha$ and $1$.
At the entry of the Choice module, the clock values are
$x_1=1,z=2+\frac{1}{2^j}+\varepsilon_2,x_2=1+\alpha,y_1=y_2=0$.
Here $x_1$ encodes the counter value of $C_1$ exactly
and perturbator cannot perturb the delay made by the controller.

Perturbator uses the Choice module to either continue the simulation or it can verify the delay made by controller. 
Due to the Choice module, the clock values are $x_1=4,z=5+\frac{1}{2^j}+\varepsilon_2,x_2=4+\alpha,y_1=y_2=0$.
If perturbator chooses to continue the simulation
then the Restore module restores  the clocks back to the normal form and
hence upon entering $l_j'$ the clock values are $x_1 = 1,
z_1=\frac{1}{2^j}+\varepsilon_2,
x_2=y_1=y_2=0$.

However, if perturbator chooses to verify,
he uses $Test~Dec^{C_1}_{n=1}$ module to verify whether
controller chose this branch ($D$) of the Test and Decrement module
when $c_1=1$ or $c_1 > 1$.

$\mathbf{Test~Dec^{C_1}_{n=1}}$ : 
On entry, we have $x_1=7,z=8+\frac{1}{2^j}+\varepsilon_2,x_2=7+\alpha,y_1=y_2=0$.
The delays at locations are: at $A':1-\alpha$ obtaining $y_1=1-\alpha$ 
on entering $B'$. A time elapse of $\alpha$ at $B'$ gives $x_1=\alpha$. Finally, at $C'$, we elapse $1-\alpha$.
  Thus the cost incurred in this module is
$3-3\alpha$. For $c_1=1$, this is $3-\frac{3}{2}-2\varepsilon_1=2-\frac{1}{2}-2\varepsilon_1$, and 
the minimum cost when $c_1>1$
is  $3-3\frac{1}{2^2}-2\varepsilon_1=2+\frac{1}{4}- 2\varepsilon_1$.
In the limit, as $\varepsilon_1$ tends to 0, the cost is $\leq 2$  if controller chose the correct branch, that is, chose
$D$ when $c_1=1$.

\item Suppose controller chooses $B$ instead of $D$ when $c_1=1$.
Then the value of  clock $x_1$ after simulating the
decrement operation will not be exact, i.e.will not be  equal to 1.
Now, if the next instruction involving controller $C_1$
is also a zero test and decrement operation,
then controller will incorrectly move to $l_j'$
instead of $l_j$ 
while simulating this next zero test and decrement operation.
For choosing $B$ instead of $D$,
controller will be punished while
simulating this next zero test and decrement operation.
Since the value of clock $x_1$ is not 1,
while simulating this next zero test and decrement operation,
controller will either go to $B$ or $D$
in the module in Figure \ref{fig_testndec}.
\begin{itemize}
\item If $B$ is chosen, $t$ should equal $1-2\alpha$
for correct simulation.
Now $\alpha$ being $1 + \varepsilon_1$,
controller cannot delay for $1-2\alpha$
at location $D$ of Figure \ref{fig_decrement}
and hence is punished.
\item If controller goes to location $D$ in Figure \ref{fig_testndec},
when $c_1 = 0$, then $x_1=1+\varepsilon_1=\alpha$
then perturbator moves to the module $Test~Dec^{C_1}_{n=0}$.
If $\varepsilon_1>0$, then the controller will get stuck 
in the transition from $D$ to $E$ (see Figure  \ref{fig_decrement1to0}) and if $\varepsilon_1<0$, then 
the module $Test~Dec^{C_1}_{n=0}$ in Figure  \ref{fig_decrement1to0}
incurs a cost of $2-2\varepsilon_1 >2$.
The module $Test~Dec^{C_1}_{n=0}$ can be drawn similar to 
$Test~Dec^{C_1}_{n=1}$.
\end{itemize}
\end{enumerate}


\begin{figure}[h]
\begin{center}

\begin{tikzpicture}[->,>=stealth',shorten >=1pt,auto,node distance=1cm,  semithick,scale=0.9]

\node at (-2,1.5) (n) {$\mathbf{Module : Decrement~C_1~from~1~to~0}$};

  \node at (-4.5,0) (i) {};
  
  \node[locr] at (-3,0) (l1) {$0$} ;
  \node[locLabel] () [above of =l1] {$D$};

  \node[locr] at (-1.5,0) (l2) {$0$} ;
  \node[locLabel] () [above of =l2] {$E$};

  \node[locr] at (0,0) (l3) {$0$} ;
  \node[locLabel] () [above of =l3] {$F$};

  
  \node[widget] at (2.5,0) (c) {$Choice$};

  \node[locr] at (5,0) (l4) {$0$} ;
  \node[locLabel] () [above of =l4] {$G$};

  \node[widget] at (7.5,1) (r) {$Restore_{Dec}^{C_1C_2}$};
  \node[widget] at (7.5,-1) (s) {$Restore_{Dec}^{C_2C_1}$};
  
  \node[locr] at (10,0) (l5) {0};
  \node[locLabel] () [above of =l5] {$l_j'$};

  \draw[trans,dashed] (i) -- (l1)node[midway,above]{$\begin{array}{c} x_1 \neq 1 \\ \wedge y_1 = 0\end{array}$}
  node[midway,below]{$\set{y_2}$};
  \draw[trans,dashed] (l1) -- (l2)node[midway,above]{$x_1{=}1$}node[midway,below]{$\set{x_2}$};
  \draw[trans,dashed] (l2) -- (l3)node[midway,above]{$y_1{=}1$}node[midway,below]{$\set{x_1}$};  
  \draw[trans,dashed] (l3) -- (c)node[midway,above]{$x_1=1$} node[midway,below]{$\set{y_1}$};
  \draw[trans,dashed] (c) -- (l4)node[midway,above]{$y_1=0$} node[midway,below]{$\set{x_2}$};
  \draw[trans,dashed] (l4) -- (r)node[midway,sloped,above]{$y_1{=}0$};
  \draw[trans,dashed] (l4) -- (s)node[midway,sloped,above]{$y_1{=}0$};
  \draw[trans,dashed] (r) -- (l5)node[midway,sloped,above]{$y_1{=}0$};
  \draw[trans,dashed] (s) -- (l5)node[midway,sloped,above]{$y_1{=}0$};
  
  \node[widget] at (2.5,-2) (c2) {$Choice$};

  \node[widget] at (1,-4) (F) {$Test~Dec^{C_1}_{n=1}$};

  \node[widget] at (4,-4) (G) {$Test~Dec^{C_1}_{n=0}$};

  \draw[trans,dashed] (c) -- (c2)node[midway,right]{$y_1{=}0$};
  \draw[trans,dashed] (c2) -- (F)node[midway,sloped,above]{$y_1{=}0$};
  \draw[trans,dashed] (c2) -- (G)node[midway,sloped,above]{$y_1{=}0$};
  
  
  \node at (-3,-6) (j) {$\mathbf{Test~Dec^{C_1}_{n=1}}$};
  
  \node at (-1.5,-6) (j) {};
  
  \node[locr] at (0,-6) (L) {$1$};
  \node[locLabel] () [above of =L] {$A'$};
   
   \node[locr] at (1.7,-6) (M) {0};
   \node[locLabel] () [above of = M] {$B'$};
   
   \node[locr] at (3.4,-6) (N) {1};
   \node[locLabel] () [above of = N] {$C'$};
   
   \node[locr] at (5.2,-6) (N1) {0};
   \node[locLabel] () [above of = N1] {$D'$};

   \node[locr] at (6.9,-6) (N2) {1};
   \node[locLabel] () [above of = N2] {$E'$};

   \node[triangle] at (8.6,-6) (O) {$0$};
  
  
  \draw[trans,dashed] (j) -- (L)node[midway,above]{$y_1{=}0$};
  \draw[trans,dashed] (L) -- (M)node[midway,above]{$x_2{=}8$};
  \draw[trans,dashed] (M) -- (N)node[midway,above]{$y_1{=}1$};
  \draw[trans,dashed] (M) -- (N)node[midway,below]{$\{y_1\}$};
    \draw[trans,dashed] (L) -- (M)node[midway,below]{$\{x_1\}$};
\draw[trans,dashed] (N) -- (N1)node[midway,above]{$x_1{=}1$};
\draw[trans,dashed] (N) -- (N1)node[midway,below]{$\{x_1\}$};
    \draw[trans,dashed] (N1) -- (N2)node[midway,above]{$y_1{=}1$};
  \draw[trans,dashed] (N1) -- (N2)node[midway,below]{$\{y_1\}$};
     \draw[trans,dashed] (N2) -- (O)node[midway,above]{$x_1{=}1$};


\end{tikzpicture}

\caption{\textbf{Decrement $C_1$ from 1 to 0 module} :
The section of the module
shown in Figure \ref{fig_testndec}
starting from location $D$.
This section is used if $c_1 =1$ before being decremented.
It keeps the fractional part of clock $z$ unchanged.
}
\label{fig_decrement1to0}
\end{center}
\end{figure}
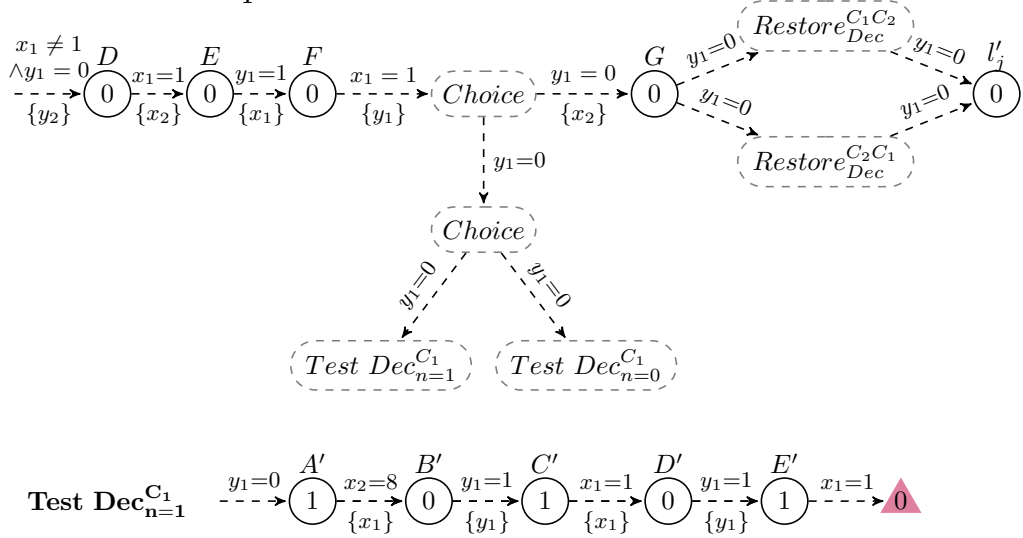
\subsection{Complete Reduction}
The entire reduction consists of constructing a module
corresponding to each instruction $I_i$, $1 \le i \le n$,
of the two-counter machine.
The first location of the module corresponding to instruction $I_1$
is the initial location.
We simulate the halting instruction $I_n$
by a target location whose cost function assigns  2 to all clock values.
We denote the robust timed automaton simulating the two
counter machine by $\mathcal{A}$,
$s$ is the initial state $(l,\zero,\zero)$.
\begin{lemma} \label{lem:reduction}
The two counter machine halts if and only if there
is a strategy $\sigma$ of controller
such that $limcost_{\mathcal{A}}(\sigma, s) \le 2$.
\end{lemma}
\proof
We first consider the case when the two counter machine halts.
Suppose it halts in $m$ steps.
The cost incurred in $m$ steps can be due to
reaching one of the target states in a test module
or reaching the \emph{halt} instruction in $m$ steps.
We consider an $\varepsilon$ such that
$0 < 3^m \delta < \varepsilon$.
In the first case, the cost is less than or equal to $2+2 \varepsilon_b$,
where by Lemma \ref{lem:ocan}, $\varepsilon_b \le \varepsilon$ 
and hence the cost is 2 in the limit.
In the second case, controller simulates the two
counter machine faithfully and reaches the target location
corresponding to the halt instruction
and hence the cost is 2 in the limit.

Now we consider the case when the two counter machine does not halt.
Controller can simulate the two counter machine
using the increment and the
zero test and decrement modules
corresponding to each of the instructions.
The cost is $\infty$ if controller
simulates the instructions faithfully
and a target state is not reached.
On the other hand, if controller makes an error,
then it will be punished by perturbator in
one of the test modules and cost will be non-zero.
Hence the proof.
\qed

Given an accumulated delay $\varepsilon$,
the accumulated delay after one step
due to the decrement and the increment modules
are $2\varepsilon+\delta_1$ and $\varepsilon /2 + \delta_2$ respectively.
The following lemma is from \cite{ocan}.
\begin{lemma} \label{lem:ocan}\cite{ocan}
Consider the two functions $f: x \to 2x + 1$
and $g: x \rightarrow x/2 + 1$.
For any $n \ge 1, x > 0$, and any $f_1 , \dots , f_n \in \{f, g\}, 
f_1 \circ f_2 \circ \dots \circ f_n(x) \le 3^n x$.
\end{lemma}
We note that the prices used in all the modules in our reduction
are only $\{0,  1\}$
and hence we have the undecidability result
as given by Theorem \ref{thm:robust_undec}.

\section{Proof of Lemma \ref{lem:rtg_to_ptg}}
\label{app:tran1}
We first map the states of the RPTG $\rtg$ 
 and the dwell-time PTG $\ptg$.  
 Let $\mathcal{S}(\mathcal{\rtg})$ denote the set of states of the form $(l, \nu)$  as well as 
 $((l, \nu),t,e)$ 
 of the RPTG and 
 $\mathcal{S}(\ptg)$ denote the set of states of the dwell-time PTG. 
\begin{definition}[state map]
 We define a State Map $\sm : \nSR{(\rtg)} \to \nSR{(\ptg)}$ as follows
 \begin{itemize}
  \item if $l$ is a controller(perturbator) location then $\sm(l,\val) = (l,\val)$ as 
  all controller locations of $\rtg$  become player 1 locations in the dwell-time PTG $\ptg$, and 
  all the perturbator locations of $\rtg$  become player 2 locations in the dwell-time PTG $\ptg$;
\item Recall that the RPTG had states of the form $((l, \nu),t,e)$ corresponding to 
perturbator states (after controller chose a time delay and edge,  perturbator decides the perturbation). Recall also that 
for every controller location $l$, and corresponding 
 edge choice $e$ made in the RPTG $\rtg$, we had the urgent player 2 location $(l,e)$ 
 immediately following the player 1 location $l$ in  the dwell-time-PTG $\ptg$ constructed.
   That is,  $\sm((l,\val),t,e) = (\urg{l,e},\val+t-\delta)$
 \end{itemize}
\end{definition}
Note that $\sm(s)$ is a unique state in $\ptg$. 

   \begin{figure}[h]
 \begin{center}
\begin{tikzpicture}[->,>=stealth',shorten >=1pt,auto,node distance=1cm,  semithick,scale=0.9]

  \node[min] at (-3,0) (A) {$k$};
  \node[locLabel] () [above of =A] {$A$};
  \node[cdel] () [below of=A] {$t$};

  \node[min] at (0,0) (B) {$k'$};
  \node[locLabel] () [above of =B] {$B$};

  \draw[trans] (A) -- (B) node[midway,above] {$e$} node[midway,below] {$g,r$};
  \node at (-1.5,-1) {RPTG $\rtg$};

  \node at (4,-1) {dwell time PTG $\ptg$};
\node[min] at (3,0) (Ap) {$k$};
  \node[locLabel] () [above of =Ap] {$A$};
  \node[cdel] () [below of=Ap] {$t-\delta$};

  \node[max] at (5,0) (mu) {$0$};
  \node[locLabel] () [above of =mu] {$\urg{A,e}$};
  \node[del] () [below of=mu] {$0$};

  \node[max] at (7,-1) (mp) {$k$};
  \node[locLabel] () [above of =mp] {$\posP{(A,e)}$};
  \node[del] () [below of=mp] {$[\delta,2\delta]$};

  \node[max] at (7,1) (mn) {$k$};
  \node[locLabel] () [above of =mn] {$\negP{(A,e)}$};
  \node[del] () [below of=mn] {$[0,\delta]$};

  \node[min] at (9,0) (Bp) {$k'$};
  \node[locLabel] () [above of =Bp] {$B$};

  \draw[trans] (Ap)--(mu) node[midway,above] {$g'$};
    \draw[trans] (mu)--(mp) node[midway,above] {$ep$};
  \draw[trans] (mu)--(mn) node[midway,above] {$en$};
  \draw[trans] (mp)--(Bp) node[midway,above] {$epo$};
  \draw[trans] (mp)--(Bp) node[midway,below] {$r$};
  \draw[trans] (mn)--(Bp) node[midway,below] {$r$};
  \draw[trans] (mn)--(Bp) node[midway,above] {$eno$};
%
%
%
%
%
%
    
\end{tikzpicture}
 \end{center}
 \caption{Transitions of RPTG $\rtg$ mapped to transitions in the constructed dwell-time PTG $\ptg$}
 \label{rta-ptg}
 \end{figure}
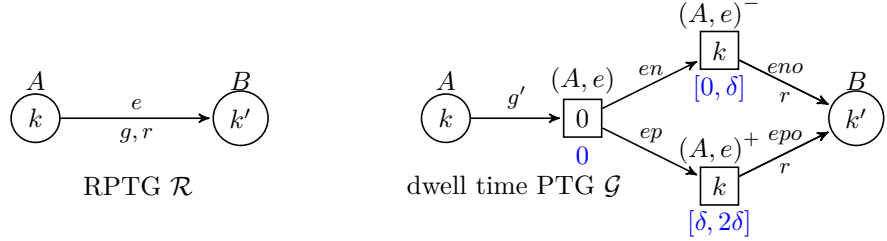

\begin{lemma}
 \label{lem:path_map}
 Given a path $\rho$ in $\rtg$ from $s$ to $s'$, there exists a unique 
 path $\rho'$ in $\ptg$ from $\sm(s)$ to $\sm(s')$. Additionally, $\costp(\rho) = \costp(\rho')$. 
\end{lemma}
The proof is quite straight forward and follows from the structure and the state map defined above. 

Next, given a strategy $\strat_1$ in $\rtg$, we shall define an equivalent strategy $\strat_1'$ in $\ptg$ in terms of the moves proposed.  Let $e$ be the edge from $l$ to $l'$ in $\rtg$. We map
$\strat(\rho.s) = (t,e)$ to $\strat'(\rho'.\sm(s)) = (t',e')$ as follows 
\begin{enumerate}
 \item \textbf{Controller strategy mapped to Player~1 strategy} : \newline 
 The strategy $\strat_1(\rho.(l,\val)) = (t,e)$ in $\rtg$ leads to the state $((l,\val),t,e)$. 
 Corresponding to this, we have 
  $\strat_1'(\rho'.\sm(l,\val)) = (t',e')$ such that $t' = t-\delta$ \footnote{$t\geq \delta$ in the $\rtg$ due to robust semantics} and the player 1 location $l$ moves into the urgent player 2 location   
  $\urg{l,e}$. This leads to $(\urg{l,e},\val+t-\delta)$.
  $e'$ is the edge in $\ptg$ between $l$ and $\urg{l,e}$.
    Recall also 
  that the time delay $t$ in $\rtg$ has been mapped to the time delay $t-\delta$ in the 
  constructed PTG $\ptg$.
 
 \item \textbf{Perturbator strategy mapped to Player~2 strategy for perturbator locations}: \newline
 A strategy $\strat_2(\rho.(l,\val)) = (t,e)$ in $\rtg$ leads to $(l',\val+t[r:=0])$. 
 Correspondingly,  we have in $\ptg$, $\strat_2'(\rho'.(l,\val)) = (t,e)$, 
 giving the state $(l',\val+t[r:=0])$ in $\ptg$.

 \item \textbf{Perturbator strategy mapped to Player~2 strategy for new locations} : \newline
 Recall that we have  $\sm((l,\val),t,e) = (\urg{l,e},\val+t-\delta)$. \newline
 If we have the strategy $\strat_2(\rho.((l,\val),t,e)) = \epsilon \in [-\delta,\delta]$ in $\rtg$ such that 
 \begin{itemize}
  \item if \boldmath{$0 \leq \epsilon \leq \delta$}\unboldmath ~then $\strat_2'(\rho'.(\urg{l,e},\val+t-\delta)) = (0,\ep)$ which results in $(\posP{(l,e)},\val+t-\delta)$ and $\strat_2'(\rho'.(\posP{(l,e)},\val+t-\delta)) = (\delta+\epsilon, \epo)$ resulting in $(l', \val + t  - \delta + \delta + \epsilon [r:=0])$. 
  
  \item if \boldmath${-\delta \leq \epsilon < 0}$\unboldmath, let $\epsilon' = -\epsilon$ then $\strat_2'(\rho'.(\urg{l,e},\val+t-\delta)) = (0,\en)$ which results in $\rho'.(\urg{l,e},\val+t-\delta) \xrightarrow{0,\en} (\negP{(l,e)},\val+t - \delta)$ and $\strat_2'(\rho'.(\negP{(l,e)},\val+t - \delta)) = (\delta - \epsilon', \eno)$ resulting in $\rho'.(\negP{(l,e)},\val+t - \delta) \xrightarrow{\delta - \epsilon', \eno} (l',\val+t-\delta+\delta - \epsilon'[r:=0])$.
 Note that on entering $\negP{(l,e)}$ with a value $\val+t-\delta$, a time 
 in $\epsilon \in [0, \delta]$ is spent at $\negP{(l,e)}$, obtaining a valuation 
 $\val+t-\delta+\epsilon$. This corresponds to altering the time $t$ spent by controller in $\rtg$
 to a value $t-\delta+\epsilon \in [t-\delta, t]$.  
      \end{itemize}
\end{enumerate}

Similarly, given a strategy $\strat'$ in $\ptg$, we shall construct the equivalent strategy 
$\strat$ in $\rtg$ as follows. 
\begin{enumerate}
 \item \textbf{Player~1 strategy to controller strategy}
  If $\strat_1'(s)$ proposes a delay $t$ then $\strat_1(f^{(-1)}(s))$ proposes a delay $t+\delta$. 
  
 \item \textbf{Player~2 strategy to perturbator strategy in perturbator locations}
  If $\strat_2'(s)$ proposes a delay $t$ then $\strat_2(f^{(-1)}(s))$ also proposes $t$. 
 
 \item \textbf{Player~2 strategy to perturbator strategy in controller locations} 
  Suppose $\strat_2'(\urg{l,e},\val+t)$ suggests the path $(\posP{(l,e)},\val,\cost) \xrightarrow{\delta+\epsilon} (l',\val')$. Then, $\strat_2((l,\val),t+\delta,e) \xrightarrow{\epsilon}(l',\val')$.
\end{enumerate}

\begin{lemma}
 \label{lem:clockRange_target}
  In the RPTG $\rtg$ given in Figure \ref{rta-ptg}, if $g$ is $0<x<1$ then $B$ is reached with $x \in [0,1+\delta]$. 
  In the corresponding PTG $\ptg$ too, $B$ is reached with $x \in [0,1+\delta]$. We can establish the same for other possible guards too.
\end{lemma}

\begin{lemma}
\label{lem:cost_trans}
$\costp(s \xrightarrow{a} s') = \costp(\sm(s) \stackrel{a'\cdots}{\rightsquigarrow} \sm(s'))$. 
That is, the cost of a  transition from $s$ to $s'$ in the RPTG $\rtg$ is the same as the cost 
of going from $\sm(s)$ to $\sm(s')$ in the dwell-time PTG $\ptg$. However, we need multiple transitions 
to reach from $\sm(s)$ to $\sm(s')$. 
\end{lemma}
Both the above lemmas follow from the definition of $\prices'$ and the delays adjusted over $l$, $\negP{(l,e)}$ and $\posP{(l,e)}$ in the PTG $\ptg$.

\begin{lemma}
 \label{lem:cost_strategy}
 Given a strategy $\strat_1$ in $\rtg$ and the corresponding strategy $\strat_1'$ in $\ptg$, 
 for every state $s$ in $\rtg$, $\costp(s,\strat_1) = \costp(\sm(s),\strat_1')$.
\end{lemma}
\begin{proof}
Recall that $\costp_{\rtg}(s,\strat_1) = sup_{\strat_2 \in \maxstrats{\rtg}}(\costp_{\rtg}(\outcome(s,\strat_1,\strat_2)))$.\\
\textbf{Part 1:} \boldmath $\costp_{\rtg}(s,\strat_1) \leq \costp_{\ptg}(\sm(s),\strat_1')$ \unboldmath\\
 Consider a strategy $\strat_2$ in $\rtg$. We can construct a strategy $\strat_2'$ 
in $\ptg$ as outlined above. From Lemma \ref{lem:cost_trans}, it is clear that 
the $\costp_{\rtg}(\outcome(s,\strat_1,\strat_2)) \leq \costp_{\ptg}(\outcome(\sm(s),\strat_1',\strat_2'))$. 

\noindent \textbf{Part 2:} \boldmath $\costp_{\ptg}(\sm(s),\strat_1') \leq \costp_{\rtg}(s,\strat_1)$ \unboldmath\\
 Consider a strategy $\strat_2'$ in $\ptg$. We can construct a strategy $\strat_2$ 
in $\ptg$ as outlined above. The selected semantics of $\ptg$ 
and Lemma \ref{lem:clockRange_target} ensure that all of $\strat_2'$ proposed 
delays can be emulated in $\rtg$ too.
\end{proof}

Along the same lines as the lemma above, we could also prove that $\costp(s,\strat_2) = \costp(\sm(s),\strat_2')$. These two results pave the way for  relating the optimal costs for 
states in the two games. We shall establish $\optcostd{\rtg}(s) = \optcost{\ptg}(\sm(s))$ by proving two inequalities \\
(1) $\optcostd{\rtg}(s) \leq \optcost{\ptg}(\sm(s))$ and \\
(2) $\optcost{\ptg}(\sm(s)) \leq \optcostd{\rtg}(s)$

%
%
%
%
%
%
%

\subsection*{\boldmath $\optcostd{\rtg}(s) \leq \optcost{\ptg}(\sm(s))$ \unboldmath} 
%
Consider a strategy $\strat_1$ in $\rtg$ and construct an equivalent strategy $\strat_1'$ 
in $\ptg$ (this is possible, Lemma \ref{lem:path_map}). Now  we shall prove that 
$$ sup_{\strat_2 \in \maxstrats{\rtg}} (\costp(\outcome(s,\strat_1,\strat_2)))  = 
sup_{\strat'_2 \in \maxstrats{\ptg}} (\costp(\outcome(\sm(s),\strat'_1,\strat'_2))). $$ To this end, let us consider a perturbator strategy $\strat_2$ in $\rtg$.  Then we can construct an equivalent Player~2 strategy $\strat_2'$ such that $\costp(\outcome(s,\strat_1,\strat_2))= \costp(\outcome(s,\strat'_1,\strat'_2))$ (follows from Lemma \ref{lem:cost_trans}). Thus, we have shown that the set of strategies in $\ptg$ is at least as large as  those in $\rtg$ and whatever costs are achieved in $\rtg$ can be achieved in $\ptg$ too.\\

\subsection*{\boldmath $\optcost{\ptg}(s) \leq \optcostd{\rtg}(s)$ \unboldmath} 

 We shall now construct strategies in $\rtg$ from strategies in 
$\ptg$. If $\strat_1'(s)$ proposes a delay $t$ then $\strat_1(f^{(-1)}(s))$ proposes $t+\delta$. 
Lemma \ref{lem:clockRange_target} ensures that $t+\delta$ will satisfy the guard. For example, if the guard was $0<x<1$ in $\rtg$ then 
 the delay chosen by $\strat_1'$ is $< 1 - \val(x) - \delta$. 

Similarly, we construct strategy $\strat_2$ from $\strat_2'$ as specified above.  
If $\strat_2'(\urg{l,e},\val+t)$ suggests the path $(\posP{(l,e)},\val) \xrightarrow{\delta+\epsilon} (l',\val')$. Then, $\strat_2((l,\val),t+\delta,e) \xrightarrow{\epsilon}(l',\val'')$. We know that, $\val''=\val'$. Once again, Lemma \ref{lem:clockRange_target} ensures that if $\val'$ is in an interval $I$ then $\val'' \in I$. For example, for the guard $0<x<1$, $\val', \val'' \in  [0,1+\delta]$. 

Once we have mapped the strategies, the proof of $\optcost{\ptg}(s) \leq \optcostd{\rtg}(s)$ 
is along the same lines as the previous case.

\begin{lemma}
 if $\strat_1$ in $\rtg$ is $(\epsilon,N)-$acceptable then 
$\strat_1'$ in $\ptg$ is also $(\epsilon,N)-$acceptable.
\end{lemma}
A strategy in $\rtg$ is said to be $(\epsilon,N)-$acceptable if (1) it is memoryless, (2) is $\epsilon-$optimal 
for every state and (3) partitions $[0,1+\delta]$ into at most $N$ intervals. \\
From the definition of equivalent strategy $\strat_1'$, it is easy to see that if $\strat_1$ is memoryless then so is $\strat_1'$. Additionally, if $\strat_1$ has $n$ intervals then 
$\strat_1'$ would also have $n$ intervals except that the intervals' end points would be shifted by $\delta$ as the delay prescribed by $\strat_1'$ are $t-\delta$ when $\strat_1$ suggests $t$. Finally, we can claim that $\epsilon-$optimality is preserved from Lemma \ref{lem:cost_strategy}.

\section{Dwell time PTG to Dwell time FRPTG}
\label{app:transf2}
\begin{figure}[h]
\begin{center}
\begin{tikzpicture}[->,>=stealth',shorten >=1pt,auto,node distance=1cm,  semithick,scale=0.9]
 
  \node[min] at (1.7,0) (Ap) {$k$};
  \node[locLabel] () [above of =Ap] {$A_b$};
  \node[cdel] () [below of=Ap] {$t-\delta$};

  \node[max] at (5,0) (mu) {$0$};
  \node[locLabel] () [above of =mu] {$\urg{A,e}_b$};
  \node[del] () [below of=mu] {$0$};


  \node[max] at (7,-1) (mp) {$k$};
  \node[locLabel] () [above of =mp] {$\posP{(A,e)}_b$};
  \node[del] () [below of=mp] {$[\delta,2\delta]$};

  \node[max] at (7,-3) (mp1) {$k$};
  \node[locLabel] () [left of =mp1,node distance = 9.2mm] {$\zeroP{(A,e)}_{b+1}$};
  \node[del] () [below of=mp1] {$0$};

  \node[max] at (7,1) (mn) {$k$};
  \node[locLabel] () [above of =mn] {$\negP{(A,e)}_b$};
  \node[del] () [below of=mn] {$[0,\delta]$};

  \node[min] at (9,0) (Bp) {$k'$};
  \node[locLabel] () [right of =Bp] {$B_b$};

  \node[min] at (9,-3) (Bp1) {$k'$};
  \node[locLabel] () [node distance =8mm,right of =Bp1] {$B_{b+1}$};

  \draw[trans] (Ap)--(mu) node[midway,above] {$g$} node[midway,below] {$a$};
  \draw[trans] (mu)--(mp) node[midway,above] {$ $};
  \draw[trans] (mu)--(mn) node[midway,above] {$ $};
  \draw[trans] (mn)--(Bp) node[midway,below] {$r$};
  
  \draw[trans] (mp)--(Bp) node[midway,below] {$r$};
  \draw[trans] (mp)--(Bp) node[midway,above] {$ $};
  \draw[trans] (mp1)--(Bp1) node[midway,above] {$r$};
  \draw[trans] (mn)--(Bp) node[midway,above] {$ $};
  \draw[trans] (mp)--(mp1) node[midway,left] {$x{\geq} 1, \fr{x}$};
  \draw[trans] (Bp) -- (Bp1) node[midway,right] {$x{=}1,\set{x}$};


\end{tikzpicture}

\caption{FRPTG}
\label{bound}
\end{center}
\end{figure}
We have already defined in section \ref{app:tran2}, the mapping between 
the states of the dwell-time PTG $\ptg$ and the constructed dwell-time FRPTG $\frptg$.
The mapping is defined in such a way that a state $(l,\val)$ in $\ptg$ 
is mapped to the state $(l_b,\val-b)$, whenever $\val \in [b,b+1]$. the integral part 
of the clock valuation is remembered in the state itself, while the valuation always stays 
in [0,1]. The only exception to this is the location 
$\posP{(l,e)_b}$, where the clock valuation can go up to 
$1+\delta$.The state 
$(\posP{(l,e)_b},\val)$ in $\frptg$ is the mapping of the state 
$(\posP{(l,e)}, b+\val)$ in $\ptg$.

\begin{lemma}
 \label{lem:path_map2}
 Given a path $\rho$ in $\ptg$ from $s$ to $s'$, there exists a unique 
 path $\rho'$ in $\frptg$ from $\sm(s)$ to $\sm(s')$. Additionally, $\costp(\rho) = \costp(\rho')$. 
\end{lemma}
The proof of Lemma \ref{lem:path_map2} is straightforward, given the mapping 
$f$. Any time elapse of 1 in any one state $(l, \val)$ 
in $\ptg$ is captured by
 starting from some $(l_b, \val-b)$, and moving to $(l_{b+1}, (\val+1)-(b+1))$ in $\frptg$ and so on.
 Whenever the clock value reaches an integral value in $\ptg$, 
 correspondingly in the $\frptg$, the state is updated by remembering the new integral part, 
 and updating the clock valuation to 0. Every path in $\ptg$ corresponds to a path in $\frptg$, where 
 the constraints on the path are shifted by an appropriate integer, depending 
 on the integral value remembered in the current state. 
 
 This also gives a mapping between the strategies 
 of $\ptg$ and $\frptg$.  Also, the costs 
 are preserved across paths : any path in $\ptg$ is mapped to a longer path 
 in $\frptg$ so that the individual time delays in $\frptg$ never exceed 1. Since the 
 prices of states are preserved by the mapping, the costs will add up to be the same. 
 It is easy to see that a {\it{copy-cat}} strategy works between $\ptg$ and $\frptg$, and hence, 
 costs, optimal costs are preserved. 
 Since strategies are copy-cat, all properties like $(\epsilon,N)$-acceptability are also preserved across games.

\section{FRPTG to Reset-free FRPTG}
\label{app:resetfree}

 Given a one clock $\FRPTG$ 
 $\frptg=(\locations_1,\locations_2,\set{x},\trans,\prices,\goals)$ with $n$ resets (including fractional resets), 
 we define a reset-free $\FRPTG$ 
 as follows : 
 $\frptg' = (\locations_1',\locations_2',\set{x},\trans',\prices',\goals')$
 where
 \begin{itemize}
  \item For $l \in \locations_1$ and $0 \leq j <n$, we have $l^j \in \locations_1'$;
  \item For $l \in \locations_2$ and $0 \leq j <n$, we have $l^j \in \locations_2'$;
     \item $S \not \in \locations_1 \cup \locations_2$ is a sink location such that $S \in \locations_2'$;
  \item $\trans'$ has the following transitions.
  \begin{itemize}   
   \item $l^j \xrightarrow{g}l'^{j}$ if $l \xrightarrow{g} l' \in \trans$;
   \item $l^j \xrightarrow{g,r}l'^{j+1}$ if $l \xrightarrow{g,r} l' \in \trans$, $j<n$ and $r$ is either $\set{x}$ or $\fr{x}$;
   \item $l^n \xrightarrow{g,r}S$ if $l \xrightarrow{g,r} l' \in \trans$ and $r$ is either $\set{x}$ or $\fr{x}$;
   \item $S \xrightarrow{} S$;
   \end{itemize}
 \item $\prices'(l^j) = \prices(l)$
    and $l^j \in \goals'$ if $l \in \goals$.
 \end{itemize}
 
We illustrate the construction of a resetfree FRPTG 
in  Figure \ref{fig:resetfree_frptg}, 
  corresponding to the FRPTG in Figure \ref{bound}.
 Note that the locations in the upper rectangle form the the first copy $\frptg$-0 and while the lower rectangle forms 
the second copy $\frptg$-1. A copy $\frptg$-$i$ indicates the number of resets seen so far from the initial location $l_0^0$ 
of the first copy $\frptg$-0.
\begin{figure}[h]
\begin{center}
\begin{tikzpicture}[->,>=stealth',shorten >=1pt,auto,node distance=1cm,  semithick,scale=0.9]
 
  \node[min] at (1.7,0) (Ap) {$k$};
  \node[locLabel] () [above of =Ap] {$A_b$};
  \node[cdel] () [below of=Ap] {$t-\delta$};

  \node[max] at (5,0) (mu) {$0$};
  \node[locLabel] () [above of =mu] {$\urg{A,e}_b$};
  \node[del] () [below of=mu] {$0$};


  \node[max] at (7,-1) (mp) {$k$};
  \node[locLabel] () [above of =mp] {$\posP{(A,e)}_b$};
  \node[del] () [below of=mp] {$[\delta,2\delta]$};

  \node[max] at (7,-3) (mp1) {$k$};
  \node[locLabel] () [left of =mp1,node distance = 9.2mm] {$\zeroP{(A,e)}_{b+1}$};
  \node[del] () [below of=mp1] {$0$};

  \node[max] at (7,1) (mn) {$k$};
  \node[locLabel] () [above of =mn] {$\negP{(A,e)}_b$};
  \node[del] () [below of=mn] {$[0,\delta]$};

  \node[min] at (9,0) (Bp) {$k'$};
  \node[locLabel] () [right of =Bp] {$B_b$};

  \node[min] at (9,-3) (Bp1) {$k'$};
  \node[locLabel] () [node distance =8mm,right of =Bp1] {$B_{b+1}$};

  \draw[trans] (Ap)--(mu) node[midway,above] {$g$} node[midway,below] {$a$};
  \draw[trans] (mu)--(mp) node[midway,above] {$ $};
  \draw[trans] (mu)--(mn) node[midway,above] {$ $};
  \draw[trans] (mn)--(Bp) node[midway,below] {$r$};
  
  \draw[trans] (mp)--(Bp) node[midway,below] {$r$};
  \draw[trans] (mp)--(Bp) node[midway,above] {$ $};
  \draw[trans] (mp1)--(Bp1) node[midway,above] {$r$};
  \draw[trans] (mn)--(Bp) node[midway,above] {$ $};
  \draw[trans] (mp)--(mp1) node[midway,left] {$x{\geq} 1, \fr{x}$};
  \draw[trans] (Bp) -- (Bp1) node[midway,right] {$x{=}1,\set{x}$};


\end{tikzpicture}

\caption{FRPTG}
\label{bound}
\end{center}
\end{figure}
\begin{figure}[h]
\begin{center}
\begin{tikzpicture}[->,>=stealth',shorten >=1pt,auto,node distance=1cm,  semithick,scale=0.9]

    \draw[dashed,draw=gray,rounded corners=10pt] (1,2) rectangle (13,-2);
    \node at (1.8, -1) (g0) {$\mathbf{\frptg-0}$};
    \draw[dashed,draw=gray,rounded corners=10pt] (1,-4) rectangle (13,-6);
    \node at (1.8, -5) (g1) {$\mathbf{\frptg-1}$};
 
  \node[min] at (2,0) (Ap) {$k$};
  \node[locLabel] () [above of =Ap] {$[A_0]^0$};
  \node[cdel] () [below of=Ap] {$t-\delta$};

  \node[max] at (6,0) (mu) {$0$};
  \node[locLabel] () [above of =mu] {$[\urg{A,e}_0]^0$};
  \node[del] () [below of=mu] {$0$};


  \node[max] at (8.5,-1) (mp) {$k$};
  \node[locLabel] () [above of =mp] {$[\posP{(A,e)}_{0}]^0$};
  \node[del] () [below of=mp] {$[\delta,2\delta]$};

  \node[max] at (8.5,-5) (mp1) {$k$};
  \node[locLabel] () [node distance=7mm,left of =mp1] {$[\posP{A,e}_{1}]^1$};
  \node[del] () [below of=mp1] {$0$};

  \node[max] at (8.5,1) (mn) {$k$};
  \node[locLabel] () [above of =mn] {$[\negP{(A,e)}_{0}]^0$};
  \node[del] () [below of=mn] {$[0,\delta]$};

  \node[min] at (11,0) (Bp) {$k'$};
  \node[locLabel] () [above of =Bp] {$[B_0]^0$};

  \node[min] at (11,-5) (Bp1) {$k'$};
  \node[locLabel] () [below of =Bp1] {$[B_1]^1$};

  \draw[trans] (Ap)--(mu) node[midway,above] {$g \cap (0{\leq} x {<}1)$} node[midway,below] {$a,c$};
  \draw[trans] (mu)--(mp) node[midway,above] {};
  \draw[trans] (mu)--(mn) node[midway,above] {};
\draw[trans,sloped] (mn)--(Bp) node[midway,below] {$0 \leq x <1$};
 \draw[trans] (mn)--(Bp) node[midway,above] {$r$};
  
  \draw[trans] (mp)--(Bp) node[midway,above] {$r$};
  \draw[trans,sloped] (mp)--(Bp) node[midway,below] {$0 \leq x < 1$};
  \draw[trans] (mp1)--(Bp1) node[midway,above] {$r$};
  \draw[trans] (mp1)--(Bp1) node[midway,below] {};
  
  \draw[trans] (mp)--(mp1) node[midway,left] {$x{\geq} 1, \fr{x}$};
  \draw[trans] (Bp) -- (Bp1) node[midway,right] {$x{=}1,\set{x}$};

%
%
%
%
%
%
%
%
%
%
\node at (8, -7) (n) {$Resetfree~\FRPTG$};
  
\end{tikzpicture}

\caption{Resetfree $\FRPTG$ - two copies $\frptg-0$ and $\frptg-1$ correspoding to the number of resets encounterd so far 
i.e; $\frptg-0$ indicates that 0 resets have been seen so far and $\frptg-1$ indicates 1 (fractional) reset has been seen.} 
\label{fig:resetfree_frptg}
\end{center}
\end{figure}

\subsection{Proof of Lemma \ref{lem:resetfree_frptg}}

\begin{proof}
Consider any state $(l, \val)$ in $\frptg$.  The reduction from 
the FRPTG $\frptg$ to the reset-free 
FRPTG $\frptg'$  
creates a new component (or copy) for each new reset, including fractional resets.
Given that there are a total of $n$ resets in the FRPTG, 
$\frptg$, $n+1$ reset-free components are created in the reset-free FRPTG 
$\frptg'$, and the last component goes to a location with cost $+\infty$. 
By assumption, the cycles in each reset-free component 
are non-negative.  
Any cycle in the FRPTG $\frptg$ involving a reset 
is mapped to a path in the reset-free FRPTG $\frptg'$ ending at the location $S$ with cost $+\infty$, 
while any reset-free cycle in $\frptg$ is mapped to  a  cycle 
in one of the $n+1$ reset-free components of the reset-free FRPTG $\frptg'$. 
Clearly, for every strategy $\sigma$ of player 1, 2 in $\frptg$, there is a corresponding strategy $\sigma'$ in the 
$\frptg'$ and vice-versa,  obtained using the above mapping of paths 
between $\frptg$ and $\frptg'$.  Given that the prices of locations are preserved between 
$\frptg$ and $\frptg'$, the optimal cost from $(l, \val)$ in 
$\frptg$ is the same as the optimal cost from $(l^0, \val)$ in $\frptg'$.

Consider a $(\epsilon,N)$-acceptable strategy $\sigma'$ in $\frptg'$. Consider a winning state $(l^0,\nu)$. Let $i$ be the minimum number of resets 
from state $(l^0,\nu)$ along any path compatible with $\sigma'$. That is, the player can win from $(l^{n-i},\nu)$ but not from $(l^{n-i+1},\nu)$. If $(l,\nu)$ is not winning then we take $i=n+1$. We denote by $\sigma'_{n-i}$ the suggestions made by $\sigma'$ in the $n-i^{th}$ copy in $\frptg'$. We  then assign $\sigma(l,\nu) = \sigma'_{n-i}(l^{n-i},\nu)$. Thus, we obtain that $\costp_{\frptg}((l,\nu),\sigma) = \costp_{\frptg'}((l^{n-i},\nu),\sigma')$.
Since $(l^0, \nu)$ and $(l^{n-i}, \nu)$ have the same outgoing transitions, we know that 
the strategy $\sigma'$ from $(l^0, \nu)$ will be atleast as costly as 
$\optcost{\frptg'}(l^{n-i},\val)$. That is, 
\begin{equation}\label{eqn1}
 \costp_{\frptg'}((l^0,\val),\sigma') \geq \optcost{\frptg'}(l^{n-i},\val)
\end{equation}
Now, if $\costp_{\frptg'}((l^{n-i},\nu),\sigma') ~> ~\costp_{\frptg'}((l^0,\nu),\sigma') + \epsilon$, then by Equation \ref{eqn1} we 
have $\costp_{\frptg'}((l^{n-i},\val),\sigma')>\optcost{\frptg'}(l^{n-i},\val)+\epsilon$ which means $\sigma'$ is not $\epsilon$-optimal. Thus we have, \\
$\costp_{\frptg'}((l^{n-i},\nu),\sigma') ~\leq ~
\costp_{\frptg'}((l^0,\nu),\sigma') + \epsilon$
\begin{equation*}
\costp_{\frptg}((l,\nu),\sigma)= \costp_{\frptg'}((l^{n-i},\nu),\sigma')
    \begin{cases}
		\leq ~\costp_{\frptg'}((l^0,\nu),\sigma') + \epsilon \\    
		\leq~\optcost{\frptg'}(l^0,\nu) + \epsilon + \epsilon \\
		\leq~\optcost{\frptg}(l,\nu) + 2\epsilon
    \end{cases}
\end{equation*}
\end{proof}

We shall now focus on informally explaining why fractional resets would not cause a problem. 
In a PTG without fractional resets, a resetting transition $e$ (say $l \xrightarrow{x=1, x:=0} m$) taken twice takes us back to the same state $(m,x=0)$ twice. This crucial property is the back bone of the transformation which removes resets in \cite{BouLar06}. The correctness proof is by constructing optcost preserving strategies for Player 1 in both $\ptg$ and its resetfree equivalent $\ptg'$. Given a winning strategy for Player 1 in PTG $\ptg$, a strategy in $\ptg'$ is constructed so as to ensure each resetting transition is taken atmost once. This is possible because a resetting transition $e$ appearing the second time, results in the same state $(m,0)$ and hence the transitions possible (and the optcost achievable) from the second resultant state $(m,0)$ can be applied the first time this state occurs itself. In other words, the second occurrences of the transition are replacable as they result in the same unique state $(m,0)$. It should be the case that a path exists such as to avoid the second 
occurrence of the resetting transition as the strategy is winning for Player 1.

Now, a similar reasoning will not work for fractional resets $e'$ (say $l' \xrightarrow{x\geq 1, [x]:=0} m'$)  as the resulting state $(m',x)$ after a fractional reset transition is not unique (as the clock $x \in [0,\delta]$) and thus we can not adopt this argument directly.  Firstly, the player 2 location $\posP{(l,e)}_i$ is entered with $x \leq 1-\delta$ (see Transformation 2) and  a delay $d$ makes $x \in [1,1+\delta]$.  This delay happens entirely in this location and is chosen entirely by Player 2. From  $\posP{(l,e)}_i$, if player 2 moves to a  
 $\zeroP{(l,e)_{i+1}}$ location, then the value of $x$ is in $[0, \delta]$. Note that 
 the value of $x$, say $\zeta$  in $\zeroP{(l,e)_{i+1}}$ is indeed the perturbation 
 that happened in the RPTG $\rtg$ : in the FRPTG, at 
 $\posP{(l,e)}_i$, player 2 elapses $\delta+\zeta$. 
 Recall that if in $\rtg$, a location $l$ was entered with value of $x$ being $\nu$, then 
 in the FRPTG, we enter 
 $(l,e)_i$ and  
 $\posP{(l,e)}_i$ with $\nu-i-\delta$. Player 2 
 at $\posP{(l,e)}_i$ makes this value to be $\nu-i+\zeta$, which is exactly same as the perturbed value 
 of $x$ in $\rtg$ when perturbator chooses a positive perturbation.
 The point to note is that whevener player 2 returns to $\posP{(l,e)}_i$, 
 the control of perturbation is his; thus, any $\zeta$ that is achieved the 
 $k$th time can be achieved the first time itself by 
 player 2. Moreover, if player 2 has a strategy 
 to revisit $\posP{(l,e)}_i$, then clearly, 
 player 1 will lose, since after $n+1$ times, the control reaches the target with cost $\infty$.
 Note also that in Algorithm 1, while 
 we solve for $\posP{(l,e)}_i$, we 
 will have the optcost function computed for 
 $\zeroP{(l,e)_{i+1}}$. Player 2 will choose to delay $\delta+\zeta$
 for that $\zeta$ where the cost is maximal 
 in the optcost of 
 $\zeroP{(l,e)_{i+1}}$.


\section{Example : Solve Reset-Free FRPTG}
\label{app:eg}
 We shall first look at how normal resets are handled. 
As detailed by the resetfree construction, each copy of the FRPTG is an SCC and there are $n+1$ copies when the FRPTG has $n$ resets. The $i+1$th copy is solved and its optcost functions are used as outside cost functions while solving the $i$th copy. 
This is depicted clearly in the figure below. \\
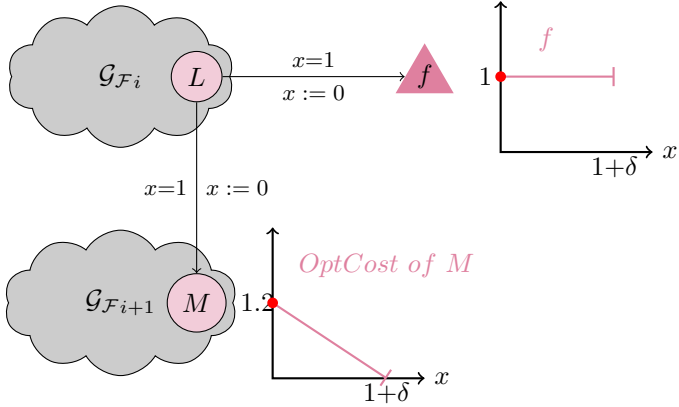
\begin{tikzpicture}[->]
 \node[cloud, draw,cloud puffs=10,cloud puff arc=120, aspect=2, inner ysep=1em,fill=black!20] at (0,0) {$\frptg_i$};
 \node[cloud, draw,cloud puffs=10,cloud puff arc=120, aspect=2, inner ysep=1em,fill=black!20] at (0,-3) {$\frptg_{i+1}$};
  \node[min,fill=purple!20] at (1,0) (A) {$L$};
\node[triangle] at (4,0)(F){$f$};

\node[min,fill=purple!20] at (1,-3) (B) {$M$};

 \draw[trans] (A)--(F) node[midway,above] {$x {=}1$} node[midway,below] {$x:=0$};
 \draw[trans] (A)--(B) node[midway,left] {$x {=}1$} node[midway,right] {$x:=0$};
 
   \draw [<->,thick] (5,1) node (yaxis) [above] {}
        |- (7,-1) node (xaxis) [right] {$x$};
        \node[cf] at (5.6,0.5) (f2) {$f$};
    \draw[-|,cf] (5,0) coordinate (b_1) -- (6.5,0) coordinate (b_2);
    \node[node distance = 2mm,left of=b_1] {$1$};
    \coordinate (b_2p) at (6.5,-1);
    \node[node distance = 2mm,below of=b_2p] {$1{+}\delta$};
     \fill[red] (b_1) circle (2pt);
   \draw [<->,thick] (2,-2) node (yaxis) [above] {}
        |- (4,-4) node (xaxis) [right] {$x$};
        \node[cf] at (3.5,-2.5) (f2) {$OptCost~of~M$};
    \draw[-|,cf] (2,-3) coordinate (c_1) -- (3.5,-4) coordinate (c_2);
    \node[node distance = 2mm,left of=c_1] {$1.2$};
     \node[node distance = 2mm,below of=c_2] {$1{+}\delta$};
      \fill[red] (c_1) circle (2pt);
\end{tikzpicture}
\\
Location $L$ in $\frptg_i$ has two transitions with resets - one to a target and another to a location $M$ in $\frptg_{i+1}$ 
whose optcost function has already been computed. Due to the clock reset, the target or $M$ are entered with clock value $x=0$ and hence the only values of interest are : the cost function $f$ of the target at $x=0$ i.e; $f(0)=1$ and optcost for location $M$ at $x=0$ i.e; $1.2$.  Since $L$ is a Player 1 location, the lower of these two values is picked and the corresponding transition is selected. 
\\
~\\

Now let us now consider fractional resets. The following figure depicts the previously conidered example with normal reset replaced with fractional reset.

 
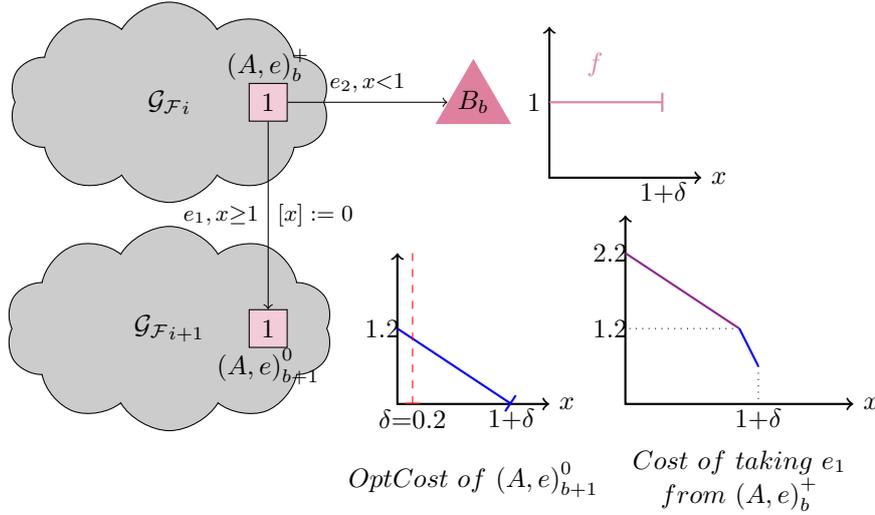
\begin{tikzpicture}[->]
 \node[cloud, draw,cloud puffs=10,cloud puff arc=120, aspect=2, inner ysep=1.6em,fill=black!20] at (0,0) {$\frptg_i$};
 \node[cloud, draw,cloud puffs=10,cloud puff arc=120, aspect=2, inner ysep=1.6em,fill=black!20] at (0,-3) {$\frptg_{i+1}$};
  \node[max,fill=purple!20] at (1.3,0) (A) {$1$};
   \node[locLabel] () [above of =A] {$\posP{(A,e)}_{b}$};
\node[triangle] at (4,0)(F){$B_b$};

\node[max,fill=purple!20] at (1.3,-3) (B) {$1$};
\node[locLabel] () [below of =B] {$\zeroP{(A,e)}_{b+1}$};

 \draw[trans] (A)--(F) node[midway,above] {$e_2,x {<}1$} node[midway,below] {};
 \draw[trans] (A)--(B) node[midway,left] {$e_1,x {\geq}1$} node[midway,right] {$[x]:=0$};
 
   \draw [<->,thick] (5,1) node (yaxis) [above] {}
        |- (7,-1) node (xaxis) [right] {$x$};
        \node[cf] at (5.6,0.5) (f2) {$f$};
    \draw[-|,cf] (5,0) coordinate (b_1) -- (6.5,0) coordinate (b_2);
    \node[node distance = 2mm,left of=b_1] {$1$};
    \coordinate (b_2p) at (6.5,-1);
    \node[node distance = 2mm,below of=b_2p] {$1{+}\delta$};
   \draw [<->,thick] (3,-2) node (yaxis) [above] {}
        |- (5,-4) node (xaxis) [right] {$x$};
        \node at (4,-5) (f2) {$OptCost~of~\zeroP{(A,e)}_{b+1}$};
    \draw[-|,fun1] (3,-3) coordinate (c_1) -- (4.5,-4) coordinate (c_2);
    \node[node distance = 2mm,left of=c_1] {$1.2$};
     \node[node distance = 2mm,below of=c_2] {$1{+}\delta$};
    \draw[-|,red,dashed] (3.2,-2) -- (3.2,-4) coordinate (c_3);
     \node[node distance = 2mm,below of=c_3] {{$\delta{=}0.2$}};

   \draw [<->,thick] (6,-1.5) node (yaxis) [above] {}
        |- (9,-4) node (xaxis) [right] {$x$};
        \node at (7.5,-5) (f2) {$\begin{array}{c}Cost~of~taking ~e_1 \\ from~\posP{(A,e)}_{b}\end{array}$};
    \draw[-,fun1] (7.5,-3) coordinate (c_1) -- (7.75,-3.5) coordinate (c_2);
     \draw[-,dashedl] (c_1) -- (6,-3) coordinate (c_4); 
    \node[node distance = 2mm,left of=c_4] {$1.2$};
     \draw[-,dashedl] (c_2) -- (7.75,-4) coordinate (c_3);
     \node[node distance = 2mm,below of=c_3] {$1{+}\delta$};
    \draw[-,fun2] (6,-2) coordinate (c_5) -- (c_1);
    \node[node distance = 2mm,left of=c_5] {$2.2$};
\end{tikzpicture}

       Recall from Transformation 2 (Dwell-time PTG $\ptg$ to FRPTG  $\frptg$) that fractional resets 
occur only along transitions from $\posP{(A,e)}_b$ to $\zeroP{(A,e)}_{b+1}$. Lets call this transition $e_1$. From the construction of \FRPTG, we also know that the only other transition possible from $\posP{(A,e)}_b$ is to location $B_b$, corresponding to the transition from $A$ to $B$ in the RPTG.  Let us denote this transition as $e_2$. by construction 
of the reset-free FRPTG,
the constraint on  $\posP{(A,e)}_b \xrightarrow{} B_b$ is $x<1$. 
 Figuring out which part of the cost functions of $B_b$ and $\zeroP{(A,e)}_{b+1}$ to consider for the optcost computation of $\posP{(A,e)}_{b}$ is a little different from the normal reset case. Here the guards on transitions can be considered as $0\leq x < 1$ for $e_2$ and $1 \leq x \leq 1+\delta$ for $e_1$.

 Hence we should consider the entire cost function of $B_b$, while taking only the $x\in[0,\delta]$  part from the function of $\zeroP{(A,e)}_{b+1}$. Recall that fractional resets removed the integer part of $x$, thereby  taking $x$ from $[1,1+\delta]$ to $[0,\delta]$. Thus the cost function of taking the transition to $\zeroP{(A,e)}_{b+1}$ is equal to
  (delay of waiting at $\posP{(A,e)}_{b}$ till $x = v \in [1,1+\delta]$) $+$ 
(optcost of $\zeroP{(A,e)}_{b+1}$ at $1-v$).
 We compute this cost as outlined in Figure \ref{fig:superimpose_diffGuard}. It is easy to see that since the  price of $\posP{(A,e)}_{b}$ is 1,  and the slope of the optcost function of $\zeroP{(A,e)}_{b+1}$ is $-1.2$, 
 it is more profitable to reach $\zeroP{(A,e)}_{b+1}$ at $x=0$ than at $x=\delta$ i.e; the transition $e_1$ when $x=1$, thus reaching $\zeroP{(A,e)}_{b+1}$ at $x=0$ after the fractional reset, rather than wait at $\posP{(A,e)}_{b}$ till $x=1+\delta$ and then reach $\zeroP{(A,e)}_{b+1}$ at $x=\delta$ yielding only $2.16$ (wait till $x=1+\delta$ incurring $1.2$ and then optcost of $\zeroP{(A,e)}_{b+1}$ at $x=\delta$ is 0.96). 
 
 We consider this cost function of taking the transition $e_1$ and the cost function of $B_b$ to compute the optcost function of $\posP{(A,e)}_{b}$. In this example, it is clearly better for Player 2 to take $e_1$ at $x=1$ and hence the cost function of taking $e_1$ is the optcost of $\posP{(A,e)}_{b}$.
\\
\\

\section{Algorithm for OptCost Computation : $l$ is a player 1 location}
\label{app:pl1}
We first prove Lemma \ref{lem:replaceFun}. 
\begin{proof}
The optcost computation for a location $\lmin$ 
 is done using the already computed optcosts 
 of all successors of $\lmin$, which we now treat as outside cost functions. 
  The Steps 1 and 2 in Algorithm \ref{algo:optcost_compute}, 
 superimpose the outside cost functions corresponding to $\lmin$ and take the interior. 
  Recall that step 3 is applied right to left :
  we start the selective replacement from $\nu  = 1+\delta$ and 
 proceed towards $0$.
    We know that up to $v_i$, for all $\nu \geq v_i$, $\optcost{}(\lmin,\nu) = f(v_i)$.
    
       Now we have to compute the optcost for $\nu \in [u_i,v_i]$. 
 As we have taken the interior of the superimposed function in Step 2, we know that 
 $g_i$ is the best (lowest) possible cost if we do not delay at $\lmin$. 
Let us determine if delaying at $\lmin$ is more profitable than following $g_i$. 
 The two options we have are : 
 \begin{enumerate}
  \item Follow $g_i$ whose slope is $-m$.
  The line segment $g_i$ is given by 
 $y = -mx + c$ where $c = f(v_i) + mv_i$, since $y=g_i(v_i) = f(v_i)$ at $x=v_i$;
     ($f$ is continuous, and is composed of $g_1, \dots, g_m$) and 
\item delay at $\lmin$ till $x = v_i$ and 
 exit at $v_i$. 
  The line segment corresponding to the delay at $\lmin$ is $y = -\prices(\lmin)*x + c'$ where $c' = f(v_i) + \prices(\lmin)*v_i$ as we delay 
 at $\lmin$ until $x=v_i$ and follow $f$, thus obtaining $f(v_i)$ at $x=v_i$
 \end{enumerate}
  Now comparing these two equations we get the following. 
\begin{align*}
 -mx+c~&\sim~-\prices(\lmin)*x + c' \\
 -mx+f(v_i) + mv_i~&\sim~-\prices(\lmin)*x + f(v_i) + \prices(\lmin)*v_i \\
 -mx+mv_i~&\sim~-\prices(\lmin)*x + \prices*v_i \\
 m*(v_i - x) ~&\sim~\prices(\lmin) * (v_i - x) \\
 \mbox{If $x \leq v_i$ and $\prices(\lmin) \leq m$},~ & \mbox{then we conclude $\sim$ is $\geq$} 
 \end{align*}
 Thus, we observe that delaying at $\lmin$ is better. The above discussion is for Player~1 
 but can be easily adapted to Player~2. 
 In a similar fashion, we can argue that delaying at $\lmin$ till $ x \leq v_i' < v_i$ is 
 worse than delaying till $x\leq v_i$ i.e; Player~1 prefers to 
 wait until $v_i$ instead of  exiting and following $g_i$ at 
 some point $v_i' < v_i$.  
\end{proof}

\subsection{OptCost Computation for All Constraints}
\label{app:other}
In the computation in Algorithm \ref{algo:optcost_compute}, we have 
assumed that all the transitions from $l \xrightarrow{e} l'$ have a guard $0 \leq x \leq 1$. We shall now illustrate how to compute optcost of $l$ if the guards on the outgoing transitions are different. 

While optimal strategies are possible with closed constraints, 
it is known that optimal strategies need not exist with open constraints.  

\paragraph*{Constraints on all the outgoing edges are  $0<x<1$}
 We shall illustrate how to obtain $\epsilon-$optimal strategies 
with open constraints. Consider the Figure \ref{fig:solve_openapp}.
Here the guard is $0<x<1$ and clearly the $\optcost{\ptg}(l,0) = 2$ and there is no strategy to achieve that. Hence 
we want to find the $\epsilon$-optimal strategy achieving $<2+\epsilon$. Pick $t = \frac{\epsilon}{m_{max}+1}$ where 
$m_{max}$ is the slope  with the largest absolute value seen among the outside cost functions. Here $m_{max} = 3$ ($\optcost{(A)}$ is 
$y=-3x+3$).  Let $\epsilon = 0.1$. Then $t = 0.025$. Now lets fix the strategy to wait at $l$ till $x < 1-t$ and go to $A$ at $x=1-t$\footnote{If there are $n$ transitions in the longest path from source to target then $t = \frac{\epsilon}{n* (m_{max}+1)}$.}. Then $\optcost{\ptg}(l,0) = 2*(1-t)+f(1-t)$ where $f$ given by $y=-3x+3$ is the optcost function of $A$. Thus $\optcost{\ptg}(l,0,0) = 2.025 < 2 + 0.1$. 
Extending this to several successors of $l$ is simple and follows all the steps of Algorithm \ref{algo:optcost_compute}. At  $1-t$,  take the transition to the location prescribed by $f'$ in Step 4. Note that this method would work 
for $0<x<1-\delta$ by simply replacing $1$ with $1-\delta$ in the discussion above. 

\begin{figure}
\begin{center}
 
\begin{tabular}{c c c}
\begin{tikzpicture}[->,>=stealth',shorten >=1pt,auto,node distance=1cm,  semithick,scale=0.9]

  \node[min] at (0,1) (l) {$2$};
    \node[locLabel] () [above of = l]{$l$};
  
  \node[min] at (2,1) (A) {$A$};
  
  \draw[trans] (l) -- (A) node[sloped,midway,above] {$0{<}x{<}1$};  
  \node at (2,-2) {};
\end{tikzpicture}
& 
\hspace{1cm}
\begin{tikzpicture}[scale=0.5]
\tikzstyle{every node}=[font=\scriptsize]

    \draw [<->,thick] (0,4) node (yaxis) [above] {$y$}
        |- (5,0) node (xaxis) [right] {$x$};
    \node at (2.5,-1.5) {$\optcost{}(A)$};

    
    \draw[fun2] (0,3) coordinate (b_1) -- (4,0) coordinate (b_2);
    \node[node distance = 2mm,left of=b_1] {$3$};
    \node[node distance = 2mm,below of=b_2] {$1$};
\end{tikzpicture} 
&
 \begin{tikzpicture}[scale=0.5]
\tikzstyle{every node}=[font=\scriptsize]

    \draw [<->,thick] (0,4) node (yaxis) [above] {$y$}
        |- (5,0) node (xaxis) [right] {$x$};

        \node at (2.5,-1.6) {Selectively Replace};
    
    \draw[fun2] (0,3) coordinate (b_1) -- (4,0) coordinate (b_2);
    \node[node distance = 3mm,left of=b_1] {$3$};    
    \node[node distance = 4mm,below of=b_2] {$1$};

    \draw[rep] (0,2) coordinate (d1) -- (b_2);
     \node[node distance = 3mm,left of=d1] {$2$};
     
\end{tikzpicture}
 \end{tabular}
 
\end{center}

\caption{Optcost Computation for guard $0<x<1$}
\label{fig:solve_openapp}
\end{figure}
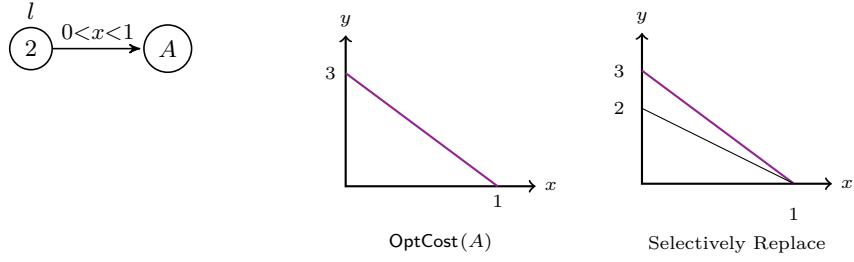

\paragraph*{Constraint $1-\delta<x<1$}
In transformation 1 from RPTG to dwell-time PTG, we replaced the constraint $H<x<H+1$ by $H-\delta <x <H+1-\delta$.
Such a constraint would correspond 
  in the resetfree $\FRPTG$   
to $1-\delta < x <1$ or $0 \leq x < 1-\delta$. 
We have already dealt with the constraint $0<x<1-\delta$.
 Now, we shall highlight the difference   to make it work for $1-\delta < x <1$.
  We shall compute as usual, by applying the steps of Algorithm\ref{algo:optcost_compute}, and also get the prescribed strategy out of the final function $f'$. Now if the computed strategy $\strat$ for $l$ 
suggests to take a transition in the interval $[0,1-\delta]$ then instead of this transition we prescribe waiting at $l$. This is because the guard on the outgoing transition(s) is 
$1-\delta<x<1$. The rest of the strategy prescribed 
by $f'$  over $(1-\delta,1]$ is retained as is.

\paragraph*{Constraints on outgoing edges are $x=0$, $0 \leq x \leq 1$, $x=1$}
Figure \ref{fig:superimpose_diffGuard} explains how to solve for optcost if the outgoing transitions 
have different guards. 

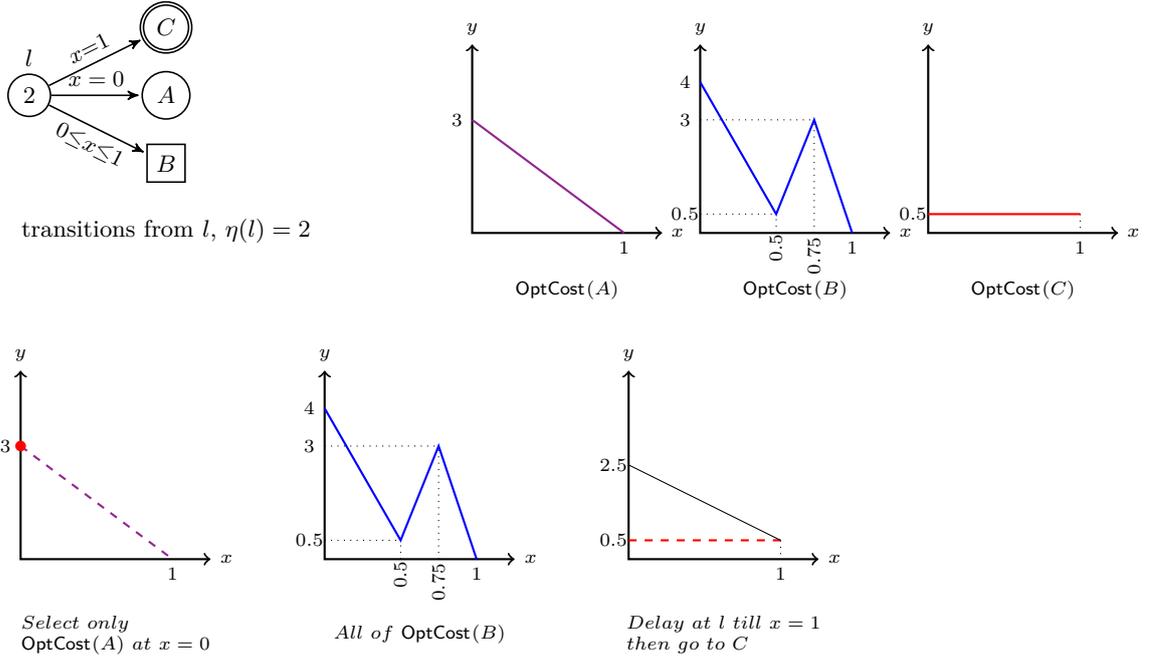
\begin{figure}
\begin{center}
 
\begin{tabular}{c c}
\begin{tikzpicture}[->,>=stealth',shorten >=1pt,auto,node distance=1cm,  semithick,scale=0.9]

  \node[min] at (0,1) (l) {$2$};
    \node[locLabel] () [above of = l]{$l$};

  \node[min] at (2,1) (A) {$A$};
  \node[max] at (2,0) (B) {$B$};
  \node[min,accepting] at (2,2) (C) {$C$};
  
  \draw[trans] (l) -- (A) node[sloped,midway,above] {$x=0$};
  \draw[trans] (l) -- (B) node[sloped,midway,below] {$0{\leq}x{\leq}1$};
  \draw[trans] (l) -- (C) node[sloped,midway,above] {$x{=}1$};
  
  \node at (2,-1) {transitions from $l$, $\prices(l)=2$};
  \node at (2,-2) {};
\end{tikzpicture}
& 
\hspace{1cm}
\begin{tikzpicture}[scale=0.5]
\tikzstyle{every node}=[font=\scriptsize]
    \draw [<->,thick] (0,5) node (yaxis) [above] {$y$}
        |- (5,0) node (xaxis) [right] {$x$};
    \node at (2.5,-1.5) {$\optcost{}(A)$};

    \draw[fun2] (0,3) coordinate (b_1) -- (4,0) coordinate (b_2);
    \node[node distance = 2mm,left of=b_1] {$3$};
    \node[node distance = 2mm,below of=b_2] {$1$};


\draw [<->,thick] (6,5) node (yaxis) [above] {$y$}
        |- (11,0) node (xaxis) [right] {$x$};
    
      \node at (8.5,-1.5) {$\optcost{}(B)$};
        
    \draw[fun1] (6,4) coordinate (a_1) -- (8,0.5) coordinate (a_2) -- (9,3) coordinate (a_3) -- (10,0) coordinate (a_4);
    
    \node[node distance = 2mm,left of=a_1] {$4$};
    
    \coordinate (a_2y) at (6,0.5);
    \node[node distance = 2mm,left of=a_2y] {$0.5$};
    \draw[dashedl] (a_2y) -- (a_2);
    \coordinate (a_2x) at (8,0);
    \node[node distance = 2mm,below of=a_2x] {\rotatebox{90}{0.5}};
    \draw[dashedl] (a_2x) -- (a_2);

    \coordinate (a_3y) at (6,3);
    \node[node distance = 2mm,left of=a_3y] {$3$};
    \draw[dashedl] (a_3y) -- (a_3);
    \coordinate (a_3x) at (9,0);
    \node[node distance = 3mm,below of=a_3x] {\rotatebox{90}{$0.75$}};
    \draw[dashedl] (a_3x) -- (a_3);
    
    \node[node distance = 2mm,below of=a_4] {$1$};    
    

\draw [<->,thick] (12,5) node (yaxis) [above] {$y$}
        |- (17,0) node (xaxis) [right] {$x$};
    
      \node at (14.5,-1.5) {$\optcost{}(C)$};
      
    \draw[fun3] (12,0.5) coordinate (c_1) -- (16,0.5) coordinate (c_2);  
    \node[node distance = 2mm,left of=c_1] {$0.5$};
    \draw[dashedl] (c_2) -- (16,0) node[below] {1}; 

\end{tikzpicture} 
 \end{tabular} 
\end{center}
\begin{tikzpicture}[scale=0.5]
\tikzstyle{every node}=[font=\scriptsize]
    \draw [<->,thick] (0,5) node (yaxis) [above] {$y$}
        |- (5,0) node (xaxis) [right] {$x$};
    \node at (2.5,-2) {$\begin{array}{l l}Select~only~\\ \optcost{}(A)~ at~ x=0\end{array}$};

    \draw[fun2,diminish] (0,3) coordinate (b_1) -- (4,0) coordinate (b_2);
    \node[node distance = 2mm,left of=b_1] {$3$};
    \node[node distance = 2mm,below of=b_2] {$1$};
    
    \fill[red] (b_1) circle (4pt);


\draw [<->,thick] (8,5) node (yaxis) [above] {$y$}
        |- (13,0) node (xaxis) [right] {$x$};
    
      \node at (10.5,-2) {$All~of~\optcost{}(B)$};
        
    \draw[fun1] (8,4) coordinate (a_1) -- (10,0.5) coordinate (a_2) -- (11,3) coordinate (a_3) -- (12,0) coordinate (a_4);
    
    \node[node distance = 2mm,left of=a_1] {$4$};
    
    \coordinate (a_2y) at (8,0.5);
    \node[node distance = 2mm,left of=a_2y] {$0.5$};
    \draw[dashedl] (a_2y) -- (a_2);
    \coordinate (a_2x) at (10,0);
    \node[node distance = 2mm,below of=a_2x] {\rotatebox{90}{0.5}};
    \draw[dashedl] (a_2x) -- (a_2);

    \coordinate (a_3y) at (8,3);
    \node[node distance = 2mm,left of=a_3y] {$3$};
    \draw[dashedl] (a_3y) -- (a_3);
    \coordinate (a_3x) at (11,0);
    \node[node distance = 3mm,below of=a_3x] {\rotatebox{90}{$0.75$}};
    \draw[dashedl] (a_3x) -- (a_3);
    
    \node[node distance = 2mm,below of=a_4] {$1$};    
    

\draw [<->,thick] (16,5) node (yaxis) [above] {$y$}
        |- (21,0) node (xaxis) [right] {$x$};
    
      \node at (18.5,-2) {$\begin{array}{l l}Delay~at~l~till~x=1\\then~go~to~C\end{array}$};
      
    \draw[fun3,diminish] (16,0.5) coordinate (c_1) -- (20,0.5) coordinate (c_2);  
    \draw[rep] (16,2.5) coordinate (d_1) -- (20,0.5) coordinate (d_2);
     
    \node[node distance = 2mm,left of=c_1] {$0.5$};
    \draw[dashedl] (c_2) -- (20,0) node[below] {1};
    \node[node distance = 2mm,left of=d_1] {$2.5$};

\end{tikzpicture}
\caption{Optcost Computation for different guards}
\label{fig:superimpose_diffGuard}
\end{figure}

As Player~1 can go to $A$ only if $x=0$, we need to consider only that point of $\optcost{}(A)$ while 
computing the optcost of $l$. Similarly, player 1 can go to $C$ only when $x=1$. Thus the function to consider, 
for taking the transition to $C$ is the 
cost of the path (or action) of delaying in $l$ till $0\leq x<1$ and going to $C$ at $x=1$. 
Upon reaching $C$ at $x=1$, the cost incurred will be $\optcost{}(C,1) = 0.5$. 
Delaying at $l$ at the rate of $\prices(l)=2$ yields a function with slope $-2$ passing through 
the point $(1,0.5)$ (corresponding to going to $C$ at $x=1$).

\section{$\lmin$ is a player 2 location}
\label{app:pl2}
Here we discuss in detail how to take care of the dwell-time requirements 
while running Algorithm \ref{algo:optcost_compute}. 
Recall that there are three kinds of player 2 locations : urgent, those with dwell-time requirements 
$[0, \delta]$ and those with dwell-time requirements $[\delta,2\delta]$. 

\begin{enumerate}
 \item \textbf{Urgent location} : superimpose and take exterior (only Steps 1 and 2).
 
  \item \textbf{$[0,\delta]$-delay location} : From Lemma \ref{lem:replaceFun}, we know that Player~2 will want to spend as much time as possible at a location $\lmin$ while keeping $x \leq v_i$ whenever 
  there is a function $g_i$ over $[u_i, v_i]$ whose slope is $> -\eta(\lmin)$. 
  Note that we proved Lemma \ref{lem:replaceFun} for player 1, however,
  an analogous result works when $\lmin$ is a player 2 location.

  Thus, if $\lmin$ is entered at $x=\val \in[u_i,v_i - \delta]$, then  player 2 spends $\delta$ time and exits (as $\delta$ is the maximum delay permitted in $\lmin$ by the dwell-time restriction) at $\nu + \delta$ to the successor as prescribed by $f$ at $x=\val+\delta$. If $\lmin$ is entered at $x=\val \in [v_i - \delta,v_i]$, then player 2 spends $v_i - \val$ at $\lmin$ and exits at $v_i$ to the successor as prescribed by $f$ at $v_i$. 
  In the superimposed optcost function $f$, a function $g_i:y = -mx+c$ having domain $[u_i,v_i]$ with slope less than $-\prices(l)$ is replaced as follows : alter $g_i$ from  $y = -mx+c$ to 
 $y = -m(x+\delta) + c + \prices(l) * \delta = -mx + c + (\prices(l) - m)*\delta$ for $x \in [u_i,v_i-\delta]$. Let us denote the new function as $h_i$ over the domain $[u_i, v_i-\delta]$. This corresponds to spending $\delta$ time until $x \leq v_i$.
 
 When $x \in [v_i - \delta, v_i]$, then Player~2  spends the time $v_i-x$ at $\lmin$ before proceeding,
  as prescribed by $f$ from $v_i$ onwards. Thus the function obtained by replacing $g_i$ for this range $[v_i - \delta, v_i]$, 
   $h'_i$ is $y = -\prices(\lmin)x + c'$,  and passes  through the point $(v_i,f(v_i))$.  
   However, $h'_i$ should intersect with $h_i$ at $v_i-\delta$ to make the resulting improved optcost function continuous (and thus usable by the predecessors of $\lmin$). We shall show that the line passing through the two points $(v_i-\delta,h_i(v_i - \delta))$ and $(v_i,f(v_i))$  has a slope $-m'=-\prices(l)$. 
 
 We have $g_i$, the original cost function, and $h_i$, 
  and we know that from $v_i$ onwards, Player~2 has to continue with the optcost as dictated by $f$ (the superimposed function). Thus we know that from the point $(v_i - \delta, h_i(v_i-\delta))$  of the new  function $h_i$, the optcost will proceed towards the point $(v_i,f(v_i))$ (recall that $g_i(v_i) = f(v_i)$). Thus given these two points, we find the line $h'_i = -m'x + c'$ as follows. 

 \begin{align*}
  -m' &=~ \frac{f(v_i) -h_i(v_i-\delta)}{v_i-(v_i-\delta)}\\
  ~&= ~\frac{[-m*v_i+c] - [-m(v_i - \delta) + c + (\prices(l) - m)*\delta]}{\delta}\\
   ~&=~\frac{-\prices(l) * \delta}{\delta}\\
   ~&=~ -\prices(l)  
  \end{align*}
 
 Similarly, we also find $c'$ by using $h'_i = -m'x + c'$ where slope is $-m' = -\prices(l)$ and this line passes through the point 
 $(v_i,f(v_i))$. 
 \begin{align*}
  f(v_i) ~&=~ -\prices(l) * v_i + c' \\
    -m* v_i + c ~&=~ -\prices(l) * v_i + c' \\
   c' ~&=~ c + (\prices(l) - m)*v_i
  \end{align*}

\item \textbf{$[\delta,2\delta]-$delay location} : For every function $g_i$ in $f$ (the superimposed function) of Step 2, we first apply the modification of always spending $\delta$ delay at $\lmin$. This is achieved by changing it from $y = -mx+c$ to $y = -m(x+\delta) + c + \prices(l) * \delta$. The domain of $g_i$ also changes from $[u_i,v_i]$ to $[u_i - \delta, v_i - \delta]$. Thus the entire superimposed function $f$ has been modified to $f'$ (lets call it the adjusted superimposed function). 
After this, proceed with $\lmin$ as though it were a $[0,\delta]$-delay location while taking $f'$ to be its adjusted superimposed function. 
\end{enumerate}

\subsection{Complexity and Termination when $\lmin$ is a player 2 location}
 \noindent{\it {\bf Computing Almost Optimal Strategies}}: The strategy corresponding to  
computed optcost when $\lmin$ is a player 2 location is derived as follows. 
\begin{enumerate}
\item $\lmin$ is urgent. 
In this case, we simply do steps 1,2 of the algorithm, superimpose and 
take exterior obtaining the function $f$. 
For $x \in [u_k, v_k]$, the strategy will dictate moving to 
location $l_k$, since $g_k$ is the optcost function 
over the domain $[u_k,v_k]$ of the successor $l_k$ of $\lmin$.

\item $\lmin$ is a $[0, \delta]$-dwell time location.
If $x \in [u_i, v_i-\delta]$ and the function is $h_i$, 
the strategy will prescribe waiting at $\lmin$
for $\delta$ amount of time and then proceed to $l_i$ 
whose cost function is $g_i$, the one replaced by $h_i$. 
If $x \in [u_i, v_i-\delta]$ and the function is $g_i$ (not replaced at Step 3), 
then the strategy suggests going immediately to $l_i$ whose cost function is $g_i$. 
Finally, if $x \in [v_i-\delta, v_i]$ for functions $h'_i$, we prescribe waiting at $\lmin$ till $v_i-x$. 

\item $\lmin$ is a $[\delta,2\delta]$-dwell time location.
The strategy prescribes waiting for $\delta$ time at $\lmin$, and 
then uses the strategy prescribed above for 
$[0, \delta]$-dwell time locations.
\end{enumerate}

We have already discussed the complexity of Algorithm 1 
in computing the optcost function for $\lmin$, and the almost optimal strategies 
when $\lmin$ is a player 1 location. Now we discuss the case when $\lmin$ is a player 2 
location. 

Assume $\lmin$ is a player 2 location. 
Let $\alpha(m,p)$ denote the total number of 
affine segments appearing in cost functions 
across all locations. 
We handle this case by making $l_{min}$ urgent and 
 solve the modified PTG $\ptg'$ (where $l_{min}$ is urgent) 
 which has one location less and then uses the 
 computed optcost functions as outside cost functions to solve for $l_{min}$ itself. 
 This can be  repeated as the optcost computed in $\ptg'$ could get updated when the optcost 
 cost of $l_{min}$ itself is computed. This process  gets repeated as many times as the number of 
 segments we started with i.e $p$. Thus the equation is $\alpha(m,p) \leq p . (1 + \alpha(m-1,p+1))$ where 
 $\alpha(m-1,p+1)$ is the number of segments used for solving $\ptg'$. 
 $1 + \alpha(m-1,p+1)$ is the number of segments used for solving for $l_{min}$ 
 and $p(1+\alpha(m-1,p+1))$  are the repetitions. 
  Solving this, one can easily check that $\alpha(m,p)$ is at most triply exponential 
in the number of locations $m$ of the resetfree component $\frptg$.
Obtaining a bound of the number of affine segments, it is easy to see that Algorithm 1 terminates; the time taken to compute almost optimal strategies and optcost functions is triply exponential.

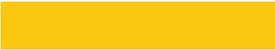
\end{document}